\documentclass[a4paper,10pt]{article}
\PassOptionsToPackage{nodayofweek}{datetime}
\usepackage{amsmath}
\usepackage{amssymb}
\usepackage{amsthm}
\usepackage{enumerate}
\usepackage{slashed}
\usepackage[usenames, dvipsnames]{color}
\usepackage{graphicx}
\usepackage{float}
\usepackage{caption}
\usepackage{subcaption}
\usepackage{pdfpages}
\usepackage{setspace}
\usepackage{geometry}
\usepackage{datetime}
\newtheorem{proposition}{Proposition}[section]
\newtheorem{lemma}[proposition]{Lemma}
\newtheorem{theorem}[proposition]{Theorem}
\newtheorem*{theorem*}{Theorem}
\newtheorem*{proposition*}{Proposition}
\newtheorem*{lemma*}{Lemma}
\newtheorem{corollary}[proposition]{Corollary}
\newtheorem*{corollary*}{Corollary}
\newtheorem*{remark*}{Remark}

\theoremstyle{definition}
\newtheorem{remark}[proposition]{Remark}

\newtheorem{definition}[proposition]{Definition}
\newtheorem*{definition*}{Definition}

\newtheorem{conjecture}[proposition]{Conjecture}

\newtheorem*{acknowledgements*}{Acknowledgements}

\begin{document}
\includepdf{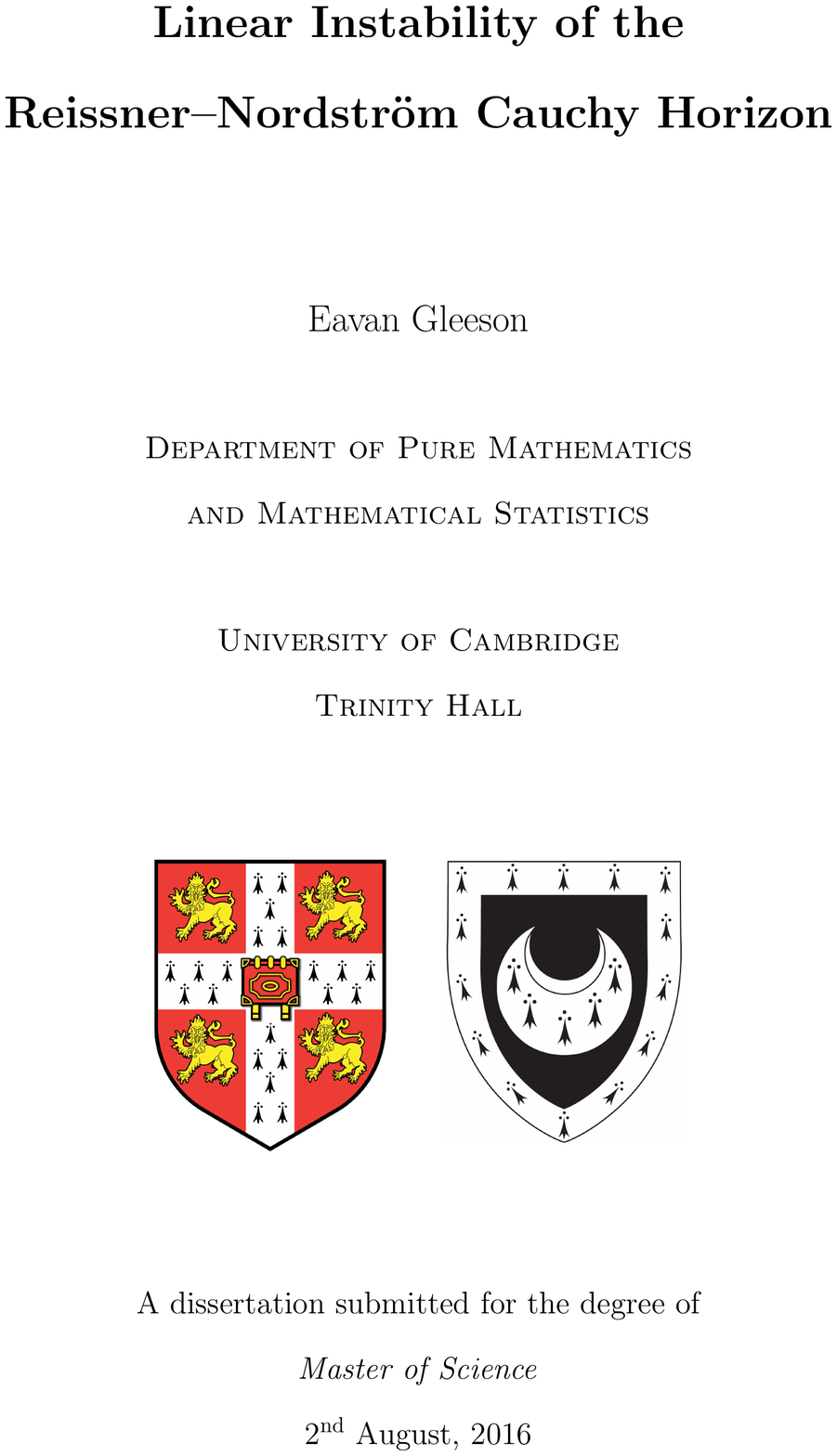}

\newpage
\section*{Acknowledgements}

First and foremost, I would like to thank my supervisor Jonathan Luk for his continuous encouragement,  his limitless enthusiasm for the theory of General Relativity and for innumerable instructive and stimulating discussions over the course of my studies.

I would also like to thank Mihalis Dafermos and Arick Shao for their very helpful comments and insights on the original version of this thesis.

Moreover, I am extremely grateful to the Cambridge Centre for Analysis, the Department of Pure Mathematics and Mathematical Statistics, University of Cambridge, and the Robert Gardiner Memorial Fund  for their financial support.

Finally, I would like to express my gratitude to Daniel Fitzpatrick, to my family and to my friends for their ceaseless support during this time.

\newpage
\section*{Abstract}

This work studies solutions of the scalar wave equation 
\[\Box_g\phi=0\]
on a fixed subextremal Reissner--Nordstr\"{o}m spacetime with non-vanishing charge $q$ and mass $M$. In a recent paper, Luk and Oh established that generic smooth and compactly supported initial data on a Cauchy hypersurface lead to solutions which are singular in the $W^{1,2}_{loc}$ sense near the Cauchy horizon in the black hole interior, and it follows easily that they are also singular in the $W^{1,p}_{loc}$ sense for $p>2$. On the other hand, the work of Franzen shows that such solutions are non-singular near the Cauchy horizon in the $W^{1,1}_{loc}$ sense. Motivated by these results, we investigate the strength of the singularity at the Cauchy horizon.  We identify a sufficient condition on the black hole interior (which holds for the full subextremal parameter range $0<|q|<M$) ensuring  $W^{1,p}_{loc}$ blow up near the Cauchy horizon of solutions arising from generic smooth and compactly supported data for every $1<p<2$. We moreover prove that provided the spacetime parameters satisfy $\frac{2\sqrt e}{e+1}<\frac{|q|}{M}<1$, we in fact have $W^{1,p}_{loc}$ blow up near the Cauchy horizon for such solutions for every $1<p<2$. This shows that the singularity is even stronger than was implied by the work of Luk--Oh for this restricted parameter range.

For the majority of this work, we restrict to the spherically symmetric case, as $W^{1,p}_{loc}$ blow up of a spherically symmetric solution with admissible initial data is sufficient to ensure the generic blow up result. 
Blow up is proved by identifying a condition near null infinity that prevents the solution from belonging to $W^{1,p}_{loc}$ in any neighbourhood of the future Cauchy horizon. This is done by means of a contradiction argument, namely it is shown that regularity of the solution in the black hole interior implies $L^p$-type upper bounds on the solution that contradict lower bounds deduced from the aforementioned condition at null infinity. Establishing the $L^p$-type upper bounds provides the main challenge, and is achieved by first establishing corresponding $L^1$ and $L^2$-type estimates and using the K-method of real interpolation to deduce the $L^p$-type estimates.

\newpage

\tableofcontents
\newpage
\section{Introduction}\label{star9}

In what follows, we study the linear wave equation
\begin{equation}\label{waveequation}
\Box_g\phi=0
\end{equation}
on a subextremal Reissner--Nordstr\"{o}m spacetime $(M,g)$ with non-vanishing charge. As usual, $\Box_g$ denotes the standard covariant wave operator (Laplace--Beltrami operator) associated with the metric $g$. In a local coordinate system, the metric $g$ can be written as
\[ g=-\left(1-\frac{2M}{r}+\frac{q^2}{r^2}\right)dt^2 +\left(1-\frac{2M}{r}+\frac{q^2}{r^2}\right)^{-1}dr^2 +r^2d\sigma_{\mathbb{S}^2},\]
where $d\sigma_{\mathbb{S}^2}$ is the standard metric on the unit $2$-sphere. $q$ and $M$, the charge and mass parameters respectively, are assumed to satisfy
\begin{equation}\label{subextremal}
0<|q|<M,
\end{equation}
namely $(M,g)$ is subextremal with non-vanishing charge. We denote $Q=\frac{|q|}{M}$, so the subextremality assumption \eqref{subextremal} implies
\[0<Q<1.\]

In recent years, much progress has been made in analysing solutions to \eqref{waveequation}, both in the black hole exterior and interior (see for instance \cite{BS}, \cite{BS2}, \cite{DR13}, \cite{DRSR1}, \cite{DRSR2}, \cite{T} and the references therein). In the interior region, both stability and instability results have been obtained (\cite{Franzen}, \cite{H}, \cite{LO} and \cite{LS}). Of particular relevance to this thesis is \cite{LO} where Luk and Oh show that generically $\phi$ is singular in the $W^{1,2}_{loc}$ sense, namely the $L^2$ norm of the derivative of $\phi$ with respect to a regular, Cauchy horizon transversal vector field blows up at the Cauchy horizon. On the other hand, Franzen proved in \cite{Franzen} that the solution is non-singular in the $W^{1,1}_{loc}$ sense, and moreover that the solution is uniformly bounded in the black hole interior, up to and including the Cauchy horizon, to which it can be continuously extended. (See also \cite{H} for a more refined estimate.)

It follows immediately from the $W^{1,2}_{loc}$ blow up result of \cite{LO} that we in fact have generic $W^{1,p}_{loc}$ blow up for every $p>2$. Indeed, given a compact neighbourhood $U$ of any point of the Cauchy horizon, we have $L^p(U) \subseteq L^2(U)$ when $p>2$. From this it follows that if generic $W^{1,p}_{loc}$ blow up of solutions does not hold, then generic $W^{1,2}_{loc}$ blow up also fails, a contradiction with \cite{LO}. 

We thus have $W^{1,1}_{loc}$ stability  and $W^{1,p}_{loc}$ instability of solutions for $p\geq2$, but uncertainty remains over the precise nature of the instability. For physical reasons, it is important to understand and quantify the strength of the singularity. Our first result in this direction is the following \emph{conditional} theorem, which holds for the full subextremal range of parameters with non-vanishing charge, i.e. $0<Q=\frac{|q|}{M}<1$, and identifies a condition on the black hole \emph{interior} which, if violated, ensures generic $W^{1,p}_{loc}$ blow up of solutions for $1<p<2$.

\begin{theorem}[Conditional Theorem, version 1]\label{conditionalinstab}
Let $\Sigma_0$ be a complete $2$-ended asymptotically flat Cauchy hypersurface for the maximal globally hyperbolic development of a subextremal Reissner--Nordstr\"{o}m spacetime with non-vanishing charge. Suppose $1<p<2$. Then generic smooth and compactly supported initial data for \eqref{waveequation} on $\Sigma_0$ give rise to solutions that are not in $W^{1,p}_{loc}$ in a neighbourhood of any point on the future Cauchy horizon $\mathcal{CH}^+$, unless the ``solution map'' in the black hole \emph{interior} is not bounded below in $W^{1,1}$ in an appropriate sense.\footnote{See Section \ref{themaintheorem} and in particular Theorem \ref{condinstab2} and Remark \ref{bdd below} for details.}
\end{theorem}

Furthermore, we show that for a certain subrange of the parameters the condition of the theorem does not hold (i.e. the ``solution map'' is bounded below in $W^{1,1}$ in the black hole interior), and so establish $W^{1,p}_{loc}$ instability in this parameter subrange for $1<p<2$. This instability result is our main result and is stated directly below.

\begin{theorem}[Main Theorem, version 1]\label{mtv1}
Let $\Sigma_0$ be a complete $2$-ended asymptotically flat Cauchy hypersurface for the maximal globally hyperbolic development of a subextremal Reissner--Nordstr\"{o}m spacetime such that $\frac{2\sqrt e}{e+1}<Q<1$, where $e$ is the Euler number (i.e. $0<\log\frac{r_+}{r_-}<1$ where $r_+>r_->0$ are the roots of $r^2-2Mr+q^2=0$). Then for each $1<p<2$, generic smooth and compactly supported initial data for \eqref{waveequation} on $\Sigma_0$ give rise to solutions that are not in $W^{1,p}_{loc}$ in a neighbourhood of any point on the future Cauchy horizon $\mathcal{CH}^+$. 
\end{theorem}

\begin{figure}[H]
\centering
\includegraphics[scale=0.4]{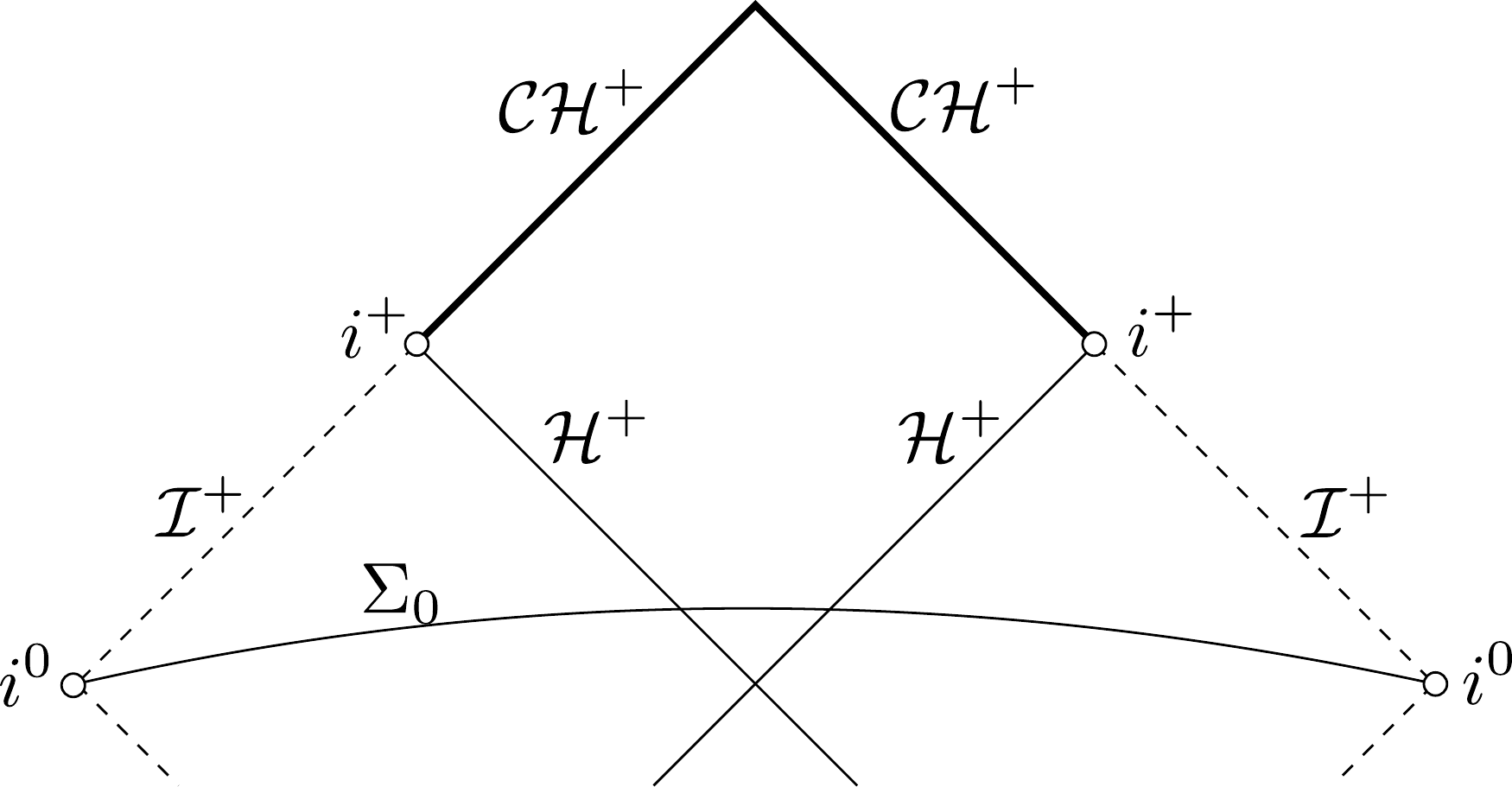}
\end{figure}

\begin{remark*}Recall that if $p\geq2$, $W^{1,p}_{loc}$ instability holds for the whole subextremal range of parameters with non-vanishing charge ($0<Q<1$). Combining this with our main theorem yields $W^{1,p}_{loc}$ instability for every $p>1$ for the parameter range $\frac{2\sqrt{e}}{e+1}<Q<1$. We remark also that Theorem \ref{mtv1} in particular shows the $W^{1,1}_{loc}$ result of Franzen \cite{Franzen} is sharp, at least for the restricted parameter range $\frac{2\sqrt{e}}{e+1}<Q<1$.
\end{remark*}

While for technical reasons arising from the analysis in the black hole interior, we have only managed to prove our instability result, Theorem \ref{mtv1}, for a subrange of the subextremal parameters\footnote{The reason for this parameter restriction is due to the need to  propagate $L^1$-type estimates from a null hypersurface through the black hole interior to the event horizon. See  Theorem \ref{intl1} for details.}, one nonetheless expects that  in fact $W^{1,p}_{loc}$ instability holds for the full subextremal range with non-vanishing charge. Indeed, if one na\"{i}vely extrapolates from results in the \emph{extremal} case $Q=1$ (linear stability of the extremal Cauchy horizon is shown in \cite{G}), then one may even think that the ``far away from extremality'' case is ``more unstable'' than the ``near extremality'' case. Similarly for the cosmological case, heuristic arguments in \cite{Brady} show that we expect the Cauchy horizon of Reissner--Nordstr\"{o}m--de Sitter to be ``more unstable'' ``far away from extremality'' than ``near extremality''. Ironically, however, we succeed in proving instability only ``near extremality'' due to the technicalities in the black hole interior, but remain optimistic that the result can be extended to the full parameter range. We thus make the following conjecture.

\begin{conjecture}
Let $\Sigma_0$ be a complete $2$-ended asymptotically flat Cauchy hypersurface for the maximal globally hyperbolic development of a subextremal Reissner--Nordstr\"{o}m spacetime with non-vanishing charge. Then, for each $1<p<2$, generic smooth and compactly supported initial data on $\Sigma_0$ give rise to solutions that are not in $W^{1,p}_{loc}$ in a neighbourhood of any point on the future Cauchy horizon $\mathcal{CH}^+$. 
\end{conjecture}

\noindent Although this conjecture remains open, we show in Theorem \ref{mtv1} that at the very least it holds for the the subrange $\frac{2\sqrt{e}}{e+1}<Q<1$ of black hole parameters. This contrasts with the expectation in the cosmological case, where a heuristic argument suggests the following conjecture:

\begin{conjecture}
For each non-degenerate Reissner--Nordstr\"{o}m--de Sitter spacetime, there exists $p>1$ such that solutions of the linear wave equation with smooth and compactly supported data on a Cauchy hypersurface are in $W^{1,p}$ near the Cauchy horizon.
\end{conjecture}

\noindent This conjecture is supported by \cite{HV}, in which Hintz--Vasy showed that for all non-degenerate Reissner--Nordstr\"{o}m--de Sitter spacetimes, solutions are in $H^{1/2+\beta}$ near the Cauchy horizon for some $\beta>0$.\footnote{Since it was moreover shown that the solution is smooth in the directions tangental to the Cauchy horizon, one can think that the singularity is one-dimensional and that $H^{1/2}$ scales like $W^{1,1}$ near the Cauchy horizon.} (See also \cite{DR07}, \cite{Dya} and \cite{MSBV}, where Schwarzschild--de Sitter, Kerr and Kerr--de Sitter spacetimes are considered.) Thus the strength of the singularity is different for the asymptotically flat and cosmological cases. This is ultimately due to the differing decay properties of solutions in the black hole exterior regions, namely in Reissner--Nordstr\"{o}m solutions decay polynomially in the exterior whereas in Reissner--Nordstr\"{o}m--de Sitter solutions decay exponentially in the exterior.

As with other instability results, Theorems \ref{conditionalinstab} and \ref{mtv1} are motivated by the strong cosmic censorship conjecture. This conjecture, perhaps the most fundamental open problem in mathematical general relativity, is the subject of a  huge body of literature. We will not discuss the conjecture itself here, but refer the reader to Section 1.4 of \cite{LO} (and the references within) where it is discussed in detail. Here, it suffices to say that, if true, the strong cosmic censorship conjecture would imply that small perturbations of the Reissner--Nordstr\"{o}m spacetime for the \emph{nonlinear} Einstein--Maxwell equations give rise to a singular Cauchy horizon. The linear wave equation which we study is regarded as a ``poor man's linearisation'' of the Einstein--Maxwell equations, and so the generic $W^{1,p}_{loc}$ blow-up result presented in this work is evidence in favour of the strong cosmic censorship conjecture. Our instability result strengthens the $W^{1,2}_{loc}$ result of Luk--Oh, implying that the Cauchy horizon is more singular than was implied by \cite{LO}.

\subsection*{Strategy of Proof}

We now turn our attention  to  the proofs of the two theorems above. As Theorem \ref{conditionalinstab} follows almost immediately from the proof of Theorem \ref{mtv1}, in particular from identifying how the proof of Theorem \ref{mtv1} may fail outside the parameter range $\frac{2\sqrt e}{e+1}<Q<1$, our main challenge is in proving Theorem \ref{mtv1}. 

The proof of Theorem \ref{mtv1} proceeds by showing that if $\frac{2\sqrt e}{e+1}<Q<1$, then for each $1<p<2$ there is a spherically symmetric solution arising from smooth and compactly supported initial data which is singular in the $W^{1,p}_{loc}$ sense, from which it follows that the set of smooth, compactly supported data which correspond to regular solutions is of codimension at least $1$.

In analogy to \cite{LO}, the existence of the spherically symmetric solution of the previous paragraph is shown by identifying a condition near null infinity that prevents a spherically symmetric solution from belonging to $W^{1,p}_{loc}$ in any neighbourhood of the future Cauchy horizon - this is shown using a contradiction argument. Indeed, we show that regularity of a spherically symmetric solution in the black hole interior implies upper bounds on the solution in the exterior region that contradict lower bounds that are implied if the condition at null infinity holds. It then remains only to identify a spherically symmetric solution arising from smooth and compactly supported data on the Cauchy surface $\Sigma_0$ which satisfies the condition near null infinity. Note that while the proof of Theorem \ref{mtv1} is philosophically similar to the proof for the case $p=2$ in \cite{LO}, Theorem \ref{mtv1} is not implied by \cite{LO}. We emphasise that the argument in \cite{LO} uses an $L^2$-type upper bound near the Cauchy horizon as the contradictive assumption and from this deduces a $L^2$-type upper bound in the black hole exterior which contradicts a $L^2$-type  lower bound obtained from the condition near null infinity. We, however, assume a $L^p$-type upper bound as our contradictive assumption, and show that it implies a $L^p$-type upper bound in the black hole exterior which contradicts a $L^p$-type lower bound deduced from the condition near null infinity. This conclusion that the $L^p$-type upper bound is false could not be reached assuming the $L^2$ upper bound assumption of \cite{LO}.

A key ingredient of the contradiction argument is the propagation of $L^p$-type upper bounds from the black hole interior to the exterior. This is achieved by establishing a chain of estimates. We show a weighted $L^p$ term involving certain derivatives of a spherically symmetric solution on a constant $r$-hypersurface in the exterior can be controlled by a corresponding weighted $L^p$ term on the event horizon and a data term. Furthermore, the $L^p$ term on the event horizon can in turn be controlled by a corresponding $L^p$ term on a Cauchy horizon transversal null hypersurface in the interior together with a data term. For $1<p<2$, proving  these various $L^p$ type estimates in a direct manner seemed intractable, so we prove them indirectly using real interpolation. Indeed, we establish analogous estimates for the cases $p=1$ and $p=2$ (the endpoint estimates) and interpolate between these estimates to deduce the family of desired intermediate estimates for $1<p<2$.

It is precisely this method of proof which allows us to conclude the conditional instability result Theorem \ref{conditionalinstab}: we obtain the necessary $L^p$-type estimates precisely when the corresponding $L^1$ and $L^2$-type estimates hold. While we succeed in establishing the $L^2$ estimates and the exterior $L^1$ estimate for all subextremal Reissner--Nordstr\"{o}m spacetimes, the same is not true for the interior $L^1$ estimate. In other words, instability may fail precisely when we are unable to estimate the $L^1$ norm of a suitable derivative of the solution on the event horizon by the $L^1$ norm of the same derivative on a Cauchy horizon transversal null hypersurface. This can be interpreted as a statement about the boundedness from below of the ``solution map'' in a $W^{1,1}$ sense, namely the condition given in Theorem \ref{conditionalinstab}.

We now give an outline of the structure of the remainder of this section.  We begin with a brief exposition on the geometry of subextremal Reissner--Nordstr\"{o}m spacetime in Section \ref{geometry}. In Sections \ref{exterior coordinates} and \ref{interior coordinates}, we introduce the coordinate systems in the black hole exterior and interior with which we work. We then discuss the linear wave equation (Section \ref{we2}), the notations and conventions which we adopt (Section \ref{nandc}) and we describe the class of initial data of interest (Section \ref{inda}). In Section \ref{sobspace} we discuss the space $W^{1,p}_{loc}$ on the interior of the Reissner--Nordstr\"{o}m black hole.  Then, armed with these preliminaries, in Section \ref{themaintheorem} we give precise statements of Theorem \ref{conditionalinstab} and Theorem \ref{mtv1} (see Theorem \ref{condinstab2} and Theorem \ref{mtv2}) and a detailed explanation of the strategy of their proofs in Section \ref{pf of mainthm}. The remainder of the thesis is devoted to proving these two results.

\subsection{The Reissner--Nordstr\"{o}m Solution}\label{geometry}

The Reissner--Nordstr\"{o}m spacetimes are a $2$-parameter family of spacetimes: they represent a charged, non-rotating black hole as an isolated system in an asymptotically flat spacetime, and are indexed by the charge $q$ and the mass $M$ of the black hole. They are the unique static and spherically symmetric solutions of the Einstein--Maxwell system. The Penrose diagram of the maximal globally hyperbolic development of subextremal Reissner--Nordstr\"{o}m with non-vanishing charge (i.e, $0<|q|<M\iff 0<Q<1$) data on a complete Cauchy hypersurface $\Sigma_0$ with two asymptotically flat ends is shown below.

\begin{figure}[h]
\includegraphics[width=8.5cm]{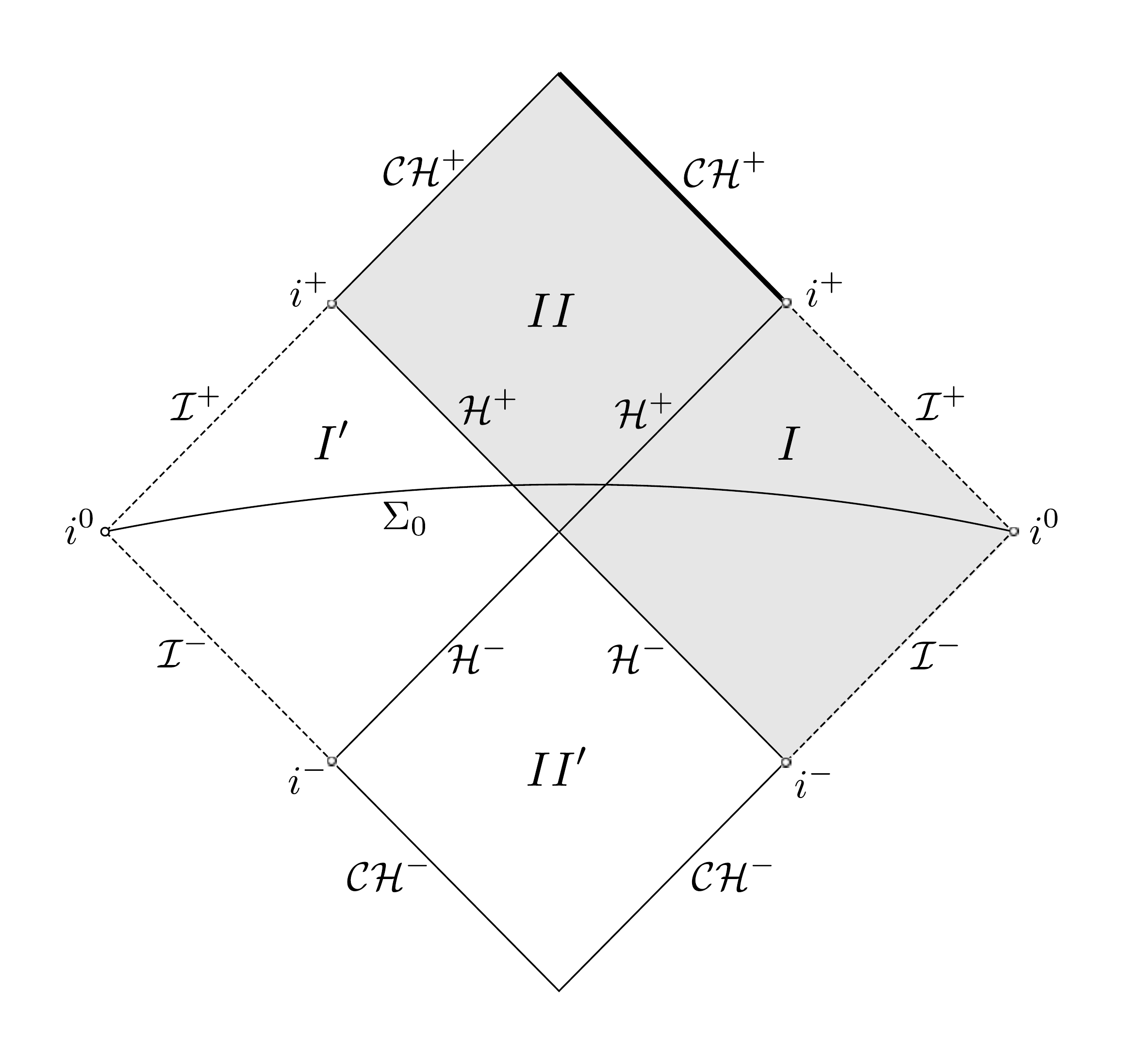}
\centering
\caption{The maximal globally hyperbolic development of Reissner--Nordstr\"{o}m data on a complete two-ended asymptotically flat Cauchy hypersurface $\Sigma_0$.}
\label{MGHD}
\end{figure}

We summarise now the key features of the geometry that will be important for us.
Recall that in a local coordinate chart, the metric $g$ may be written as
\begin{equation}\label{g}
g=-\left(1-\frac{2M}{r}+\frac{q^2}{r^2}\right)dt^2 +\left(1-\frac{2M}{r}+\frac{q^2}{r^2}\right)^{-1}dr^2 +r^2d\sigma_{\mathbb{S}^2}.
\end{equation}
We denote by $r_+$ and $r_-$ the distinct positive roots of the quadratic $r^2-2Mr+q^2$ and assume that $r_+>r_-$. 

The black hole region (region $II$ in Figure \ref{MGHD}) is the complement of the causal past of future null infinity $\mathcal{I}^+$, in other words no signal from the black hole region can reach future null infinity. The black hole region is separated from the exterior regions ($I$ and $I'$ in the figure) by its boundary, the bifurcate null hypersurface referred to as the future event horizon and given by $\mathcal{H}^+=\{r=r_+\}\backslash I^-(\mathcal{I}^+)$. The white hole region (region $II'$), past null infinity $\mathcal{I}^-$ and the past event horizon $\mathcal{H}^-$ are defined by time reversal. The expression for the metric given in \eqref{g} has a coordinate singularity at $r=r_+$ but is valid everywhere else in $\{r_-<r\}$.

 The null hypersurface $\{r=r_-\}$ is a smooth Cauchy horizon. The component in the future of the black hole region is denoted $\mathcal{CH}^+$ and called the future Cauchy horizon. The past Cauchy horizon $\mathcal{CH}^-$ is defined similarly by time reversal. The presence of the smooth Cauchy horizon means that the maximal globally hyperbolic development may be extended smoothly and non-uniquely as a solution of the Einstein--Maxwell equations. The strong cosmic censorship conjecture asserts that this property is non-generic.
 
 Due to the symmetry of the asymptotically flat ends, it will suffice only to consider part  of the maximal globally hyperbolic development, namely the region shaded in Figure \ref{MGHD}. Furthermore, it will be sufficient to consider only the incoming component of $\mathcal{CH}^+$, that is the component on the right of the shaded region shown in bold. Once we have the blow-up result for this component of the Cauchy horizon, the result for the other component follows by analogy.
 
 The shaded region is composed of the black hole interior and an exterior region. We now describe suitable coordinates for these regions.

 \subsubsection{Coordinates for the Black Hole Exterior}\label{exterior coordinates}
 We define null coordinates $u$ and $v$ in the black hole exterior $\{r>r_+\}$ as follows. Set 
 \[\Omega^2=1-\frac{2M}{r}+\frac{q^2}{r^2} = \frac{(r-r_+)(r-r_-)}{r^2} >0,\]
 and define 
 \[r^*= r+(M+\frac{2M^2-q^2}{2\sqrt{M^2-q^2}})\log(r-r_+) +(M-\frac{2M^2-q^2}{2\sqrt{M^2-q^2}})\log(r-r_-), \]
 so that
 \[\frac{dr^*}{dr}=\frac{1}{\Omega^2}.\]
Note $\frac{dr^*}{dr}>0$, so $r^*$ is a strictly increasing function of $r$ in $\{r>r_+\}$.
 Set
 \[v=\frac{1}{2}(t+r^*),\hspace{5mm} u=\frac{1}{2}(t-r^*).\]
 Then
 \[\frac{\partial}{\partial v}=\frac{\partial}{\partial t}+ \frac{\partial}{\partial r^*}, \hspace{5mm} \frac{\partial}{\partial u}=\frac{\partial}{\partial t}- \frac{\partial}{\partial r^*}.\]
 Let $(\theta,\varphi)$ be a spherical coordinate system on $\mathbb{S}^2$ and $d\sigma_{\mathbb{S}^2}=d\theta^2+\sin^2\theta\, d\varphi^2$ the standard metrc on $\mathbb{S}^2$. Then, with respect to the $(u,v,\theta,\varphi)$ coordinates, the Reissner--Nordstr\"{o}m metric is
 \[g=-4\Omega^2dudv+r^2d\sigma_{\mathbb{S}^2}.\]
 Furthermore, we have
 
\begin{equation}\label{ln1}
\lambda:=\partial_vr=\Omega^2, \hspace{5mm} \nu:=\partial_ur=-\Omega^2.
\end{equation}
In this coordinate system,
the limit $\{v=\infty\}$ corresponds to future null infinity $\mathcal{I}^+$, while $\{u=\infty\}$ corresponds to the future event horizon $\mathcal{H}^+$. This is shown in Figure \ref{extcoords} below.

\begin{figure}[H]
\centering
\includegraphics[scale=1]{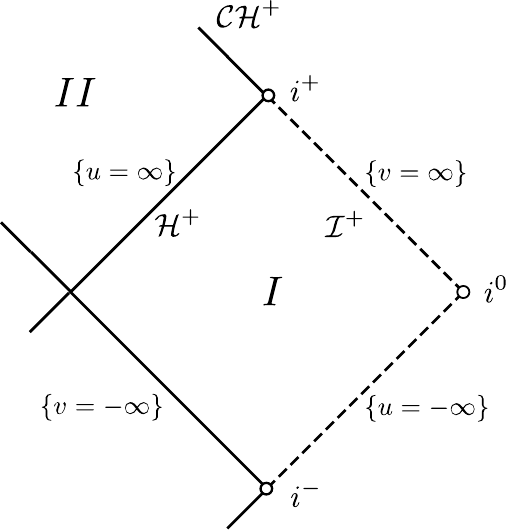}
\caption{Coordinates in the black hole exterior.}
\label{extcoords}
\end{figure}

 \subsubsection{Coordinates for the Black Hole Interior}\label{interior coordinates}
 We define null coordinates $u$ and $v$ in the black hole interior $\{r_-<r<r_+\}$ as follows. Set 
 \[\Omega^2=-(1-\frac{2M}{r}+\frac{q^2}{r^2})= \frac{(r_+-r)(r-r_-)}{r^2}>0,\]
 and define 
 \[r^*= r+(M+\frac{2M^2-q^2}{2\sqrt{M^2-q^2}})\log(r_+-r) +(M-\frac{2M^2-q^2}{2\sqrt{M^2-q^2}})\log(r-r_-), \]
 so that
 \[\frac{dr^*}{dr}=-\frac{1}{\Omega^2}.\]
Thus $\frac{dr^*}{dr}<0$, so $r^*$ is a strictly decreasing function of $r$ in $\{r_-<r<r_+\}$.
 Set
 \[v=\frac{1}{2}(r^*+t),\hspace{5mm} u=\frac{1}{2}(r^*-t).\]
 Then,
 \[\frac{\partial}{\partial v}=\frac{\partial}{\partial r^*}+ \frac{\partial}{\partial t}, \hspace{5mm} \frac{\partial}{\partial u}=\frac{\partial}{\partial r^*}- \frac{\partial}{\partial t}.\]
As before, we let $(\theta,\varphi)$ be a spherical coordinate system on $\mathbb{S}^2$ and $d\sigma_{\mathbb{S}^2}=d\theta^2+\sin^2\theta\, d\varphi^2$ the standard metrc on $\mathbb{S}^2$. Then, with respect to the $(u,v,\theta,\varphi)$ coordinates, the Reissner--Nordstr\"{o}m metric is
 \[g=-4\Omega^2dudv+r^2d\sigma_{\mathbb{S}^2}.\]
 Furthermore, in this region we have
\begin{equation}\label{ln2}
\lambda:=\partial_vr=-\Omega^2, \hspace{5mm} \nu:=\partial_ur=-\Omega^2.
\end{equation}
In this coordinate system,
the limit $\{v=\infty\}$ corresponds to the incoming part of the future Cauchy horizon $\mathcal{CH}^+$, while $\{u=-\infty\}$ corresponds to the future event horizon $\mathcal{H}^+$, as shown in Figure \ref{intcoords} below.

\begin{figure}[h]
\centering
\includegraphics[scale=1]{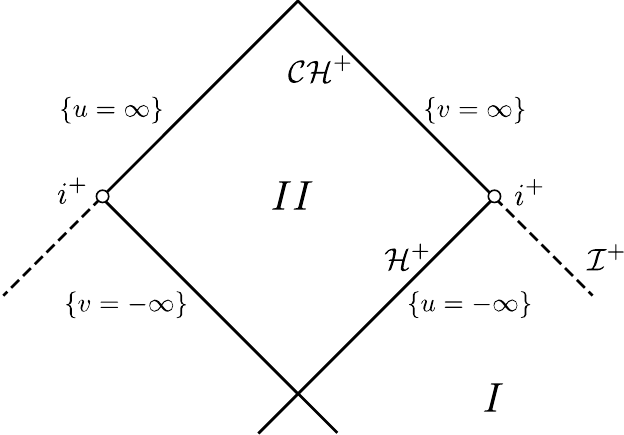}
\caption{Coordinates in the black hole interior.}
\label{intcoords}
\end{figure}

\subsubsection{The Wave Equation}\label{we2}

In both the interior and exterior coordinate systems, the wave equation \eqref{waveequation} takes the form
\[\partial_u\partial_v\phi=-\frac{\partial_vr\partial_u\phi}{r}-\frac{\partial_ur\partial_v\phi}{r}+\frac{\Omega^2\slashed\Delta\phi}{r^2},\]
where $\slashed\Delta$ is the Laplace--Beltrami operator on the standard unit $2$-sphere. In spherical symmetry, the equation takes the form
\begin{equation}\label{extwe}
\partial_u\partial_v\phi=-\frac{\partial_vr\partial_u\phi}{r}-\frac{\partial_ur\partial_v\phi}{r}=-\frac{\lambda\partial_u\phi}{r}-\frac{\nu\partial_v\phi}{r}.
\end{equation}
 So by \eqref{ln1} and \eqref{ln2}, the wave equation in spherical symmetry takes the form
\begin{equation}\label{exteriorwave}
\partial_u\partial_v\phi=\frac{\Omega^2}{r}(\partial_v\phi-\partial_u\phi)
\end{equation}
in the exterior region and
\begin{equation}\label{interiorwave}
\partial_u\partial_v\phi=\frac{\Omega^2}{r}(\partial_v\phi+\partial_u\phi)
\end{equation}
in the interior region.

\subsubsection{Notation and Conventions}\label{nandc}
We adopt the notation and conventions used in \cite{LO}. We recall these here for the reader's convenience.
\begin{itemize}
\item
$C_u$ and $\underline{C}_v$ will denote constant $u$ and $v$ hypersurfaces respectively. When considering a constant $v$ hypersurface that crosses the event horizon, we will denote the interior and exterior components of the hypersurface by $\underline{C}_v^{int}$ and $\underline{C}_v^{ext}$. Similarly, to eliminate ambiguity over whether a constant $u$-hypersurface is in the black hole interior or exterior, we will sometimes write $C_u^{int}$ and $C_u^{ext}$. We omit these superscripts when it is clear from context.
\item
On the constant $u$-hypersurfaces $C_u$, integration is always with respect to the measure $dv$, and on constant $v$-hypersurfaces $\underline{C}_v$ integration is always with respect to the measure $du$. 
\item
Constant $r$-hypersurfaces will be denoted $\gamma_R=\{r=R\}$, and similarly (with slight abuse of notation), constant $r^*$-hypersurfaces will be denoted $\gamma_{R^*}=\{r^*=R^*\}$, where $R^*=r^*(R)$.
\item
Unless otherwise stated, constant $r$-hypersurfaces $\gamma_r$ and constant $r^*$-hypersufaces $\gamma_{r^*}$ are parameterised by the $v$ coordinate and integration is respect to the measure $dv$.
\item
In spacetime regions, we integrate with respect to $du\,dv$. We emphasise that this is \emph{not} integration with respect to the volume form induced from the metric.
\item
We shall use the notation $v_{R^*}(u)$ to denote the unique value of $v$ such that $r^*(u,v)=R^*$. So $r^*(u,v_{R^*}(u))=R^*$. We define $u_{R^*}(v)$ similarly. Again with slight abuse of notaion, we shall sometimes write $v_R(u)$ and $u_R(v)$ instead of $v_{R^*}(u)$ and $u_{R^*}(v)$ respectively.
\end{itemize}

\subsubsection{Initial Data}\label{inda}

We now describe the initial data of interest to us.\footnote{Note that here the term initial data does not refer to data prescribed on the Cauchy hypersurface $\Sigma_0$. Instead we are prescribing initial data for solutions defined in the shaded region of Figure \ref{initialsurface}, as for most of this thesis these solutions are the solutions with which we work. In Theorem \ref{mtv2} we show that one particular solution $\phi_{sing}$ with initial data as prescribed in this section can be used to construct a solution on the entire spacetime with smooth and compactly supported data on $\Sigma_0$.} The data shall be prescribed on two transversal null hypersurfaces, $C^{ext}_{-U_0}$ for some $U_0>0$ large enough and $\underline{C}_1$. In fact, we only prescribe the data on a portion of $\underline{C}_1$, namely $(\underline{C}_1^{ext}\cap\{u\geq -U_0\})\cup(\{\underline{C}_1^{int}\cap\{u\leq -1\})$.  So we prescribe data on the surface
\[S:=(\underline{C}_1^{ext}\cap\{u\geq -U_0\})\cup(\{\underline{C}_1^{int}\cap\{u\leq -1\})\cup(C^{ext}_{-U_0}\cap\{v\geq 1\}).\] These surfaces are shown in bold in Figure \ref{initialsurface} below.

\begin{figure}[h]
\centering
\includegraphics[width=0.5\textwidth]{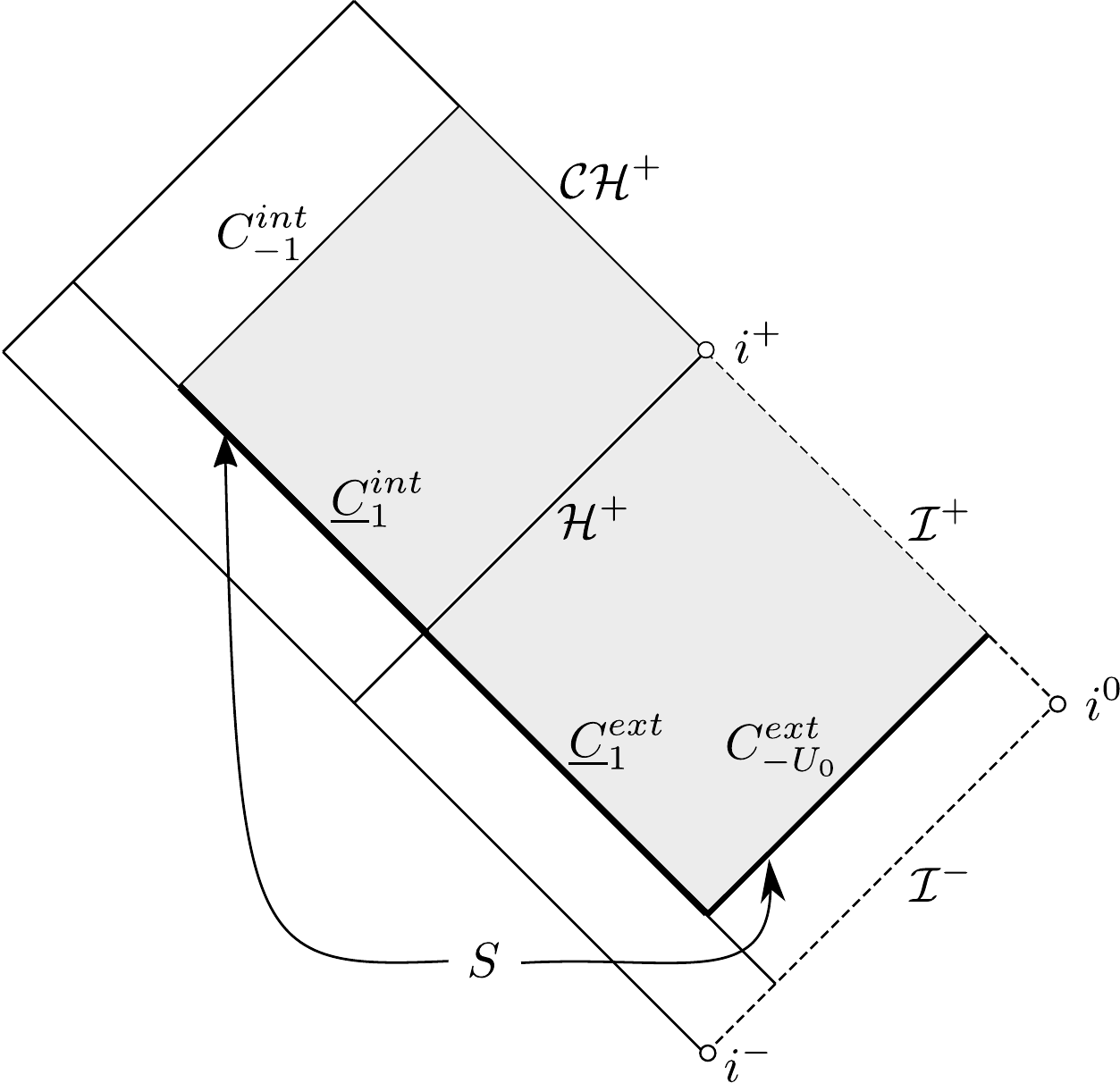}
\caption{The initial data are prescribed on the dark surface $S$. }
\label{initialsurface}
\end{figure}

We shall be concerned with smooth solutions $\phi$ of the wave equation such that there exists some constant $D>0$ such that
\begin{equation}\label{ic1}
\sup_{C_{-U_0}\cap\{v\geq 1\}} r^2\big|\frac{\partial_v\phi}{\partial_vr}\big|\leq D, \hspace{5mm} \sup_{C_{-U_0}\cap\{v\geq 1\}} r^3\big|\frac{\partial_v(r\phi)}{\partial_vr}\big|\leq D
\end{equation}
and
\begin{equation}\label{ic2}
\sup_{(\underline{C}_1^{ext}\cap\{u\geq -U_0\})\cup (\underline{C}_1^{int}\cap\{u\leq -1\})} \big |\frac{\partial_u\phi}{\partial_ur}\big|\leq D, \hspace{2mm} \sup_{(\underline{C}_1^{ext}\cap\{u\geq -U_0\})\cup (\underline{C}_1^{int}\cap\{u\leq -1\})} \big |\frac{\partial_u(r\phi)}{\partial_ur}\big|\leq D.
\end{equation}

In particular, we will be interested in solutions $\phi$ for which
\begin{equation}\label{limit}
\lim_{v\to\infty}r^3\partial_v(r\phi)(-U_0,v)\mbox{ exists}.
\end{equation}
For $\phi$ with smooth, spherically symmetric data on $S$ satisfying \eqref{ic1}, \eqref{ic2} and \eqref{limit}, we set
\[\mathcal{L}=\mathcal{L}_\phi:=\lim_{v\to\infty}r^3\partial_v(r\phi)(-U_0,v)-\int_{-U_0}^\infty2M\Phi(u)\,du,\]
where $\Phi(u)=\lim_{v\to\infty}r\phi(u,v)$. Note that it follows from the results in \cite{DR1} that this limit exists.

\subsubsection{The $W^{1,p}_{loc}$ Space on the Black Hole Interior}\label{sobspace}
For $p>1$ we now formally define the $W^{1,p}_{loc}$ space on the Reissner--Nordstr\"{o}m black hole, to which we have previously referred. We consider the region $\{r_-\leq r<r_+\}$ as the manifold with boundary
\[\{\mbox{Interior of Reissner--Nordstr\"{o}m black hole}\}\cup\mathcal{CH}^+,\]
and recall first the definition of the space $W^{1,p}_{loc}(\mathcal{M})$ for a general $n$-dimensional manifold with boundary $\mathcal{M}$. 

Indeed, let $\partial\mathcal{M}$ denote the boundary of an $n$-dimensional manifold with boundary $\mathcal{M}$ and suppose $(U_\alpha,\varphi_\alpha)$ is an atlas for $\mathcal{M}$. Given a function $f:\mathcal{M}\to\mathbb{R}$ and an open set $\mathcal{U}\subseteq U_\alpha$ for some $\alpha$, the $W^{1,p}$ norm of $f$ on $\mathcal{U}$, $\|f\|_{W^{1,p}(\mathcal{U})}$, is defined as the sum of the $L^p$ norms of $f\circ\varphi_\alpha^{-1}: \mathbb{R}^n\to \mathbb{R}$ (or $f\circ\varphi_\alpha^{-1}: \mathbb{R}^{n-1}\times \mathbb{R}_+\to \mathbb{R}$ if $U_\alpha\cap\partial\mathcal{M}\not=\emptyset$) and its coordinate derivatives on $\varphi_\alpha(\mathcal{U})$. The space $W^{1,p}_{loc}(\mathcal{M})$ is defined to be the space of functions $f:\mathcal{M}\to\mathbb{R}$ such that the $W^{1,p}$ norm of $f$ on $\mathcal{U}$ is finite for every open $\mathcal{U}$ with compact closure such that $\mathcal{U}\subset \overline{\mathcal{U}}\subset U_\alpha$. In particular, a smooth change of coordinate system results in an equivalent space.

We now take $\mathcal{M}$ to be $\{r_-<r<r_+\} \cup\mathcal{CH}^+$. In order to consider $\mathcal{M}$ as a manifold with boundary $\partial\mathcal{M}=\mathcal{CH}^+$, we must find coordinates which are regular at the Cauchy horizon. Following \cite{LO}, we specify such coordinates below.  Indeed, for $v$ sufficiently large, define $V=V(v)$ to be the solution of the equation $\frac{dV}{dv}=e^{-2|\kappa_-|v}$ such that $V\to1$ as $v\to\infty$ (where $\kappa_-:=\frac{r_--r_+}{2r_-^2}<0$ is the surface gravity of the Cauchy horizon). Also, for $u$ sufficiently large, define $U=U(u)$ to be the solution of $\frac{dU}{du}=e^{-2|\kappa_-|u}$ with $U\to1$ as $u\to\infty$. The Cauchy horizon is then given by
\[\mathcal{CH}^+=\{U=1\}\cup\{V=1\},\]
and we may attach this boundary to the interior of the black hole to form a manifold with boundary. One can easily verify that the metric can be smoothly extended to the boundary $\mathcal{CH}^+$.

Let $\mathcal{U}$ be a small neighbourhood of some point on the incoming part $\{V=1\}$ of the boundary $\mathcal{CH}^+$ of $\mathcal{M}$ such that $\mathcal{U}$ has compact closure. The $W^{1,p}$ norm of a smooth, spherically symmetric\footnote{It is sufficient for us to consider explicitly the $W^{1,p}$ norms of spherically symmetric solutions, as to show generic blow up of solutions, it is enough to find one solution which blows up and the solution which we specify is a spherically symmetric.} solution $\phi$ to \eqref{waveequation} is equivalent to
\begin{equation}\label{first1p}
\left(\int_{\mathcal{U}}\left(|\partial_V\phi|^p+|\partial_u\phi|^p+|\phi|^p\right)(u,V)\,du\,dV\right)^{1/p}
\end{equation}
in the $(u,V)$ coordinates. However, since  $e^{-2|\kappa_-|v}\sim e^{-2|\kappa_-|r^*}\sim\Omega^2$ in $\mathcal{U}$ (because $\overline{\mathcal{U}}$ is compact), changing to $(u,v)$ coordinates we have $\operatorname{det}J=e^{-2|\kappa_-|v}\sim\Omega^2$ and $\partial_V\phi=e^{2|\kappa_-|v}\partial_v\phi\sim\Omega^{-2}\partial_v\phi$ in $\mathcal{U}$. Thus,  in $(u,v)$ coordinates \eqref{first1p} is equivalent to the expression
\begin{equation}\label{second1p}
\left( \int_{\mathcal{U}}  \left(\frac{1}{\Omega^{2p-2}}|\partial_v\phi|^p +\Omega^2(|\partial_u\phi|^p+|\phi|^p)\right)(u,v)\,du\,dv\right)^{1/p}.
\end{equation}
Thus, for the smooth, spherically symmetric solution to satisfy $\phi\in W^{1,p}_{loc}(\mathcal{M})$, \eqref{second1p} must be finite for every such $\mathcal{U}$ (as well as a similar statement for the outgoing portion of the Cauchy horizon). On the other hand, to show a smooth, spherically symmetric solution $\phi\not\in W^{1,p}_{loc}(\mathcal{M})$, it is sufficient to show that 
\begin{equation}\label{amountsto}
\int_{C_u\cap\{v\geq1\}} \left(\frac{1}{\Omega}\right)^{2p-2}|\partial_v\phi|^p\,dv
\end{equation}
blows up for all $u$ in some subset of $(-\infty,\infty)$ with positive measure. However, we actually prove a stronger result, namely that \eqref{amountsto} blows up for every $u\in(-\infty,\infty)$ for smooth, spherically symmetric solutions $\phi$ satisfying \eqref{ic1}, \eqref{ic2}, \eqref{limit} and $\mathcal{L}_\phi\not=0$, provided $\frac{2\sqrt e}{e+1}<Q<1$.

\begin{remark}\label{remark}
For $p>1$, $2p-2>0$ and it is easy to see (using l'H\^{o}pital's rule) that 
\[\lim_{x\to\infty}\frac{(\log x)^{a_0}}{x^{2p-2}}=0\]
for any $a_0>0$, or in other words $x^{2p-2}$ grows faster that $(\log x)^{a_0}$. But in the black hole interior, $\Omega^2(u,v)\to0$ as $v\to\infty$ for any $u\in(-\infty,\infty)$, so $1/\Omega(u,v)\to \infty$ as $v\to\infty$ (where $\Omega:=\sqrt{\Omega^2}$). It follows that if 
\[\int_{C_u\cap\{v\geq1\}} \log^{a_0}\left(\frac{1}{\Omega}\right)|\partial_v\phi|^p\,dv\]
blows up then so does \eqref{amountsto}. So, in order to show a smooth, spherically symmetric solution satisfies $\phi\not\in W^{1,p}_{loc}(\mathcal{M})$, it is sufficient to show that for some $a_0>0$
\[\int_{C_u\cap\{v\geq1\}} \log^{a_0}\left(\frac{1}{\Omega}\right)|\partial_v\phi|^p\,dv=\infty \]
for all $u$ in some positive measure subset of $(-\infty,\infty)$. We in fact show that for $\frac{2\sqrt e}{e+1}<Q<1$, this statement holds true for every $u\in(-\infty,\infty)$ for smooth, spherically symmetric solutions $\phi$ satisfying \eqref{ic1}, \eqref{ic2}, \eqref{limit} and $\mathcal{L}_\phi\not=0$, so that $\phi\not\in W^{1,p}$ in any neighbourhood of the incoming future Cauchy horizon (or, abusing notation slightly, $\phi\not\in W^{1,p}_{loc}$ in any neighbourhood of the incoming future Cauchy horizon).
\end{remark}

\subsection{The Main Theorems}\label{themaintheorem}
In this section, we give precise statements of our two key results, the conditional instability result (Theorem \ref{conditionalinstab}) and the instability result for the ``near extremal'' subrange of subextremal black hole parameters (Theorem \ref{mtv1}). We begin with the conditional theorem.

\begin{theorem}[Conditional Theorem, version 2]\label{condinstab2}
Let $\Sigma_0$ be a complete $2$-ended asymptotically flat Cauchy hypersurface for a subextremal Reissner--Nordstr\"{o}m spacetime with non-vanishing charge. Suppose $1<p<2$. Then the set of smooth and compactly supported initial data on $\Sigma_0$ giving rise to solutions in $W^{1,p}_{loc}$ near the future Cauchy horizon $\mathcal{CH}^+$ has codimension at least $1$, \emph{unless} there is a sequence of smooth, spherically symmetric solutions $\phi_n$ of \eqref{waveequation} such that in the black hole interior
\begin{enumerate}
\item
$\partial_u\phi_n=0$ on $\underline{C}_1\cap\{u\leq -1\}$,
\item
$\|\partial_v\phi_n\|_{L^1(\mathcal{H}^+\cap\{v\geq1\})}=\int_{\mathcal{H}^+\cap\{v\geq1\}}|\partial_v\phi_n|=1$ and
\item
$\|\partial_v\phi_n\|_{L^1(C_{-1}\cap\{v\geq1\})}=\int_{C_{-1}\cap\{v\geq1\}}|\partial_v\phi_n|\searrow 0$.
\end{enumerate}
\begin{figure}[H]
\centering
\includegraphics[scale=0.5]{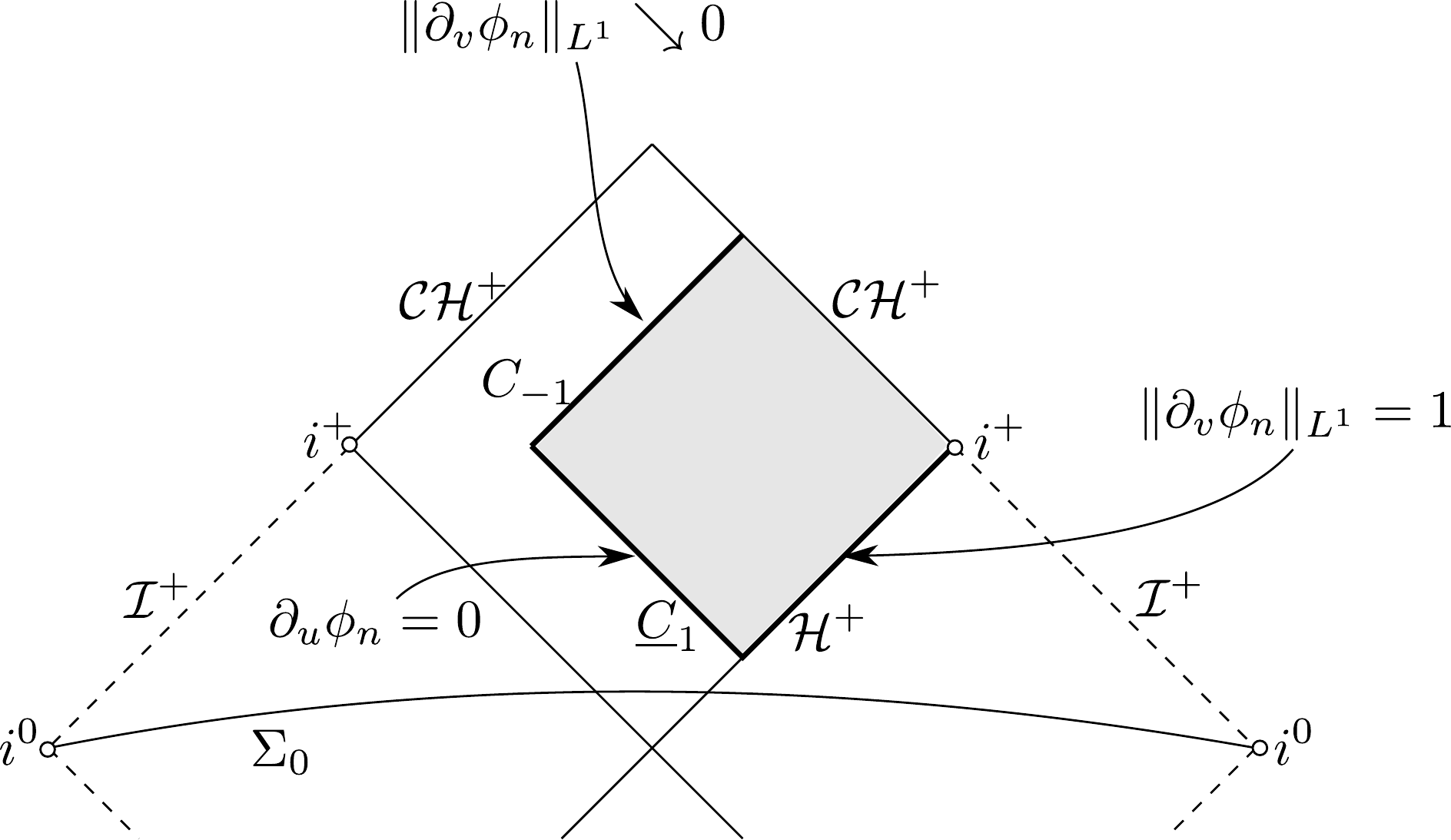}
\end{figure}
\end{theorem}

\begin{remark}\label{bdd below}
The condition in Theorem \ref{condinstab2} can be viewed as a statement about the boundedness of the solution map. Indeed, given smooth, spherically symmetric  data $(\partial_v\phi\vert_{C_{-1}\cap\{v\geq1\}},\partial_u\phi\vert_{\underline{C}_1\cap\{u\leq -1\}})$, then $\partial_v\phi\vert_{\mathcal{H}^+\cap\{v\geq1\}}$ is uniquely determined. Define
\[
T: f \mapsto Tf
\]
where $Tf:=\partial_v\phi\vert_{\mathcal{H}^+\cap\{v\geq1\}}$ for the smooth, spherically symmetric solution $\phi$ to the wave equation \eqref{waveequation} such that $\partial_v\phi\vert_{C_{-1}\cap\{v\geq1\}}=f$ and $\partial_u\phi\vert_{\underline{C}_1\cap\{u\leq -1\}}=0$. Then $T$ is a well-defined linear operator and  Theorem \ref{condinstab2} says that generically solutions blow up in $W^{1,p}_{loc}$ unless the operator $T$ is unbounded in the $L^1$ sense, or in other words unless the solution map (the inverse of $T$) is not bounded below.
\end{remark}

We now give a precise statement of our main theorem, the instability result ``near extremality''.

\begin{theorem}[Main Theorem, version 2]\label{mtv2}
Let $\Sigma_0$ be a complete $2$-ended asymptotically flat Cauchy hypersurface for a subextremal Reissner--Nordstr\"{o}m spacetime such that $\frac{2\sqrt e}{e+1}<Q<1$.  For $p>1$, the set of smooth and compactly supported initial data on $\Sigma_0$ giving rise to solutions in $W^{1,p}_{loc}$ near the future Cauchy horizon $\mathcal{CH}^+$ has codimension at least $1$. 
\end{theorem}

We now discuss the proof of Theorem \ref{mtv2}. We defer discussing the proof of Theorem \ref{condinstab2} to the end of this Section (as it relies on the proof of Theorm \ref{mtv2}).

\subsubsection*{Proof of Theorem \ref{mtv2}}

The bulk of the work in proving Theorem \ref{mtv2} goes into proving the following theorem, which roughly states that for $1<p<2$ and $\frac{2\sqrt e}{e+1}<Q<1$, a smooth, spherically symmetric solution $\phi$ of the wave equation is singular in the $W^{1,p}_{loc}$ sense, provided the condition at null infinity $\mathcal{L}_\phi\not=0$ is satisfied.

\begin{theorem}\label{mainthm}
Assume $\frac{2\sqrt e}{e+1}<Q<1$. Suppose $1<p<2$ and let $\phi$ be a solution to the wave equation \eqref{waveequation} with smooth, spherically symmetic data on $S$ satisfying \eqref{ic1}, \eqref{ic2} and \eqref{limit} and such that $\mathcal{L}_\phi\not=0$. Then, in the black hole interior, $\phi$ satisfies
\[
\int_1^\infty\log^{\alpha(p)}\left(\frac{1}{\Omega}\right)|\partial_v\phi|^p(u,v)\,dv=+\infty
\]
for every $u\in (-\infty,\infty)$ and every $\alpha(p)\geq 4p+4$. 
\end{theorem}

In order to use this result, we rely on a theorem of Luk--Oh, proved in \cite{LO}, which asserts that there is in fact a smooth, spherically symmetric solution satisfying $\mathcal{L}_\phi\not=0$.

\begin{theorem}[Luk--Oh \cite{LO}]\label{LukOhexist}
For $U_0>0$ sufficiently large, there exists a spherically symmetric solution $\phi_{sing}$ to \eqref{waveequation} (defined on the domain of dependence of $\underline{C}_1\cup C_{-U_0}$) with smooth and compactly supported initial data on $\underline{C}_1$ and zero data on $C_{-U_0}$ such that
\[\mathcal{L}_{\phi_{sing}}\not=0.\]
In fact, the support of the initial data $\phi_{sing}\vert_{\underline{C}_1}$ is contained in $\underline{C}_1^{ext}\cap\{-U_0\leq u\leq -U_0+1\}$.
\end{theorem}

Combining Theorems \ref{mainthm} and \ref{LukOhexist} and Remark \ref{remark}, we arrive at the conclusion that there are smooth, spherically symmetric solutions to \eqref{waveequation} on the domain of dependence of $\underline{C}_1\cup C_{-U_0}$ which are not in $W^{1,p}$ in a neighbourhood of any point of  the ``incoming'' future Cauchy horizon, and hence are not in $W^{1,p}_{loc}(\mathcal{M})$ (for $1<p<2$ and $\frac{2\sqrt e}{e+1}<Q<1$). However, we can actually use them to get a stronger result, namely our instability result, Theorem \ref{mtv2}, which we prove below.

\begin{proof}[Proof of Theorem \ref{mtv2}.]
For $p=2$, this is the main result of \cite{LO}, while for $p>2$, this follows from the result for $p=2$ and the fact that $L^p(U)\subseteq L^2(U)$ for every compact set $U$.  We emphasise that the restriction $\frac{2\sqrt e}{e+1}<Q$ is not necessary for the case $p\geq 2$.

We thus assume $1<p<2$. Our goal is to show that the quotient of the space of smooth and  compactly supported initial data on $\Sigma_0$ by the space of  smooth and compactly supported initial data on $\Sigma_0$ leading to solutions in $W^{1,p}_{loc}$ near $\mathcal{CH}^+$ has dimension at least $1$, or equivalently, that the quotient space has a non-trivial element. Thus it suffices to show that there exist smooth and compactly supported initial data on $\Sigma_0$ leading to a solution with infinite $W^{1,p}$ norm on a neighbourhood of some point of $\mathcal{CH}^+$ (with compact closure). We show this by specifying a solution $\phi:M\to\mathbb{R}$ such that the initial data $(\phi,n_{\Sigma_0}\phi)\vert_{\Sigma_0}$ is smooth and compactly supported and such that $\|\phi\|_{W^{1,p}}$ blows up on a neighbourhood of \emph{every} point on $\mathcal{CH}^+$. 

We use Theorem \ref{LukOhexist} to construct this solution. Indeed, as in \cite{LO}, we first show that for sufficiently large $U_0$, there is a smooth solution $\phi_0$ of the wave equation \eqref{waveequation} in the whole spacetime $M$ with smooth and compactly supported data on $\Sigma_0$ such that $\phi_0=\phi_{sing}$ in the exterior region restricted to the future of $\Sigma_0$, which we denote by $\mathcal{F}$. Here $\phi_{sing}$ is the solution from Theorem \ref{LukOhexist}. (This is shown in the proof of Corollary 1.6 of \cite{LO}, but we repeat the argument here for completeness.) By finite speed of propagation, it's enough to prove this property for a particular Cauchy hypersurface. It is convenient to consider $\Sigma_0$ spherically symmetric and asymptotic to the $\{t=0\}$ hypersurface near each end.  We assume furthermore that $\Sigma_0$ intersects $\underline{C}_2$ (and hence also $\underline{C}_1$) in the black hole interior. We choose $U_0$ large enough so that the segment $\underline{C}_1^{ext}\cap \{-U_0\leq u \leq -U_0+1\}$ lies in the past of $\Sigma_0$. Such a hypersurface $\Sigma_0$ is illustrated in Figure \ref{sig0setup} below. Let $\chi:\mathbb{R}\to\mathbb{R}$ be a smooth, positive cutoff function such that
\begin{equation*}
\chi(x)=\left\{\begin{array}{ll} 1&\mbox{if }x\geq 2,\\0&\mbox{if }x\leq1.\end{array}\right.
\end{equation*}
  We denote by $(f,g)$ the data on $\Sigma_0$ defined by
\[(f,g)(p)=\left\{\begin{array}{ll} (\phi_{sing},n_{\Sigma_0}\phi_{sing})(p) & \mbox{if } p \mbox{ is in the domain of dependence of }\underline{C}_2\cup C_{-U_0}\\
\chi(v(p))\cdot(\phi_{sing},n_{\Sigma_0}\phi_{sing})(p)& \mbox{if } 1\leq v(p)\leq 2\\ (0,0) & \mbox{otherwise}\end{array}\right.\]
for $p\in\Sigma_0$. Note that $v(p)$ denotes the $v$-value of the point $p$.
The Cauchy hypersurface $\Sigma_0$ necessarily exits the domain of dependence of $\underline{C}_1\cup C_{-U_0}$, and it follows that $(f,g)$ is compactly supported on $\Sigma_0$. Moreover, $(f,g)$ is smooth.\footnote{The cutoff function $\chi$ is introduced so as to ensure smoothness of $g$ at $v=1$, as $n_{\Sigma_0}\phi_{sing}$ need not vanish there. Strictly speaking we should also use a cutoff  to ensure smoothness of $g$ at $u=-U_0$. However, this is not actually necessary due to the construction of $\phi_{sing}$ (see Section 5 of \cite{LO}). It is easily seen (by choosing the initial data for $\phi_{sing}$ appropriately) that $\phi_{sing}$ may be chosen such that $\phi_{sing}\equiv 0$ on $[-U_0,-U_0+\varepsilon]\times \{v\geq1\}$, so that $g$ is indeed smooth at $u=-U_0$.}

\begin{figure}[h]
\centering
\includegraphics[scale=0.6]{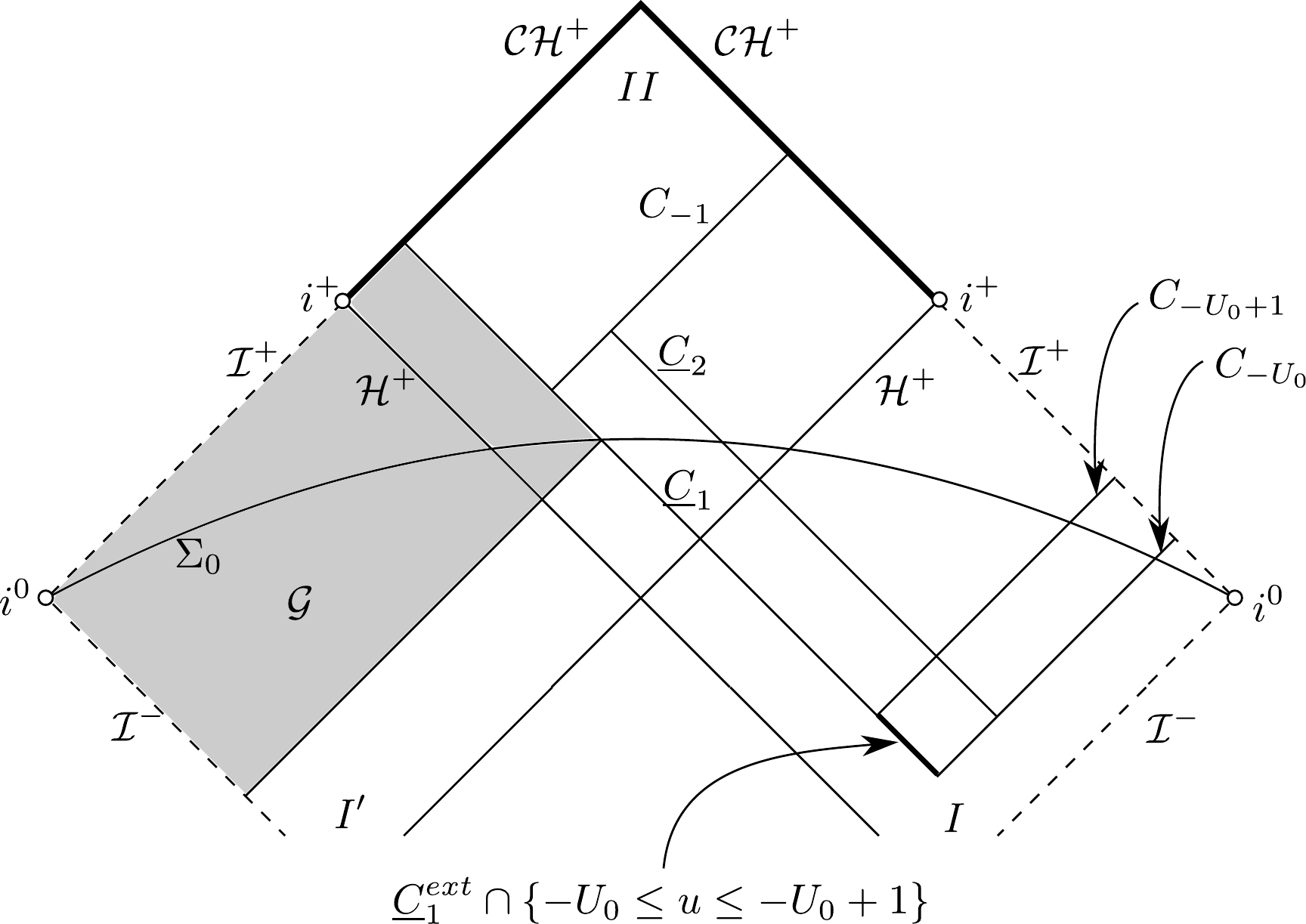}
\caption{The hypersurface $\Sigma_0$.}
\label{sig0setup}
\end{figure}

Let $\phi_0:M\to \mathbb{R}$ be the solution to the linear wave equation \eqref{waveequation} with data $(\phi_0,n_{\Sigma_0}\phi_0)=(f,g)$ on $\Sigma_0$. By finite speed of propagation together with the support proprties of the initial data for $\phi_{sing}$ on $\underline{C}_1\cup C_{-U_0}$, we see that $\phi_0$ agrees with $\phi_{sing}$ in $\mathcal{F}$, and so $\phi_0$ is the desired solution. Moreover, $\phi_0$ is spherically symmetric and $\phi_0\equiv0$ in $\mathcal{G}$, where $\mathcal{G}$ is defined to be the domain of dependence of $(\Sigma_0\cap I')\cup (\Sigma_0\cap II\cap\{v\leq1\})$. $\mathcal{G}$
is the shaded region  indicated in Figure \ref{sig0setup}.

Now $\phi_0=\phi_{sing}$ in $\mathcal{F}$ and furthermore, for any spherically symmetric solution $\tilde{\phi}$,  $\mathcal{L}_{\tilde{\phi}}$ depends only on the values of $\tilde{\phi}$ in the exterior region $\mathcal{F}$. Hence
\[\mathcal{L}_{\phi_0}=\mathcal{L}_{\phi_{sing}}\not=0.\]
 Also,  defining $\mathcal{L}'$ to be the analogous quantity to $\mathcal{L}$ defined in the exterior region $I'$, we have $\mathcal{L}_{\phi_0}'=0$ because $\phi_0\equiv0$ on the region $\mathcal{G}$. 

Moreover, by symmetry, or arguing analogously for the region $I'\cup II$, we have that there is a spherically symmetric solution $\phi_0':M\to\mathbb{R}$ with smooth and compactly supported initial data on $\Sigma_0$ such that $\mathcal{L}_{\phi_0'}'\not=0$ and $\phi_0'\equiv0$ in the region $\mathcal{G}'$ (defined analogously to $\mathcal{G}$), so that $\mathcal{L}_{\phi_0'}=0$.

Let
\[\phi=\phi_0+\phi_0'.\]
Then $\phi$  is a spherically symmetric solution of \eqref{waveequation} with smooth and compactly supported initial data on $\Sigma_0$ and satisfies
\begin{equation}\label{l}
\mathcal{L}_\phi=\mathcal{L}_{\phi_0}+\mathcal{L}_{\phi_0'}=\mathcal{L}_{\phi_0}+0\not=0
\end{equation}
and
\begin{equation}\label{l'}
\mathcal{L}_\phi'=\mathcal{L}_{\phi_0}'+\mathcal{L}_{\phi_0'}'=0+\mathcal{L}_{\phi_0'}'\not=0.
\end{equation}
Due to \eqref{l}, $\phi$ satisfies the hypothesis of Theorem \ref{mainthm}. Hence
\[\int_1^\infty\log^{\alpha(p)}\left(\frac{1}{\Omega}\right)|\partial_v\phi|^p(u,v)\,dv=+\infty\]
for every $u\in (-\infty,\infty)$ and every $\alpha(p)\geq 4p+4$. Hence, by  Remark \ref{remark}, $\|\phi\|_{W^{1,p}}$ blows up on a neighbourhood of every point of the ``incoming'' future Cauchy horizon.

Also, by \eqref{l'} and by construction of $\phi$, $\phi$ also satisfies the analogous result to Theorem \ref{mainthm} for region $I'$. It follows that $\|\phi\|_{W^{1,p}}$ blows up on a neighbourhood of every point of the ``outgoing'' future Cauchy horizon. In particular, $\|\phi\|_{W^{1,p}}$ blows up near every point of the bifurcate future Cauchy horizon, so $\phi\not\in W^{1,p}_{loc}$ near the future Cauchy horizon, as required.
\end{proof}

The proof of our main result, Theorem \ref{mtv2}, will thus be complete once we prove Theorem \ref{mainthm}. In the next Section, Section \ref{pf of mainthm}, we  reduce the task of proving this theorem to proving  another result (Theorem \ref{reduction}) using a contradiction argument. Then, in Section \ref{Outline of proof}, we state the two main ingredients in the proof of  Theorem \ref{reduction} and, assuming these two key results, we prove Theorem \ref{reduction}. Sections \ref{Interior} and \ref{Exterior} are devoted to proving the two stated key results, so as to close the proof of Theorem \ref{reduction}. Their proofs are quite involved and require the use of interpolation theory. We give some more details about the interpolation in Section \ref{interpintro} but defer a thorough discussion of interpolation and the  proofs of the results from Interpolation Theory which we shall need to Appendix \ref{Interpolation Theory}. We discuss the proof of the conditional result Theorem \ref{condinstab2} in Section \ref{condrel}, and give the proof in detail at the end of Section \ref{Interior}.

\subsection{Proof of Theorem \ref{mainthm}}\label{pf of mainthm}
We now set about reducing the task of proving Theorem \ref{mainthm} to proving a more tractable result.  In the black hole interior, for each $u\in(-\infty,\infty)$, $\lim_{v\to\infty}\frac{\log(1/\Omega)(u,v)}{v}$ has a finite positive value, so proving Theorem \ref{mainthm} is equivalent to proving the following theorem.

\begin{theorem}\label{bow}
Assume $\frac{2\sqrt e}{e+1}<Q<1$. Suppose $1<p<2$ and suppose $\phi$ is a solution of wave equation \eqref{waveequation} with smooth, spherically symmetric data on $S$ satisfying \eqref{ic1}, \eqref{ic2} and \eqref{limit} and such that $\mathcal{L}_\phi\not=0$. Then, in the black hole interior, $\phi$ satisfies
\begin{equation*}
\int_1^\infty v^{\alpha(p)}|\partial_v\phi|^p(u,v)\,dv=+\infty
\end{equation*}
for every $u\in (-\infty,\infty)$ and every $\alpha(p)\geq 4p+4$. 
\end{theorem}
\noindent The same proof works for any particular value of $u\in (-\infty,\infty)$. Consequently we prove the above theorem only for the case $u=-1$ and the general statement follows immediately.

The proof of Theorem \ref{bow} is by contradiction. Indeed, given $\phi$ satisfying the hypothesis of Theorem \ref{bow}, we shall assume
\begin{equation}\label{ca}
\int_1^\infty v^{4p+4}|\partial_v\phi|^p(-1,v)\,dv<+\infty.
\end{equation}
We can then propagate this estimate through the black hole region to the event horizon, and from the event horizon to the hypersurface $r=R$ in the exterior region (for $R>r_+$ sufficiently large). More precisely, we use \eqref{ca} to deduce
\begin{equation}\label{contra}
\int_{\gamma_R\cap\{v\geq 1\}}v^{3p-1}|\partial_v(r\phi)|^p <+\infty\hspace{3mm}\mbox{ for any } R>r_+\mbox{ sufficiently large.}
\end{equation}
However, in \cite{LO} the authors prove the following theorem.

\begin{theorem}[Luk--Oh \cite{LO}]\label{contrat}
Suppose $\phi$ is a solution of wave equation \eqref{waveequation} with smooth, spherically symmetric data on $S$ satisfying \eqref{ic1}, \eqref{ic2} and \eqref{limit} and such that $\mathcal{L}\not=0$. Then there exists a large constant $R_1=R_1(M)>r_+$ such 
that if 
\begin{equation}\label{contra3}
\sup_{\gamma_{R_1}\cap\{u\geq1\}}u^3|\phi|\leq A'
\end{equation}
for some $A'>0$, then there exist $R_0=R_0(\mathcal{L},A',D,U_0,R_1)>R_1$ and $U=U(\mathcal{L},A',D,U_0,\partial_v(r\phi)\vert_{C_{-U_0}})$ sufficiently large such that
\begin{equation}\label{contra2}
|\partial_v(r\phi)|(u,v)\vert_{\gamma_{R_0}\cap \{u\geq U\}}|\geq\frac{\mathcal{|L|}}{8}v^{-3}.
\end{equation}
\end{theorem}
\noindent Note in particular that we may (and do) assume $R_1$ is large enough that $R_1^*\geq 1$.

 Thus, provided we show that \eqref{contra3} follows from \eqref{ca}, it follows from \eqref{contra2} that
\begin{equation*}
\int_{\gamma_{R_0}\cap\{v\geq v_{R_0}(U)\}}v^{3p-1}|\partial_v(r\phi)|^p\geq\left(\frac{|\mathcal{L}|}{8}\right)^p \int_{\gamma_{R_0}\cap\{v\geq v_{R_0}(U)\}}v^{-1}=+\infty,
\end{equation*}
contradicting \eqref{contra} with $R=R_0>R_1$.
Thus the proof of Theorem \ref{bow} (and hence of Theorem \ref{mtv2}) is complete once we prove the following theorem.

\begin{theorem}\label{reduction}
Assume $\frac{2\sqrt e}{e+1}<Q<1$.  Let $1<p<2$ and suppose $\phi$ is a solution of wave equation \eqref{waveequation} with smooth, spherically symmetric data on $S$ satisfying \eqref{ic1} and \eqref{ic2}. Assume in the interior of the black hole $\phi$ satisfies 
\begin{equation}\label{bounda}
\int_1^\infty v^{4p+4}|\partial_v\phi|^p(-1,v)\,dv=A<+\infty
\end{equation}
for some $A>0$. Then there exist $C_1=C_1(A,R_1,D,p,\phi)>0$ and $C_2=C_2(A,R,D,p,\phi)>0$ such that
\begin{equation}\label{v4}
\sup_{\gamma_{R_1}\cap\{v\geq1\}}v^3|\phi|\leq C_1
\end{equation}
and for every $R>R_1$
\begin{equation}\label{eg} \int_{\gamma_R\cap\{v\geq 1\}}v^{3p-1}|\partial_v(r\phi)|^p \leq C_2.
\end{equation}
\end{theorem}

\subsubsection{Outline of Proof of Theorem \ref{reduction}}\label{Outline of proof}

The key ingredients in the proof of Theorem \ref{reduction} are the two theorems stated directly below. The first of these two theorems is an estimate in the black hole \emph{interior} and is proved in Section \ref{Interior}.

\begin{theorem}\label{interior reduction}
Assume $\frac{2\sqrt e}{e+1}<Q<1$, $1<p<2$ and $0<\theta<1$ are such that
\[\frac{1}{p}=\frac{1-\theta}{1}+\frac{\theta}{2}.\]
Let $\alpha>1$ be an integer\footnote{We may only consider positive integers $\alpha$  due to the fact that the analogous result for the case $p=2$ (\cite{LO}), upon which the proof of Theorem \ref{interior reduction} relies, is only known to be valid for such $\alpha$, due to the fact that its proof is by induction. See Proposition \ref{intl2} for details.} and $0<\varepsilon<1$. Then, given  a solution $\phi$ of the wave equation \eqref{waveequation} with smooth, spherically symmetric data on $S$ satisfying \eqref{ic2}, there exists a constant $C=C(\alpha,\varepsilon,p,D,\phi)>0$\footnote{We remark that it is possible to prove the estimate with a constant that is independent of both $D$ and $\phi$, as in Theorem \ref{lp exterior estimate}. The argument we use to prove Theorem \ref{interior reduction}, however, is crucial to proving the condition in Theorem \ref{condinstab2}, albeit at the expense of the constant depending on $D$ and $\phi$. See Section \ref{condrel2} and Remark \ref{rem} in Section \ref{remlab}  for details.}  such that in the black hole interior
\[\int_{\mathcal{H}^+\cap\{v\geq1\}} v^{(\alpha-\varepsilon)\frac{p\theta}{2}}|\partial_v\phi|^p \leq C\left(\int_1^\infty v^{\frac{\alpha p\theta}{2}}|\partial_v\phi|^p(-1,v)\,dv+\int_{\underline{C}_1^{int}\cap\{u\leq -1\}}\Omega^{-p\theta}|\partial_u\phi|^p\right).\]
\end{theorem}

The second theorem we shall need relates to the black hole \emph{exterior} and is proved in Section \ref{Exterior}.

\begin{theorem}\label{lp exterior estimate}
Let $\alpha,\varepsilon>0$ be such that $\alpha-(1+2\varepsilon)>0$ and suppose $\phi$ is a solution of wave equation \eqref{waveequation} with smooth, spherically symmetric data on $S$. Assume $1<p<2$ and $0<\theta<1$ with
\[\frac{1}{p}=\frac{1-\theta}{1}+\frac{\theta}{2}.\]
If $R>R_1$, then there exists $C=C(R,\alpha,\varepsilon,p)>0$ such that
\[\int_{\gamma_R\cap\{v\geq1\}} v^{(\alpha-(1+2\varepsilon))\frac{p\theta}{2}}|\partial_v\phi|^p \leq C \left(\int_{\mathcal{H}^+\cap\{v\geq1\}} v^{\frac{\alpha p\theta}{2}}|\partial_v\phi|^p+\int_{\underline{C}_1^{ext}\cap\{u\geq u_R(1)\}} \Omega^{-p\theta}|\partial_u\phi|^p\right) \]
and
\[\int_{\gamma_R\cap\{v\geq1\}} v^{(\alpha-(1+2\varepsilon))\frac{p\theta}{2}}|\partial_u\phi|^p \leq C \left(\int_{\mathcal{H}^+\cap\{v\geq1\}} v^{\frac{\alpha p\theta}{2}}|\partial_v\phi|^p+\int_{\underline{C}_1^{ext}\cap\{u\geq u_R(1)\}} \Omega^{-p\theta}|\partial_u\phi|^p\right).\]
\end{theorem}

Proving these two theorems is the main challenge faced in proving Theorem \ref{reduction} and hence our main result, Theorem \ref{mtv2}. Indeed, their proofs form the bulk of the remainder of this thesis. Both proofs make use of interpolation theory to deduce the desired estimates. We begin by proving ``endpoint estimates'', namely analogous estimates to those of the theorems for the cases $p=1$ and $p=2$, directly by analysis of the wave equation \eqref{waveequation} in spherical symmetry in the interior and exterior regions separately. We then appeal to the K-method of real interpolation to allow us to deduce the desired ``intermediate estimates'' for $1<p<2$. More details are given in Section \ref{interpintro}.

In addition to the theorems above, we need two more results about the black hole exterior to complete the proof of Theorem \ref{reduction}, though, unlike the theorems stated above, their proofs are straightforward and can be found in Section \ref{easyres}. We state these results directly below.

\begin{proposition}\label{alt2}
Assume $p>1$.
Suppose $\phi$ is a solution of wave equation \eqref{waveequation} with smooth, spherically symmetric data on $S$ satisfying \eqref{ic1} and \eqref{ic2} and let $R>r_+$ be such that $R^*\geq1$. Then there exist $C=C(R,\phi,p)>0$ and $\tilde{C}=\tilde{C}(R,p)$ such that
\begin{eqnarray*}
\left(\int_{\gamma_R\cap\{v\geq1\}} v^{3p-1}|\phi|^p\right)^{1/p}\leq C+\tilde{C}\left( \left(\int_{\gamma_R\cap\{v\geq1\}}v^{4p}|\partial_u\phi|^p\right)^{1/p}+ \left(\int_{\gamma_R\cap\{v\geq1\}}v^{4p}|\partial_v\phi|^p\right)^{1/p}\right).
\end{eqnarray*}
\end{proposition}

\begin{proposition}\label{alt1}
Assume $p>1$.
Suppose $R>r_+$ is such that $R^*\geq1$. Then, given a solution $\phi$ of the wave equation \eqref{waveequation} with smooth, spherically symmetric data on $S$ satisfying \eqref{ic1} and \eqref{ic2}, there exists $C=C(R,p)>0$ such that
\begin{eqnarray*}
\sup_{\gamma_R\cap\{v\geq R^*\}}v^3|\phi|\leq C\left( \left(\int_{\gamma_{R}\cap\{v\geq1\}}v^{4p}|\partial_u\phi|^p\right)^{1/p}+ \left(\int_{\gamma_{R}\cap\{v\geq1\}}v^{4p}|\partial_v\phi|^p\right)^{1/p}\right).\,\,\,\,
\end{eqnarray*}
\end{proposition}

We now give the proof of Theorem \ref{reduction} making use of Theorem \ref{interior reduction}, Theorem \ref{lp exterior estimate}, Proposition \ref{alt2} and Proposition \ref{alt1}. Then the remainder of this work is devoted to proving these results, so as to close the proof of Theorem \ref{reduction} and hence Theorem \ref{mtv2}.

\begin{proof}[Proof of Theorem \ref{reduction}.]
We begin by showing \eqref{eg}. Let $R>R_1$ so $R^*\geq R_1^*\geq 1$ (since $r^*$ is increasing in $r$ on the exterior region). For brevity, set $\tilde{\gamma}_R=\gamma_R\cap\{v\geq1\}$. Then, setting $\Omega^2_R=\Omega^2(r=R)$, by the product rule we have
\begin{eqnarray}
\left(\int_{\gamma_R\cap\{v\geq1\}}v^{3p-1}|\partial_v(r\phi)|^p\,dv\right)^{1/p} &\leq&2 \left(\int_{\tilde{\gamma}_R} v^{3p-1}(r^p|\partial_v\phi|^p+|\partial_vr|^p|\phi|^p)\,dv\right)^{1/p}\notag\\
&\leq&2R\left(\int_{\tilde{\gamma}_R}v^{3p-1}|\partial_v\phi|^p\,dv\right)^{1/p}+2\Omega^2_R\left(\int_{\tilde{\gamma}_R}v^{3p-1}|\phi|^p\,dv\right)^{1/p}\hspace{5mm}\notag\\
&\leq& 2R\left(\int_{\tilde{\gamma}_R}v^{4p}|\partial_v\phi|^p\,dv\right)^{1/p}+2\Omega^2_R\left(\int_{\tilde{\gamma}_R}v^{3p-1}|\phi|^p\,dv\right)^{1/p}.\hspace{5mm}\notag
\end{eqnarray}
Using Proposition \ref{alt2} to estimate the last term on the right hand side yields
\begin{eqnarray}
\left(\int_{\tilde{\gamma}_R}v^{3p-1}|\partial_v(r\phi)|^p\,dv\right)^{1/p} &\leq& C+\tilde{C}\left( \left(\int_{\tilde{\gamma}_{R}}v^{4p}|\partial_u\phi|^p\right)^{1/p}+ \left(\int_{\tilde{\gamma}_{R}}v^{4p}|\partial_v\phi|^p\right)^{1/p}\right),\hspace{8mm}\label{esst18}
\end{eqnarray}
where $C=C(R,\phi,p)$ and $\tilde{C}=\tilde{C}(R,p)$. We need to control the right hand side.

To this end, fix $\varepsilon_0<1$. Then, choosing $\alpha$ and $\varepsilon$ such that $\varepsilon=\varepsilon_0/2$ and $(\alpha-(1+2\varepsilon))\frac{p\theta}{2}=4p$, the first inequality of Theorem \ref{lp exterior estimate} yields
\begin{eqnarray}
\int_{\tilde{\gamma}_R}v^{4p}|\partial_v\phi|^p &\leq& C\left(\int_{\mathcal{H}^+\cap\{v\geq1\}}v^{4p+(1+\varepsilon_0)\frac{p\theta}{2}}|\partial_v\phi|^p+\int_{\underline{C}_1^{ext}\cap\{u\geq u_R(1)\}}\Omega^{-p\theta}|\partial_u\phi|^p\right)\notag\\
&\leq& C\left(\int_{\mathcal{H}^+\cap\{v\geq1\}}v^{4p+2}|\partial_v\phi|^p+\int_{\underline{C}_1^{ext}\cap\{u\geq u_R(1)\}}\Omega^{-p\theta}|\partial_u\phi|^p\right), \label{esst19}
\end{eqnarray}
where $C=C(R,\alpha,\varepsilon,p)=C(R,p)$. Note that we have used that $4p+(1+\varepsilon_0)\frac{p\theta}{2}<4p+2\frac{p\theta}{2}<4p+2$.
Similarly, the second inequality of Theorem \ref{lp exterior estimate} gives
\begin{eqnarray}
\int_{\tilde{\gamma}_R}v^{4p}|\partial_u\phi|^p\leq C\left(\int_{\mathcal{H}^+\cap\{v\geq1\}}v^{4p+2}|\partial_v\phi|^p+\int_{\underline{C}_1^{ext}\cap\{u\geq u_R(1)\}}\Omega^{-p\theta}|\partial_u\phi|^p\right), \hspace{7mm}\label{esst19'}
\end{eqnarray}
where $C=C(R,p)$. 

 Set
\[\gamma_0:=\left\{\begin{array}{ll}4p+2&\mbox{if }(4p+2)\frac{2}{p\theta}\not\in\mathbb{Z}\\4p+2+\frac{p\theta}{3} &\mbox{if }(4p+2)\frac{2}{p\theta}\in\mathbb{Z}\end{array}\right.,\]
so that $\gamma_0\geq 4p+2$ and $\frac{2}{p\theta}\gamma_0\not\in\mathbb{Z}$.
Then combining \eqref{esst18} with \eqref{esst19} and \eqref{esst19'} we have
\begin{equation}\label{esst20}
\left(\int_{\tilde{\gamma}_R}v^{3p-1}|\partial_v(r\phi)|^p\right)^{1/p}\leq C+C'\left(\int_{\mathcal{H}^+\cap\{v\geq1\}} v^{\gamma_0}|\partial_v\phi|^p+\int_{\underline{C}_1^{ext}\cap\{u\geq u_R(1)\}}\Omega^{-p\theta}|\partial_u\phi|^p\right)^{1/p},\hspace{3mm}
\end{equation}
for $C=C(R,\phi,p)$ and $C'=C'(R,p)$.

Again we need to control the right hand side. Due to the assumption \eqref{ic2} on the initial data, the $\underline{C}_1^{ext}$ integral is bounded, so it remains to bound the $\mathcal{H}^+$ integral, and this is achieved using Theorem \ref{interior reduction}. Set $\alpha =\lceil \frac{2}{p\theta}\gamma_0\rceil$ and $\varepsilon=\lceil \frac{2}{p\theta}\gamma_0\rceil- \frac{2}{p\theta}\gamma_0$. Then $\alpha\in\mathbb{Z}$, $\alpha\geq \gamma_0\geq 4p+2>1$ since $p>1$. Moreover, since $\frac{2}{p\theta}\gamma_0\not\in\mathbb{Z}$, it follows that $0<\varepsilon<1$. Note $\alpha-\varepsilon=\frac{2}{p\theta}\gamma_0\Longrightarrow (\alpha-\varepsilon)\frac{p\theta}{2}=\gamma_0$, and so by Theorem \ref{interior reduction} we have
\begin{eqnarray*}
\int_{\mathcal{H}^+\cap\{v\geq1\}} v^{\gamma_0}|\partial_v\phi|^p&\leq& C\left(\int_1^\infty v^{\gamma_0 +\varepsilon\frac{p\theta}{2}}|\partial_v\phi|^p(-1,v)\,dv+\int_{\underline{C}_1^{int}\cap\{u\leq -1\}}\frac{|\partial_u\phi|^p}{\Omega^{p\theta}}\right)\\
&\leq&C\left(\int_1^\infty v^{\gamma_0 +1}|\partial_v\phi|^p(-1,v)\,dv+\int_{\underline{C}_1^{int}\cap\{u\leq -1\}}\Omega^{-p\theta}|\partial_u\phi|^p\right)\\
&\leq&C\left(\int_1^\infty v^{4p+4}|\partial_v\phi|^p(-1,v)\,dv+\int_{\underline{C}_1^{int}\cap\{u\leq -1\}}\Omega^{-p\theta}|\partial_u\phi|^p\right)
\end{eqnarray*}
for some $C(\alpha,\varepsilon,p,D,\phi)=C(p,D,\phi)>0$, since $\gamma_0<4p+3$ by definition.

Thus, substituting the last inequality into \eqref{esst20}, we have
\begin{eqnarray*}
\left(\int_{\tilde{\gamma}_R}v^{3p-1}|\partial_v(r\phi)|^p\right)^{1/p}&\leq&C+ C\left(\int_1^\infty v^{4p+4}|\partial_v\phi|^p(-1,v)\,dv+\int_{\underline{C}_1^{int}\cap\{u\leq -1\}}\Omega^{-p\theta}|\partial_u\phi|^p\right.\\
& &+\left.\int_{\underline{C}_1^{ext}\cap\{u\geq u_R(1)\}}\Omega^{-p\theta}|\partial_u\phi|^p\right)^{1/p}\\
&\leq& C(A,R,D,p,\phi)
\end{eqnarray*}
by \eqref{bounda} and \eqref{ic2}. Indeed, as $p=\frac{2}{2-\theta}$, we have
\begin{eqnarray*}
\int_{\underline{C}_1^{ext}\cap\{u\geq u_R(1)\}}\Omega^{-p\theta}|\partial_u\phi|^p=\int_{\underline{C}_1^{ext}\cap\{u\geq u_R(1)\}}\left|\frac{\partial_u\phi}{\Omega^2}\right|^p\Omega^{p(2-\theta)} &\leq& D^p\int_{\underline{C}_1^{ext}\cap\{u\geq u_R(1)\}}\Omega^2 \\
&=&D^p\int_{u_R(1)}^\infty(-\partial_ur)(u,1)\,du\\
&\leq&D^p(R-r_+),
\end{eqnarray*}
and $\int_{\underline{C}_1^{int}\cap\{u\leq -1\}}\Omega^{-p\theta}|\partial_u\phi|^p$ can be estimated similarly.

It remains now to prove \eqref{v4}. We have
\begin{eqnarray}
\sup_{\gamma_{R_1}\cap\{v\geq1\}}v^3|\phi|&\leq& \sup_{\gamma_{R_1}\cap\{1\leq v\leq R_1^*\}}v^3|\phi|+\sup_{\gamma_{R_1}\cap\{v\geq R_1^*\}}v^3|\phi|\notag\\
&\leq& C(R_1,\phi)+\sup_{\gamma_{R_1}\cap\{v\geq R_1^*\}}v^3|\phi|,\label{vlih}
\end{eqnarray}
and we use Proposition \ref{alt1} (with $R=R_1$) to estimate the second term on the right hand side. This yields 
\begin{eqnarray}
\sup_{\gamma_{R_1}\cap\{v\geq R_1^*\}}v^3|\phi|\leq C\left( \left(\int_{\gamma_{R_1}\cap\{v\geq1\}}v^{4p}|\partial_u\phi|^p\right)^{1/p}+ \left(\int_{\gamma_{R_1}\cap\{v\geq1\}}v^{4p}|\partial_v\phi|^p\right)^{1/p}\right),\,\,\,\,\label{supest}
\end{eqnarray}
where $C=C(R_1,p)$.
 Now, arguing exactly as before (for the proof of \eqref{eg}), we can estimate the right hand side of \eqref{supest} by
\begin{eqnarray*}
C(R_1,p,D,\phi)\left(\int_1^\infty v^{4p+4}|\partial_v\phi|^p(-1,v)\,dv+\int_{\underline{C}_1^{int}\cap\{u\leq -1\}}\Omega^{-p\theta}|\partial_u\phi|^p+\int_{\underline{C}_1^{ext}\cap\{u\geq u_R(1)\}}\Omega^{-p\theta}|\partial_u\phi|^p\right)^{1/p}\\
\leq C(A,R_1,D,p,\phi)\hspace{110mm}
\end{eqnarray*}
for some $C(A,R_1,D,p,\phi)>0$.
Thus by \eqref{vlih}
\begin{eqnarray*}
\sup_{\gamma_{R_1}\cap\{v\geq1\}}v^3|\phi|
&\leq& C(R_1,\phi)+C(A,R_1,D,p,\phi),
\end{eqnarray*}
which proves \eqref{v4}.
\end{proof}

The above proof will be complete once we prove  Theorem \ref{interior reduction}, Theorem \ref{lp exterior estimate}, Proposition \ref{alt2} and Proposition \ref{alt1}. The proofs of the two propositions are straightforward, but to prove Theorem \ref{interior reduction} and Theorem \ref{lp exterior estimate} we will need to interpolate $L^1$ and $L^2$-type estimates for the derivative of $\phi$ in order to deduce $L^p$-type estimates. For this, we will need some tools from interpolation theory.

\begin{remark}
We remark that the only place in the proof of Theorem \ref{reduction} where we needed the parameter restriction $\frac{2\sqrt e}{e+1}<Q<1$ was  Theorem \ref{interior reduction}. The rest of the above proof is valid for $0<Q<1$, namely on any subextremal Reissner--Nordstr\"{o}m spacetime with non-vanishing charge.
\end{remark}

\subsection{Interpolation}\label{interpintro}
As noted above, proving Theorems  \ref{interior reduction} and \ref{lp exterior estimate} requires us to prove inequalities of the form
\begin{equation}\label{exampleest}
\int_{\gamma_R\cap\{v\geq1\}} v^{(\alpha-(1+2\varepsilon))\frac{p\theta}{2}}|\partial_v\phi|^p \leq C \left(\int_ {\mathcal{H}^+\cap\{v\geq1\}} v^{\frac{\alpha p\theta}{2}}|\partial_v\phi|^p+\int_{\underline{C}_1^{ext}\cap\{u\geq u_R(1)\}} \Omega^{-p\theta}|\partial_u\phi|^p\right) 
\end{equation}
and
\[\int_{\gamma_R\cap\{v\geq1\}} v^{(\alpha-(1+2\varepsilon))\frac{p\theta}{2}}|\partial_u\phi|^p \leq C \left(\int_{\mathcal{H}^+\cap\{v\geq1\}} v^{\frac{\alpha p\theta}{2}}|\partial_v\phi|^p+\int_{\underline{C}_1^{ext}\cap\{u\geq u_R(1)\}} \Omega^{-p\theta}|\partial_u\phi|^p\right) \]
in the black hole exterior for $R>R_1$, and 
\[\int_{\mathcal{H}^+\cap\{v\geq1\}} v^{(\alpha-\varepsilon)\frac{p\theta}{2}}|\partial_v\phi|^p \leq C\left(\int_1^\infty v^{\frac{\alpha p\theta}{2}}|\partial_v\phi|^p(-1,v)\,dv+\int_{\underline{C}_1^{int}\cap\{u\leq-1\}}\Omega^{-p\theta}|\partial_u\phi|^p\right)\]
in the black hole interior (so long as $\frac{2\sqrt e}{e+1}<Q<1$), where $\alpha\in\mathbb{N}$, $\alpha>1$, $1<p<2$, $0<\theta<1$ and 
\[\frac{1}{p}=\frac{1-\theta}{1}+\frac{\theta}{2}.\]

For the case $p=2$ ($\theta=1$), similar inequalities are proved directly in \cite{LO}, but the arguments given there do not generalise to the case $1<p<2$. The reason is chiefly because the exponent $p=2$ allows one to deduce positivity of tems where it occurs, and hence deduce desirable inequalities. However, this is of course no longer true when $p\not=2$, and so a different strategy is needed to deduce the more general $L^p$-type estimates above.

Our strategy is to use real interpolation to establish these estimates. Namely, we prove endpoint estimates, that is estimates for the cases $p=1$ and $p=2$ (some of which are already established in \cite{LO}, some of which we prove), and appeal to real interpolation to deduce the intermediate estimates with $1<p<2$.

For instance, in order to prove \eqref{exampleest}, we first prove the corresponding estimates for $p=1$ (and $\theta=0$) and for $p=2$ (and $\theta=1$), that is we prove
\[\int_{\gamma_R\cap\{v\geq1\}} |\partial_v\phi| \leq C \left(\int_{\mathcal{H}^+\cap\{v\geq1\}} |\partial_v\phi|+\int_{\underline{C}_1^{ext}\cap\{u\geq u_R(1)\}} |\partial_u\phi|\right) \]
and 
\[\int_{\gamma_R\cap\{v\geq1\}} v^{\alpha-(1+2\varepsilon)}(\partial_v\phi)^2 \leq C \left(\int_{\mathcal{H}^+\cap\{v\geq1\}} v^\alpha(\partial_v\phi)^2+\int_{\underline{C}_1^{ext}\cap\{u\geq u_R(1)\}} \left(\frac{\partial_u\phi}{\Omega}\right)^2\right).\]
As we will show later, these two estimates amount to showing that an operator $T$ mappring 
\[T:(\partial_v\phi\vert_{\mathcal{H}^+\cap\{v\geq1\}},\partial_u\phi\vert_{\underline{C}_1^{ext}\cap\{u\geq u_R(1)\}})\mapsto \partial_v\phi\vert_{\gamma_R\cap\{v\geq1\}}\] 
is bounded in both a $L^1$ and a weighted $L^2$ sense. It then follows from the theory of real interpolation that this operator is bounded in a weighted $L^p$ sense, namely \eqref{exampleest}.

Background information on the K-method of real interpolation and proofs of the results we shall use are included in Appendix \ref{Interpolation Theory}. Note that standard terminology, notation and conventions from Interpolation Theory introduced there will be used throughout the remainder of this work. 

\subsubsection{Remarks on the Proof of Theorem \ref{condinstab2}}\label{condrel}
Recall that Theorem \ref{condinstab2} asserts generic  $W^{1,p}_{loc}$ blow up  of solutions to \eqref{waveequation} for the full subextremal range of black hole parameters with not vanishing charge \emph{unless} the solution map is not bounded below. On the other hand, Theorem \ref{mtv2} guarantees the generic $W^{1,p}_{loc}$ blow up of solutions for the parameter range $\frac{2\sqrt e}{e+1}<Q<1$. Thus,
in order to prove the conditional instability result Theorem \ref{condinstab2}, we examine the proof of Theorem \ref{mtv2} to identify where the proof does not work for the full parameter range and this yields the desired condition which prevents us from deducing instability for $0< Q\leq\frac{2\sqrt e}{e+1}$.

Notice that the only part of the proof of Theorem \ref{mtv2} that does not hold for the full range $0<Q<1$ is Theorem \ref{mainthm}, and in turn the only part of the proof of Theorem \ref{mainthm} that does not hold for the full range is Theorem \ref{interior reduction}, namely an estimate in the black hole interior. As mentioned above, Theorem \ref{interior reduction} is proved using interpolation theory: we establish analogous estimates to those in Theorem \ref{interior reduction} for the cases $p=1$ and $p=2$ in Section \ref{Interior}. While the estimate for $p=2$ holds for all subextremal parameters with non-vanishing charge $0<Q<1$, the estimate for $p=1$ does not. We succeed in  establishing this black hole interior $L^1$-type estimate only for the parameter range $\frac{2\sqrt e}{e+1}<Q<1$. We note that if we could extend the $L^1$-type estimate to the case $0<Q\leq \frac{2\sqrt e}{e+1}$, then we could deduce $W^{1,p}_{loc}$ instability for this range exactly as for the case $\frac{2\sqrt e}{e+1}<Q<1$. If the $L^1$ estimates do not hold for this range however, then our method of proof does not guarantee instability. We therefore conclude $W^{1,p}_{loc}$ instability unless  the $L^1$ estimate is false, precisely the condition given in Theorem \ref{condinstab2}. We give the proof of Theorem \ref{condinstab2} in detail in Section \ref{condrel2}, after the proof of the interior $L^p$ estimate for the case $\frac{2\sqrt e}{e+1}<Q<1$, namely Theorem \ref{interior reduction}. 

\subsection*{Outline}
In the next Section, Section \ref{Interior}, we consider the black hole interior and prove the necessary $L^1$ and $L^2$-type estimates for this region. We then use real interpolation to conclude Theorem \ref{interior reduction}. We also give the proof of the conditional instability result, Theorem \ref{condinstab2}.

In Section \ref{Exterior}, the focus is the black hole exterior. Here we prove the relevant $L^1$ and $L^2$ estimates for the  exterior region and then interpolate between them to deduce Theorem \ref{lp exterior estimate}. We also prove Propositions \ref{alt2} and \ref{alt1} and thus close the proof of Theorem \ref{reduction}.

Appendix \ref{Interpolation Theory} is devoted to the K-method of real interpolation, and discusses all the background material needed to employ it in the proofs of Theorems \ref{interior reduction} and \ref{lp exterior estimate}. Finally, Appendix \ref{Switch} contains the proof of a technical lemma needed for the proof of Theorem \ref{interior reduction}.

\newpage
\section{Estimates in the Black Hole Interior}\label{Interior}
Recall that in order to close the proof of Theorem \ref{reduction}, we need to prove Theorem \ref{interior reduction}, Theorem \ref{lp exterior estimate}, Proposition \ref{alt2} and Proposition \ref{alt1}. The latter three are results pertaining to the black hole exterior, and we defer their proofs to the next section, and now focus on proving Theorem \ref{interior reduction}, which relates to the black hole interior.

\subsection{Proof of Theorem \ref{interior reduction}}

In this section, we prove Theorem \ref{interior reduction},  which we restate here for convenience. 

\begin{theorem*}
Assume $\frac{2\sqrt e}{e+1}<Q<1$, $1<p<2$ and $0<\theta<1$ are such that
\[\frac{1}{p}=\frac{1-\theta}{1}+\frac{\theta}{2}.\]
Let $\alpha>1$ be an integer and $0<\varepsilon<1$. Then, given  a solution $\phi$ of the wave equation \eqref{waveequation} with smooth, spherically symmetric data on $S$ satisfying \eqref{ic2}, there exists a constant $C=C(\alpha,\varepsilon,p,D,\phi)>0$ such that in the black hole interior
\[\int_{\mathcal{H}^+\cap\{v\geq1\}} v^{(\alpha-\varepsilon)\frac{p\theta}{2}}|\partial_v\phi|^p \leq C\left(\int_1^\infty v^{\frac{\alpha p\theta}{2}}|\partial_v\phi|^p(-1,v)\,dv+\int_{\underline{C}_1^{int}\cap\{u\leq -1\}}\frac{|\partial_u\phi|^p}{\Omega^{p\theta}}\right).\]
\end{theorem*}

\noindent Notice that this is a result about the black hole interior, and so for the rest of this section we shall be working in the coordinate system for the interior region introduced in Section \ref{interior coordinates}.

In \cite{LO}, Luk--Oh proved an analogous result for the case $p=2$. As mentioned above, the direct argument they used to establish the $L^2$-type estimate does not go through to the case $1<p<2$, and so to prove the $L^p$-type estimate, we instead prove an $L^1$-type estimate and interpolate between this and the $L^2$-type estimate of \cite{LO}.

 In the next subsection we establish the $L^1$-type estimate. Then we give the statement of the $L^2$-type estimate deduced in \cite{LO} and provide a proof  (as the structure of the estimate was not explicitly stated in \cite{LO} as it was only necessary to establish finiteness there). Finally, we interpolate between the $L^1$ and $L^2$ estimates to deduce the required $L^p$-type estimate above.

\subsubsection{$L^1$ Estimates}\label{interior l1 estimates}

The $L^1$-type estimates that we need are proved using ideas from the proof of Proposition 13.1 of \cite{D}.

\begin{proposition}\label{intl1}
 Suppose $\frac{2\sqrt{e}}{e+1}<Q<1$ and let $\phi$ be a solution of the wave equation \eqref{waveequation} with smooth, spherically symmetric data on $S$. Then there exists $C>0$ (depending only on the spacetime parameters) such that
\begin{equation}\label{star}
\int_{\mathcal{H}^+\cap\{v\geq1\}}|\partial_v\phi|\leq C\left(\int_1^\infty|\partial_v\phi|(-1,v)\,dv+\int_{\underline{C}_1^{int}\cap\{u\leq-1\}}|\partial_u\phi|\right).
\end{equation}
\end{proposition}

\begin{proof}
Recall from \eqref{ln2} that in the black hole interior we have $\lambda=\partial_vr=-\Omega^2\leq 0$ and $\nu=\partial_ur=-\Omega^2\leq 0$. Note that we are in spherical symmetry and may rewrite the wave equation \eqref{extwe} as
\[\partial_u(r\partial_v\phi)=-\lambda\partial_u\phi.\]
For each $v\geq 1$, we integrate over $[u,-1]\times\{v\}$ to deduce
\[r|\partial_v\phi|(u,v)\leq r|\partial_v\phi|(-1,v)+\int_u^{-1}|\lambda\partial_u\phi|\,du,\]
and hence
\begin{eqnarray*}
\sup_{u\in[-\infty,-1]}r|\partial_v\phi|(u,v) \leq r|\partial_v\phi|(-1,v)+\int_{-\infty}^{-1}\frac{|\lambda|}{r} r|\partial_u\phi|\,du.
\end{eqnarray*}
Integrating over $v\in[1,\infty)$ gives
\begin{eqnarray}
\int_1^\infty\sup_{u\in[-\infty,-1]}r|\partial_v\phi|(u,v)\,dv &\leq&\int_1^\infty r|\partial_v\phi|(-1,v)\,dv+\int_1^\infty \int_{-\infty}^{-1}\frac{|\lambda|}{r} r|\partial_u\phi|\,du\,dv.\notag
\end{eqnarray}
But
\begin{eqnarray}
\int_1^\infty \int_{-\infty}^{-1}\frac{|\lambda|}{r} r|\partial_u\phi|\,du\,dv&\leq& \sup_{u\in[-\infty,-1]}\int_1^\infty\frac{|\lambda|}{r}\,dv \int_{-\infty}^{-1}\sup_{v\in[1,\infty)}r|\partial_u\phi|\,du\notag\\
&=&\sup_{u\in[-\infty,-1]}\int_1^\infty \frac{-\partial_vr}{r}\,dv\int_{-\infty}^{-1}\sup_{v\in[1,\infty)}r|\partial_u\phi|\,du\notag\\
&\leq&\log\frac{r_+}{r_-}\int_{-\infty}^{-1}\sup_{v\in[1,\infty)}r|\partial_u\phi|\,du,\notag
\end{eqnarray}
and hence
\begin{equation}\label{supu}
\int_1^\infty\sup_{u\in[-\infty,-1]}r|\partial_v\phi|(u,v)\,dv \leq\int_1^\infty r|\partial_v\phi|(-1,v)\,dv + \log\frac{r_+}{r_-}\int_{-\infty}^{-1}\sup_{v\in[1,\infty)}r|\partial_u\phi|\,du.
\end{equation}

Similarly, however, we may write the wave equation \eqref{extwe} as
\[\partial_v(r\partial_u\phi)=-\nu\partial_v\phi.\]
For each $u\leq -1$, we integrate over $\{u\}\times[1,v']$ to get
\[r|\partial_u\phi|(u,v')\leq r|\partial_u\phi|(u,1)+\int_1^{v'}|\nu\partial_v\phi|\,dv.\]
Hence, taking the supremum over $v'\geq1$ and integrating with respect to $u$ yields
\[\int_{-\infty}^{-1}\sup_{v\in[1,\infty)}r|\partial_u\phi|(u,v)\,du \leq \int_{-\infty}^{-1}r|\partial_u\phi|(u,1)\,du+\int_{-\infty}^{-1}\int_1^{\infty}|\nu\partial_v\phi|\,dv\,du.\]
Now,
\begin{eqnarray*}
\int_{-\infty}^{-1}\int_1^{\infty}|\nu\partial_v\phi|\,dv\,du&\leq&\sup_{v\in[1,\infty]}\int_{-\infty}^{-1}\frac{|\nu|}{r}\,du\int_1^\infty\sup_{u\in[-\infty,-1]}r|\partial_v\phi|\,dv\\
&\leq&\sup_{v\in[1,\infty]}\int_{-\infty}^{-1}\frac{-\partial_ur}{r}\,du\int_1^\infty\sup_{u\in[-\infty,-1]}r|\partial_v\phi|\,dv\\
&\leq&\log\frac{r_+}{r_-}\int_1^\infty\sup_{u\in[-\infty,-1]}r|\partial_v\phi|\,dv,
\end{eqnarray*}
so 
\[\int_{-\infty}^{-1}\sup_{v\in[1,\infty)}r|\partial_u\phi|(u,v)\,du \leq \int_{-\infty}^{-1}r|\partial_u\phi|(u,1)\,du+\log\frac{r_+}{r_-}\int_1^\infty\sup_{u\in[-\infty,-1]}r|\partial_v\phi|\,dv.\]
Combining this with \eqref{supu} yields
\begin{eqnarray*}
\int_1^\infty\sup_{u\in[-\infty,-1]}r|\partial_v\phi|(u,v)\,dv &\leq&\int_1^\infty r|\partial_v\phi|(-1,v)\,dv + \log\frac{r_+}{r_-}\int_{-\infty}^{-1}r|\partial_u\phi|(u,1)\,du\\
& & +\left( \log\frac{r_+}{r_-}\right)^2\int_1^\infty\sup_{u\in[-\infty,-1]}r|\partial_v\phi|\,dv.
\end{eqnarray*}
Thus, provided $ (\log\frac{r_+}{r_-})^2<1$ (or equivalently $\log\frac{r_+}{r_-}<1$, as $r_+>r_-$), we have
\[\int_1^\infty\sup_{u\in[-\infty,-1]}r|\partial_v\phi|(u,v)\,dv \leq C\left(\int_1^\infty r|\partial_v\phi|(-1,v)\,dv + \int_{-\infty}^{-1}r|\partial_u\phi|(u,1)\,du\right)\]
for some $C>0$ depending only on $r_+,r_-$ (and hence only on the spacetime parameters). Thus
\[\int_1^\infty r|\partial_v\phi|(-\infty,v)\,dv \leq C\left(\int_1^\infty r|\partial_v\phi|(-1,v)\,dv + \int_{-\infty}^{-1}r|\partial_u\phi|(u,1)\,du\right),\]
which gives the desired inequality as $r_-\leq r(u,v)\leq r_+$ in the black hole interior.

Now, recalling that $Q:=\frac{|q|}{M}$, we see that
\begin{eqnarray*}
\log\frac{r_+}{r_-}<1 \iff\frac{M+\sqrt{M^2-q^2}}{M-\sqrt{M^2-q^2}}<e
&\iff&\frac{(M+\sqrt{M^2-q^2})^2}{M^2-(M^2-q^2)}<e\\
&\iff& \frac{M+\sqrt{M^2-q^2}}{|q|}<\sqrt e\\
&\iff&\frac{1}{Q}+\sqrt{\frac{1}{Q^2}-1}<\sqrt e\\
&\iff&Q>\frac{2\sqrt e}{e+1},
\end{eqnarray*}
so \eqref{star} holds for $\frac{2\sqrt{e}}{e+1}<Q<1$, as desired.
\end{proof}

\begin{remark}
We emphasise that this is the only point in the entire proof of Theorem \ref{mtv2} that we need the assumption $\log\frac{r_+}{r_-}<1$ or equivalently $\frac{2\sqrt e}{e+1}<Q<1$.
\end{remark}

\subsubsection{$L^2$ Estimates}\label{star8}

We now give the $L^2$ type estimate which we shall need. 

\begin{proposition}\label{intl2}
Let $\alpha>0$ be an integer and $0<\varepsilon<1$. Then, given a solution $\phi$ of the wave equation \eqref{waveequation} with smooth, spherically symmetric data on $S$, there exists $C=C(\alpha,\varepsilon)>0$ such that
\begin{eqnarray}
\int_{\mathcal{H}^+\{v\geq1\}} v^{\alpha-\varepsilon}(\partial_v\phi)^2\leq C\left(\int_1^\infty v^\alpha(\partial_v\phi)^2(-1,v)\,dv+\int_{\underline{C}_1^{int}\cap\{u\leq -1\}}  \big|\frac{\partial_u\phi}{\Omega}\big|^2\right). \hspace{2mm}\label{el2}
\end{eqnarray}
\end{proposition}

This result is proved in \cite{LO}, though it is not formulated in the form we have stated it above: in \cite{LO} the aim was to show that the left hand side of \eqref{el2} can be bounded under suitable conditions, so the form of the right hand side was not important. For us however, the form of the right hand side of the estimate is of crucial importance (as we intend to interpolate between it and the $L^1$ estimate of the previous section). Consequently, we provide the proof below so as to emphasise the form of the estimate, but we stress that the proof is taken from \cite{LO} (in particular Propositions 4.5 and 4.6 of \cite{LO}). Before giving the proof however we introduce some notation and definitions from \cite{LO} that will be needed.

Given $\tau\geq1$, we set
\[\Gamma_\tau^{(1)}=\{(-\tau,v):v\geq\tau\}\]
and
\[\Gamma_\tau^{(2)}=\{(u,\tau):u\leq-\tau\}.\]
For $1\leq \tau\leq\tau'$, we define
\[\mathcal{H}^+(\tau,\tau')=\mathcal{H}^+\cap\{\tau\leq v\leq\tau'\}\]
and 
\[\mathcal{CH}^+(\tau,\tau')=\mathcal{CH}^+\cap\{-\tau'\leq u\leq -\tau\}.\] 
We denote by $\mathcal{D}(\tau,\tau')$ the spacetime region enclosed by $\Gamma_\tau^{(1)}$, $\Gamma_\tau^{(2)}$,  $\mathcal{H}^+(\tau,\tau')$, $\mathcal{CH}^+(\tau,\tau')$, $\Gamma_{\tau'}^{(1)}$ and $\Gamma_{\tau'}^{(2)}$. These objects are illustrated in the Penrose diagram below. 
\begin{figure}[h]
\centering
\includegraphics[scale=0.45]{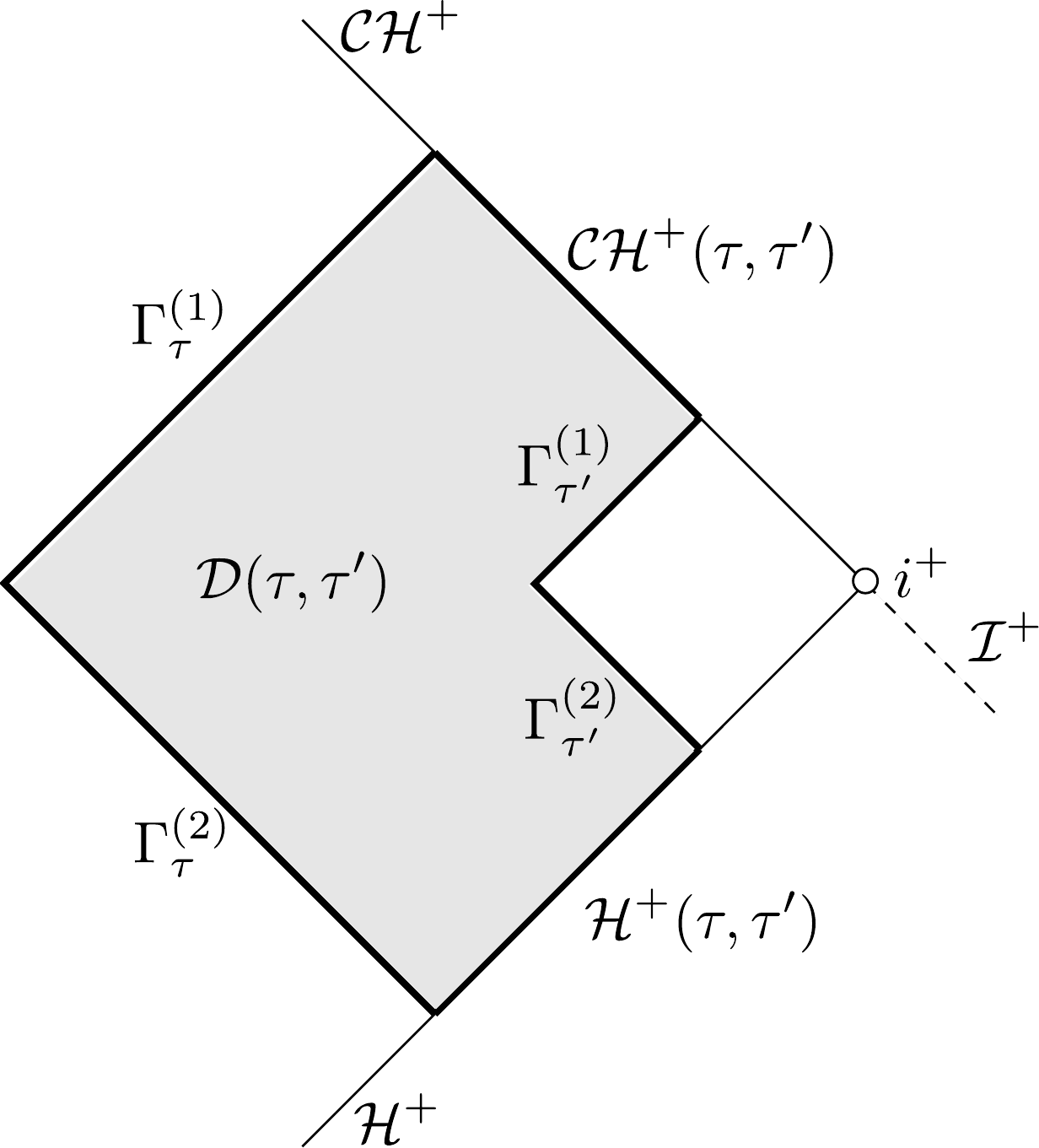}
\caption{The region $\mathcal{D}(\tau,\tau')$.}
\label{newnotation}
\end{figure}

Finally, we let $\chi_1$ and $\chi_2$ be smooth cutoff functions near the Cauchy horizon and the event horizon respectively, given by
\[\chi_1(r):=\left\{\begin{array}{ll}1&\mbox{if }r_-\leq r\leq r^{(1)}\\0&\mbox{if }r\geq r^{(1)}+(r^{(1)}-r_-)\end{array}\right. \hspace{5mm}\mbox{for }r^{(1)}>r_-\]
and 
\[\chi_2(r):=\left\{\begin{array}{ll}1&\mbox{if }r^{(2)}\leq r\leq r_+\\0&\mbox{if }r\leq r^{(2)}-(r_+-r^{(2)})\end{array}\right. \hspace{5mm}\mbox{for }r^{(2)}<r_+.\]
 We state Proposition 4.5 of \cite{LO} directly below as the proof of Propostion \ref{intl2} relies on it.
 
 \begin{proposition}\label{takenwholesale}
Let $1\leq \tau\leq\tau'$ and let $\phi$ be a smooth, spherically symmetric solution of the wave equation \eqref{waveequation}. Then for every $\beta\geq0$, there exists a constant $C=C(\beta)>0$ such that if  $r^{(1)}$ and $r^{(2)}$ are sufficiently close to $r_-$ and $r_+$ (independently of $\beta$), then
 \begin{eqnarray}
 \int_{\Gamma_{\tau'}^{(1)}}\left(1+\chi_1(r)\log^\beta(\frac{1}{\Omega})\right)(\partial_v\phi)^2 + \int_{\Gamma_{\tau'}^{(2)}}\Omega^{-2}(\partial_u\phi)^2 + \int_{\mathcal{CH}^+(\tau,\tau')}(\partial_u\phi)^2  +\int_{\mathcal{H}^+(\tau,\tau')}(\partial_v\phi)^2\notag\\
+ \iint_{\mathcal{D}(\tau,\tau')}\left(\Omega^2+\beta\chi_1(r)\log^{\beta-1}(\frac{1}{\Omega})\right)(\partial_v\phi)^2 
 +\left(\Omega^2+\chi_2(r)\Omega^{-2}\right)(\partial_u\phi)^2\notag\hspace{15mm}\\
 \leq C\left(\int_{\Gamma_\tau^{(1)}}\left(1+\chi_1(r)\log^\beta(\frac{1}{\Omega})\right)(\partial_v\phi)^2 + \int_{\Gamma_{\tau}^{(2)}}\Omega^{-2}(\partial_u\phi)^2\right).\label{twe}
 \end{eqnarray}
 \end{proposition}

We may now give the proof of Proposition \ref{intl2}.

\begin{proof}[Proof of Proposition \ref{intl2}]
First,  note that there exists $V=V(\varepsilon)>1$ such that $v^\varepsilon\geq \log^2(1+v)\,\forall v\geq V$, and so
\begin{eqnarray*}
\int_{\mathcal{H}^+\{v\geq1\}} v^{\alpha-\varepsilon}(\partial_v\phi)^2&=&\int_1^Vv^{\alpha-\varepsilon}(\partial_v\phi)^2(-\infty,v)\,dv + \int_V^\infty v^{\alpha-\varepsilon}(\partial_v\phi)^2(-\infty,v)\,dv\\
&\leq&\max_{v\in[1,V]} \left(\frac{\log^2(1+v)}{v^\epsilon}\right)\int_1^V \frac{v^\alpha}{\log^2(1+v)}(\partial_v\phi)^2(-\infty,v)\,dv \\
& &+\int_V^\infty \frac{v^\alpha}{\log^2(1+v)}(\partial_v\phi)^2(-\infty,v)\,dv \\
&\leq& C(V,\varepsilon) \int_{\mathcal{H}^+\cap\{v\geq1\}}\frac{v^\alpha}{\log^2(1+v)}(\partial_v\phi)^2.
\end{eqnarray*}
Thus it will be sufficitent for us to show that we can control
\[\int_{\mathcal{H}^+\cap\{v\geq1\}}\frac{v^\alpha}{\log^2(1+v)}(\partial_v\phi)^2\]
by the right hand side of \eqref{el2}

We show this by first showing that 
\begin{equation}\label{preest}
\tau^{\alpha}\int_{\mathcal{H}^+(\tau,\infty)}(\partial_v\phi)^2\leq C\left(\int_1^\infty \left(1+\chi_1(r)\log^\alpha(\frac{1}{\Omega})\right)(\partial_v\phi)^2(-1,v)\,dv+\int_{\underline{C}_1^{int}\cap\{u\leq -1\}}  \big|\frac{\partial_u\phi}{\Omega}\big|^2\right)
\end{equation}
for $\tau\geq 1$, where $C=C(\alpha)>0$. To see this, we prove the statement
\begin{eqnarray}
\tau^n\int_{\mathcal{H}^+(\tau,\infty)}(\partial_v\phi)^2 +\sum_{j=0}^n(\alpha-j)\tau^j\iint_{\mathcal{D}(\tau,\infty)}\left(\chi_1(r)\log^{\alpha-j-1}(\frac{1}{\Omega})\right)(\partial_v\phi)^2\hspace{10mm}\notag\\
+\tau^n\iint_{\mathcal{D}(\tau,\infty)}\chi_2(r)\Omega^{-2}(\partial_u\phi)^2 +\tau^n\iint_{\mathcal{D}(\tau,\infty)}\Omega^2\left((\partial_v\phi)^2+(\partial_u\phi)^2\right) 
\leq \mathcal{I}_n\label{indarg}
\end{eqnarray}
for $n=\alpha$ and $\tau\geq1$, where
\[\mathcal{I}_n=C_n(\alpha)\left(\int_1^\infty \left(1+\chi_1(r)\log^\alpha(\frac{1}{\Omega})\right)(\partial_v\phi)^2(-1,v)\,dv+\int_{\underline{C}_1^{int}\cap\{u\leq -1\}}  \big|\frac{\partial_u\phi}{\Omega}\big|^2\right)\]
 by induction on $n$. Note that \eqref{preest} follows immediately  from \eqref{indarg} with $n=\alpha$.
 
 For $n=0$, this follows from Proposition \ref{takenwholesale}. Indeed, when $n=0$, the left hand side of \eqref{indarg} reduces to
 \[\int_{\mathcal{H}^+(\tau,\infty)}(\partial_v\phi)^2+ \iint_{\mathcal{D}(\tau,\infty)}\left(\Omega^2+\alpha\chi_1(r)\log^{\alpha-1}(\frac{1}{\Omega})\right)(\partial_v\phi)^2+\iint_{\mathcal{D}(\tau,\infty)}\left(\Omega^2+\chi_2(r)\Omega^{-2}\right)(\partial_u\phi)^2,\]
 which by Proposition \ref{takenwholesale} with $\beta=\alpha$ (and sending $\tau'\to \infty$) is bounded  by
 \[ C\left(\int_{\Gamma_\tau^{(1)}}\left(1+\chi_1(r)\log^\alpha(\frac{1}{\Omega})\right)(\partial_v\phi)^2 + \int_{\Gamma_{\tau}^{(2)}}\Omega^{-2}(\partial_u\phi)^2\right),\]
 where $C=C(\alpha)>0$. But by Proposition \ref{takenwholesale} (with $\tau=1$ and $\tau'=\tau$), this in turn is bounded by 
 \begin{eqnarray*}
\mathcal{I}_0&=& \tilde{C}(\alpha)\left(\int_{\Gamma_1^{(1)}}\left(1+\chi_1(r)\log^\alpha(\frac{1}{\Omega})\right)(\partial_v\phi)^2 + \int_{\Gamma_{1}^{(2)}}\Omega^{-2}(\partial_u\phi)^2\right)\\
 &=& C_0(\alpha)\left(\int_1^\infty \left(1+\chi_1(r)\log^\alpha(\frac{1}{\Omega})\right)(\partial_v\phi)^2(-1,v)\,dv+\int_{\underline{C}_1^{int}\cap\{u\leq -1\}}  \big|\frac{\partial_u\phi}{\Omega}\big|^2\right),
 \end{eqnarray*}
 as desired, so \eqref{indarg} holds for $n=0$.
 
 Now, for the sake of induction, assume \eqref{indarg} holds for $n=0,1,\ldots,n_0-1$ for some positive integer $n_0\leq \alpha$.  Then by the pigeonhole principle, for every $k\in\mathbb{N}\cup\{0\}$ there exists $\tau_k\in[2^k,2^{k+1}]$ such that
 \begin{eqnarray}
 \int_{\Gamma_{\tau_k}^{(1)}}(\left(\Omega^2+(\alpha-n_0+1)\chi_1(r)\log^{\alpha-n_0}(\frac{1}{\Omega})\right)(\partial_v\phi)^2 +\Omega^2(\partial_u\phi)^2)& &\notag\\
 +\int_{\Gamma_{\tau_k}^{(2)}}(\left(\Omega^2+\chi_2(r)\Omega^{-2}\right)(\partial_u\phi)^2+\Omega^2(\partial_v\phi)^2)\hspace{5mm} &\leq&\hspace{5mm} C\mathcal{I}_{n_0-1}\tau_k^{-n_0},\label{pighole}
 \end{eqnarray}
 for some absolute constant $C>0$. Now, the right hand side of \eqref{twe} with $\beta=\alpha-n_0$ and $\tau=\tau_k$ is bounded by a constant times the left hand side of \eqref{pighole} (where the constant may depend on $n_0$ but is independent of $\tau_k$). Thus by Proposition \ref{takenwholesale} with $\beta=\alpha-n_0$, $\tau=\tau_k$ and $\tau'=4\tau_k$, we have
 \begin{eqnarray}
 \int_{\mathcal{H}^+(\tau_k,4\tau_k)}(\partial_v\phi)^2+\iint_{\mathcal{D}(\tau_k,4\tau_k)}(\alpha-n_0)\chi_1(r)\log^{\alpha-n_0-1}(\frac{1}{\Omega})(\partial_v\phi)^2 & &\notag\\+\iint_{\mathcal{D}(\tau_k,4\tau_k)}\chi_2(r)\Omega^{-2}(\partial_u\phi)^2+\iint_{\mathcal{D}(\tau_k,4\tau_k)}\Omega^2((\partial_v\phi)^2+(\partial_u\phi)^2) 
 &\leq& C_{n_0}\cdot C\mathcal{I}_{n_0-1}\tau_k^{-n_0}.\label{s1}
 \end{eqnarray}
 Also, by the inductive hypothesis, namely \eqref{indarg} with $n=n_0-1$ and $\tau=\tau_k$, we have
 \begin{eqnarray}
 \sum_{j=0}^{n_0-1}(\alpha-j)\tau_k^j\iint_{\mathcal{D}(\tau_k,4\tau_k)}\chi_1(r)\log^{\alpha-j-1}(\frac{1}{\Omega})(\partial_v\phi)^2\hspace{25mm}\notag\\
 \leq \sum_{j=0}^{n_0-1}(\alpha-j)\tau_k^j\iint_{\mathcal{D}(\tau_k,\infty)}\chi_1(r)\log^{\alpha-j-1}(\frac{1}{\Omega})(\partial_v\phi)^2 \leq \mathcal{I}_{n_0-1}.\label{s2}
 \end{eqnarray}
 Multiplying \eqref{s1} by $\tau_k^{n_0}$ and summing the result with \eqref{s2} yields
 \begin{eqnarray}
 \tau_k^{n_0}\int_{\mathcal{H}^+(\tau_k,4\tau_k)}(\partial_v\phi)^2+\sum_{j=0}^{n_0}(\alpha-j)\tau_k^j\iint_{\mathcal{D}(\tau_k,4\tau_k)}\chi_1(r)\log^{\alpha-j-1}(\frac{1}{\Omega})(\partial_v\phi)^2\notag\\
 +\tau_k^{n_0}\iint_{\mathcal{D}(\tau_k,4\tau_k)}\chi_2(r)\Omega^{-2}(\partial_u\phi)^2+\tau_k^{n_0}\iint_{\mathcal{D}(\tau_k,4\tau_k)}\Omega^2\left((\partial_v\phi)^2+(\partial_u\phi)^2\right)\notag\\
 \leq C_{n_0}\cdot C\mathcal{I}_{n_0-1}+\mathcal{I}_{n_0-1}\notag\\
 =\tilde{C}_{n_0}\mathcal{I}_{n_0-1}.\hspace{17.5mm}\notag
 \end{eqnarray}

Now, $\tau_k\in[2^k,2^{k+1}]$ for each integer $k\geq0$ inplies $[2^{k+1},2^{k+2}]\subseteq [\tau_k,4\tau_k]$, and hence
 \begin{eqnarray}
 \tau_k^{n_0}\int_{\mathcal{H}^+(2^{k+1},2^{k+2})}(\partial_v\phi)^2+\sum_{j=0}^{n_0}(\alpha-j)\tau_k^j\iint_{\mathcal{D}(2^{k+1},2^{k+2})}\chi_1(r)\log^{\alpha-j-1}(\frac{1}{\Omega})(\partial_v\phi)^2\notag\\
 +\tau_k^{n_0}\iint_{\mathcal{D}(2^{k+1},2^{k+2})}\chi_2(r)\Omega^{-2}(\partial_u\phi)^2+\tau_k^{n_0}\iint_{\mathcal{D}(2^{k+1},2^{k+2})}\Omega^2\left((\partial_v\phi)^2+(\partial_u\phi)^2\right)\notag\\
 \leq\tilde{C}_{n_0}\mathcal{I}_{n_0-1}.\label{forg}
\end{eqnarray}

Summing up these bounds for $k\geq0$ and using Proposition \ref{takenwholesale} and the inductive hypothesis for the interval $\tau\in[1,2]$ yields the desired estimate \eqref{indarg} with $n=n_0$. To see this, consider for instance
\[\tau^{n_0}\int_{\mathcal{H}^+(\tau,\infty)}(\partial_v\phi)^2, \hspace{3mm}\tau\geq 1.\]
There exists a unique integer $l\geq 1$ such that $\tau\in[2^{l-1},2^l)$. If $l\geq 2$ $(\tau\geq2)$, then by \eqref{forg} 
\begin{eqnarray*}
\tau^{n_0}\int_{\mathcal{H}^+(\tau,\infty)}(\partial_v\phi)^2 &=&\tau^{n_0}\int_{\mathcal{H}^+(\tau,2^l)}(\partial_v\phi)^2 + \tau^{n_0}\sum_{m=l-1}^\infty \int_{\mathcal{H}^+(2^{m+1},2^{m+2})}(\partial_v\phi)^2\\
&\leq& \left(\frac{\tau}{\tau_{l-2}}\right)^{n_0}\tau_{l-2}^{n_0} \int_{\mathcal{H}^+(2^{l-1},2^{l})}(\partial_v\phi)^2 + \sum_{m=l-1}^\infty \left(\frac{\tau}{\tau_m}\right)^{n_0}\tau_m^{n_0}\int_{\mathcal{H}^+(2^{m+1},2^{m+2})}(\partial_v\phi)^2\\
&\leq& 4^{n_0}\tau_{l-2}^{n_0} \int_{\mathcal{H}^+(2^{l-1},2^{l})}(\partial_v\phi)^2+      \sum_{m=l-1}^\infty \left(\frac{2^l}{2^m}\right)^{n_0}\tau_m^{n_0}\int_{\mathcal{H}^+(2^{m+1},2^{m+2})}(\partial_v\phi)^2\\
&=&4^{n_0}\tau_{l-2}^{n_0} \int_{\mathcal{H}^+(2^{l-1},2^{l})}(\partial_v\phi)^2+
\sum_{m=l-1}^\infty \left(\frac{1}{2^{n_0}}\right)^{m-l}\tau_m^{n_0}\int_{\mathcal{H}^+(2^{m+1},2^{m+2})}(\partial_v\phi)^2\\
&\leq&4^{n_0}\tilde{C}_{n_0}\mathcal{I}_{n_0-1}+\sum_{m=l-1}^\infty \left(\frac{1}{2^{n_0}}\right)^{m-l}\tilde{C}_{n_0}\mathcal{I}_{n_0-1}\\
&=&C_{n_0}\mathcal{I}_{n_0-1}\\
 &=&C'_{n_0}(\alpha)\left(\int_1^\infty \left(1+\chi_1(r)\log^\alpha(\frac{1}{\Omega})\right)(\partial_v\phi)^2(-1,v)\,dv+\int_{\underline{C}_1^{int}\cap\{u\leq -1\}}  \big|\frac{\partial_u\phi}{\Omega}\big|^2\right),
\end{eqnarray*}
since (by the induction hypothesis)
\begin{equation*}
\mathcal{I}_{n_0-1}=C_{n_0-1}(\alpha)\left(\int_1^\infty \left(1+\chi_1(r)\log^\alpha(\frac{1}{\Omega})\right)(\partial_v\phi)^2(-1,v)\,dv+\int_{\underline{C}_1^{int}\cap\{u\leq -1\}}  \big|\frac{\partial_u\phi}{\Omega}\big|^2\right).
\end{equation*}
On the other hand, if $l=1$ $(\tau\in[1,2))$, by the previous case we have
\begin{eqnarray*}
\tau^{n_0}\int_{\mathcal{H}^+(\tau,\infty)}(\partial_v\phi)^2 &\leq&\tau^{n_0}\int_{\mathcal{H}^+(\tau,2)}(\partial_v\phi)^2 +2^{n_0}\int_{\mathcal{H}^+(2,\infty)}(\partial_v\phi)^2\\
&\leq&\tau^{n_0}\int_{\mathcal{H}^+(1,2)}(\partial_v\phi)^2 +2^{n_0}\int_{\mathcal{H}^+(2,\infty)}(\partial_v\phi)^2\\
&\leq &2^{n_0}\left(\int_1^\infty \left(1+\chi_1(r)\log^\alpha(\frac{1}{\Omega})\right)(\partial_v\phi)^2(-1,v)\,dv+\int_{\underline{C}_1^{int}\cap\{u\leq -1\}}  \big|\frac{\partial_u\phi}{\Omega}\big|^2\right)\\
& &+ \tilde{C}_{n_0}\mathcal{I}_{n_0-1}\\
 &=&C'_{n_0}(\alpha)\left(\int_1^\infty \left(1+\chi_1(r)\log^\alpha(\frac{1}{\Omega})\right)(\partial_v\phi)^2(-1,v)\,dv+\int_{\underline{C}_1^{int}\cap\{u\leq -1\}}  \big|\frac{\partial_u\phi}{\Omega}\big|^2\right),
\end{eqnarray*}
 where we estimate $\int_{\mathcal{H}^+(1,2)}(\partial_v\phi)^2$ using Proposition \ref{takenwholesale} (with $\tau=1, \tau'=2$ and $\beta=\alpha$).

We conclude that for any $\tau\geq1$ we have
\begin{eqnarray*}
\tau^{n_0}\int_{\mathcal{H}^+(\tau,\infty)}(\partial_v\phi)^2&\leq& C_{n_0}(\alpha)\left(\int_1^\infty \left(1+\chi_1(r)\log^\alpha(\frac{1}{\Omega})\right)(\partial_v\phi)^2(-1,v)\,dv+\int_{\underline{C}_1^{int}\cap\{u\leq -1\}}  \big|\frac{\partial_u\phi}{\Omega}\big|^2\right)\\
&:=&\mathcal{I}_{n_0}
\end{eqnarray*}
and we can treat the other terms on the left hand side of \eqref{indarg} (with $n=n_0$) in the same way. Thus we conclude that \eqref{indarg} holds with $n=n_0$. So by induction \eqref{indarg} holds for every positive integer $n\leq \alpha$. We have thus established  \eqref{preest}, as it is a direct consequence of \eqref{indarg} with $n=\alpha$.

In particular, it follows from \eqref{preest} that for any $\tau\geq1$
\begin{equation}\label{tauk}
\int_{\mathcal{H}^+(\tau,2\tau)}\tau^\alpha(\partial_v\phi)^2\leq C(\alpha)\left(\int_1^\infty \left(1+\chi_1(r)\log^\alpha(\frac{1}{\Omega})\right)(\partial_v\phi)^2(-1,v)\,dv+\int_{\underline{C}_1^{int}\cap\{u\leq -1\}}  \big|\frac{\partial_u\phi}{\Omega}\big|^2\right).
\end{equation}
Now, for $v\in[2^k,2^{k+1}]$, $\log(1+v)\geq\log(2^k)=Ck$, so
\begin{eqnarray*}
\int_{\mathcal{H}^+\cap\{v\geq1\}}\frac{v^\alpha}{\log^2(1+v)}(\partial_v\phi)^2&=&\int_{\mathcal{H}^+(1,2)}\frac{v^\alpha}{\log^2(1+v)}(\partial_v\phi)^2+\sum_{k=1}^\infty\int_{\mathcal{H}^+(2^k,2^{k+1})}\frac{v^\alpha}{\log^2(1+v)}(\partial_v\phi)^2\\
&\leq&c(\alpha) \int_{\mathcal{H}^+(1,2)} 1^\alpha (\partial_v\phi)^2+ C^{-2}\sum_{k=1}^\infty\int_{\mathcal{H}^+(2^k,2^{k+1})}\frac{v^\alpha}{k^2}(\partial_v\phi)^2\\
&\leq&c(\alpha) \int_{\mathcal{H}^+(1,2)} 1^\alpha (\partial_v\phi)^2+2^\alpha C^{-2}\sum_{k=1}^\infty\int_{\mathcal{H}^+(2^k,2^{k+1})}\frac{({2^k})^\alpha}{k^2}(\partial_v\phi)^2.
\end{eqnarray*}
But setting $\tau=2^k$ in \eqref{tauk} for each $k\geq0$, we deduce that 
\[\int_{\mathcal{H}^+\cap\{v\geq1\}}\frac{v^\alpha}{\log^2(1+v)}(\partial_v\phi)^2\]
is bounded by
\begin{eqnarray*}
\left(c(\alpha) +2^\alpha C^{-2}\sum_{k=1}^\infty\frac{1}{k^2} \right)\cdot C(\alpha)\left(\int_1^\infty \left(1+\chi_1(r)\log^\alpha(\frac{1}{\Omega})\right)(\partial_v\phi)^2(-1,v)\,dv+\int_{\underline{C}_1^{int}\cap\{u\leq -1\}}  \big|\frac{\partial_u\phi}{\Omega}\big|^2\right)\\
=C'(\alpha)\left(\int_1^\infty \left(1+\chi_1(r)\log^\alpha(\frac{1}{\Omega})\right)(\partial_v\phi)^2(-1,v)\,dv+\int_{\underline{C}_1^{int}\cap\{u\leq -1\}}  \big|\frac{\partial_u\phi}{\Omega}\big|^2\right).
\end{eqnarray*}
Finally, we note that
\[\int_1^\infty \left(1+\chi_1(r)\log^\alpha(\frac{1}{\Omega})\right)(\partial_v\phi)^2(-1,v)\,dv\leq C\int_1^\infty v^\alpha(\partial_v\phi)^2(-1,v)\,dv\]
because $\displaystyle{ \lim_{v\to\infty}\frac{v}{\log\frac{1}{\Omega(-1,v)}}>0}$, and so
\[\int_{\mathcal{H}^+\cap\{v\geq1\}}\frac{v^\alpha}{\log^2(1+v)}(\partial_v\phi)^2\leq C(\alpha)\left(\int_1^\infty v^\alpha(\partial_v\phi)^2(-1,v)\,dv+\int_{\underline{C}_1^{int}\cap\{u\leq -1\}}  \big|\frac{\partial_u\phi}{\Omega}\big|^2\right),\]
as desired.
\end{proof}

\subsubsection{$L^p$ Estimates}\label{star10}

We are now ready to prove Theorem \ref{interior reduction}. The proof considers an operator which maps smooth, spherically symmetric  ``data'' for the wave equation on the Cauchy horizon transversal null hypersurface  $C^{int}_{-1}\cap\{v\geq 1\}$ to a coresponding solution of the wave equation on the event horizon $\mathcal{H}^+\cap\{v\geq1\}$. The $L^1$ and $L^2$-type estimates established in the previous sections mean precisely that this operator is bounded in an appropriate $L^1$ and $L^2$ sense, and we can use real interpolation to deduce that the operator must also be bounded in an $L^p$ sense, giving the $L^p$-type estimate we desire. We make use of arguments from Section 5.5.1 of \cite{BL} to calculate the interpolation spaces needed. Note that we will use the material about the K-method of real interpolation from Appendix \ref{Interpolation Theory} extensively in the coming proof, and so refer the unfamiliar reader to that Appendix now.

\begin{proof}[Proof of Theorem \ref{interior reduction}]
Let $\phi$ be a  solution to \eqref{waveequation} with smooth, spherically symmetric data on $S$ satisfying \eqref{ic2}. Let $\phi_1$ and $\phi_2$ be smooth solutions to the wave equation in spherical symmetry \eqref{extwe} such that 
\begin{equation}\label{tphi}
\partial_u\phi_1\vert_{\underline{C}_1^{int}\cap\{u\leq-1\}}=0, \hspace{5mm} \partial_v\phi_1\vert_{C_{-1}^{int}\cap\{v\geq1\}}=\partial_v\phi\vert_{C_{-1}^{int}\cap\{v\geq1\}}
\end{equation}
and
\begin{equation*}\partial_u\phi_2\vert_{\underline{C}_1^{int}\cap\{u\leq-1\}}=\partial_u\phi\vert_{\underline{C}_1^{int}\cap\{u\leq-1\}}, \hspace{5mm} \partial_v\phi_2\vert_{C_{-1}^{int}\cap\{v\geq1\}}=0.\end{equation*}
Note that $\phi_1$ and $\phi_2$ are determined up to a constant on $(\underline{C}_1\cap\{u\leq-1\})\cup(C_{-1}\cap\{v\geq1\})$ and hence $\phi_1$ and $\phi_2$ are determined up to a constant on $\{v\geq1\}\cap\{-\infty\leq u\leq-1\}$. Consequently $\partial_v\phi_1$ and $\partial_v\phi_2$ are uniquely determined on $\mathcal{H}^+\cap\{v\geq1\}$. Furthermore, by linearity and uniqueness, we have $\phi=\phi_1+\phi_2+c$ on $\{v\geq1\}\cap\{-\infty\leq u\leq-1\}$ and hence,
\[\partial_v\phi=\partial_v\phi_1+\partial_v\phi_2  \mbox{ on }\{v\geq1\}\cap\{-\infty\leq u\leq-1\}.\]

We claim that there exist $C_1=C_1(\alpha,\varepsilon,p)>0$ and $C_2=C_2(\alpha,\varepsilon,p,D,\phi)>0$ such that
\begin{equation}\label{phi1}
\int_{\mathcal{H}^+\cap\{v\geq1\}} v^{(\alpha-\varepsilon)\frac{p\theta}{2}}|\partial_v\phi_1|^p \leq C_1\int_1^\infty v^{\frac{\alpha p\theta}{2}}|\partial_v\phi|^p(-1,v)\,dv
\end{equation}
and 
\begin{equation}\label{phi2}
\int_{\mathcal{H}^+\cap\{v\geq1\}} v^{(\alpha-\varepsilon)\frac{p\theta}{2}}|\partial_v\phi_2|^p \leq C_2\int_{\underline{C}_1^{int}\cap\{u\leq -1\}}\Omega^{-p\theta}|\partial_u\phi|^p.
\end{equation}
Once we have proved these two estimates, the result follows immediately because
\begin{eqnarray*}
\int_{\mathcal{H}^+\cap\{v\geq1\}} v^{(\alpha-\varepsilon)\frac{p\theta}{2}}|\partial_v\phi|^p &\leq& 2^P\left(\int_{\mathcal{H}^+\cap\{v\geq1\}} v^{(\alpha-\varepsilon)\frac{p\theta}{2}}|\partial_v\phi_1|^p +\int_{\mathcal{H}^+\cap\{v\geq1\}} v^{(\alpha-\varepsilon)\frac{p\theta}{2}}|\partial_v\phi_2|^p \right)\\
&\leq& 2^pC_1\int_1^\infty v^{\frac{\alpha p\theta}{2}}|\partial_v\phi|^p(-1,v)\,dv+2^pC_2\int_{\underline{C}_1^{int}\cap\{u\leq -1\}}\Omega^{-p\theta}|\partial_u\phi|^p\\
&\leq& C  \left(\int_1^\infty v^{\frac{\alpha p\theta}{2}}|\partial_v\phi|^p(-1,v)\,dv+\int_{\underline{C}_1^{int}\cap\{u\leq -1\}}\Omega^{-p\theta}|\partial_u\phi|^p\right),
\end{eqnarray*}
with $C=C(\alpha,\varepsilon,p,D,\phi)$, as desired.

In order to prove the claim, we first prove \eqref{phi2}. Let $q=2/p>1$ and suppose $1/q+1/q'=1$. Set $\beta=(\alpha-\varepsilon)\frac{p\theta}{2}$. Then
\begin{eqnarray}
\int_{\mathcal{H}^+\cap\{v\geq1\}}v^{(\alpha-\varepsilon)\frac{p\theta}{2}}|\partial_v\phi_2|^p &=&\int_{\mathcal{H}^+\cap\{v\geq1\}}v^{\beta}|\partial_v\phi_2|^p\notag\\
&\leq&\left(\int_{\mathcal{H}^+\cap\{v\geq1\}}v^{(\beta+1+\varepsilon_0)q}|\partial_v\phi_2|^{pq}\right)^{1/q} \left(\int_{\mathcal{H}^+\cap\{v\geq1\}}v^{-(1+\varepsilon_0)q'}\right)^{1/q'}\notag\\
&=&C(p) \left(\int_{\mathcal{H}^+\cap\{v\geq1\}}v^{(\beta+1+\varepsilon_0)q}|\partial_v\phi_2|^{2}\right)^{p/2}.\notag
\end{eqnarray}
Now there exists a positive integer $\gamma$ such that $(\gamma-\varepsilon)> (\beta+1+\varepsilon_0)q$, and so by the $L^2$ estimate (Proposition \ref{intl2}) applied to $\phi_2$ we get
\begin{eqnarray}
\int_{\mathcal{H}^+\cap\{v\geq1\}}v^{(\alpha-\varepsilon)\frac{p\theta}{2}}|\partial_v\phi_2|^p &\leq&C(p)\left(\int_{\mathcal{H}^+\cap\{v\geq1\}}v^{(\gamma-\varepsilon)} |\partial_v\phi_2|^{2}\right)^{p/2}\notag\\
&\leq&C'\left(\int_{C_{-1}\cap\{v\geq1\}}v^{\gamma }|\partial_v\phi_2|^2 +\int_{\underline{C}_1\cap\{u\leq-1\}}\big|\frac{\partial_u\phi_2}{\Omega}\big|^2\right)^{p/2}\notag\\
&=&C'\left(\int_{\underline{C}_1\cap\{u\leq-1\}}\big|\frac{\partial_u\phi}{\Omega}\big|^2\right)^{p/2},\label{one}
\end{eqnarray}
where $C'=C'(\alpha,\varepsilon,p)$.
But 
\begin{eqnarray}
\left(\int_{\underline{C}_1\cap\{u\leq-1\}}\big|\frac{\partial_u\phi}{\Omega}\big|^2\right)^{p/2}&=&\left(\int_{\underline{C}_1\cap\{u\leq-1\}}\frac{|\partial_u\phi|^p}{\Omega^{p\theta}}\big|\frac{\partial_u\phi}{\Omega^2}\big|^{2-p} \Omega^{2+p(\theta-2)}\right)^{p/2}\notag\\
&\leq&D^{(2-p)p/2}\left(\int_{\underline{C}_1\cap\{u\leq-1\}}\frac{|\partial_u\phi|^p}{\Omega^{p\theta}} \Omega^{2+p(\theta-2)}\right)^{p/2},\notag
\end{eqnarray}
where we have used \eqref{ic2}. Also,
\[\frac{1}{p}=\frac{1-\theta}{1}+\frac{\theta}{2}\]
and hence $p=\frac{2}{2-\theta}$. Thus $\Omega^{2+p(\theta-2)}=\Omega^0=1$, and so
\begin{eqnarray}
\left(\int_{\underline{C}_1\cap\{u\leq-1\}}\big|\frac{\partial_u\phi}{\Omega}\big|^2\right)^{p/2}&\leq&
C(p,D)\left(\int_{\underline{C}_1\cap\{u\leq-1\}}\frac{|\partial_u\phi|^p}{\Omega^{p\theta}} \right)^{p/2}.\notag
\end{eqnarray}
So by \eqref{one}
\begin{eqnarray}
\int_{\mathcal{H}^+\cap\{v\geq1\}}v^{(\alpha-\varepsilon)\frac{p\theta}{2}}|\partial_v\phi_2|^p &\leq&
C(\alpha,\varepsilon,p,D)\left(\int_{\underline{C}_1\cap\{u\leq-1\}}\frac{|\partial_u\phi|^p}{\Omega^{p\theta}} \right)^{p/2},\label{three}
\end{eqnarray}
We note that $\int_{\underline{C}_1\cap\{u\leq-1\}}\frac{|\partial_u\phi|^p}{\Omega^{p\theta}}<+\infty$ because $p=\frac{2}{2-\theta}$ so
\begin{eqnarray*}
\int_{\underline{C}_1\cap\{u\leq-1\}}\frac{|\partial_u\phi|^p}{\Omega^{p\theta}}=\int_{\underline{C}_1\cap\{u\leq-1\}} \left| \frac{\partial_u\phi}{\Omega^2}\right|^p\Omega^{p(2-\theta)}
&\leq& D^p\int_{\underline{C}_1\cap\{u\leq-1\}}\Omega^{(2-\theta)p}\\
&=&D^p\int_{-\infty}^{-1}(-\partial_ur)(u,1)\,du\\ &<&+\infty.
\end{eqnarray*}
If $\int_{\underline{C}_1\cap\{u\leq-1\}}\frac{|\partial_u\phi|^p}{\Omega^{p\theta}} =0$, then \eqref{phi2} is true by \eqref{three}. Otherwise, we have
\begin{eqnarray*}
\int_{\mathcal{H}^+\cap\{v\geq1\}}v^{(\alpha-\varepsilon)\frac{p\theta}{2}}|\partial_v\phi_2|^p &\leq&
\frac{C(\alpha,\varepsilon,p,D)}{\left(\int_{\underline{C}_1\cap\{u\leq-1\}}\frac{|\partial_u\phi|^p}{\Omega^{p\theta}} \right)^{1-p/2}  }\int_{\underline{C}_1\cap\{u\leq-1\}}\frac{|\partial_u\phi|^p}{\Omega^{p\theta}} \\
&=&C(\alpha,\varepsilon,p,D,\phi)\int_{\underline{C}_1\cap\{u\leq-1\}}\frac{|\partial_u\phi|^p}{\Omega^{p\theta}}, 
\end{eqnarray*}
so \eqref{phi2} holds in this case also.

We now turn our attention to proving \eqref{phi1}. The idea is to use the K-method of real interpolation (see Appendix \ref{Interpolation Theory}) to interpolate between the bounds achieved in Proposition \ref{intl1} and Proposition \ref{intl2} for solutions of the form $\phi_1$ (i.e, for solutions with $\partial_u\phi_1=0$ on $\underline{C}_1\cap\{u\leq-1\}$) to deduce the desired estimate.  For clarity, set $p_0=1$ and $p_1=2$, so that
\[\frac{1}{p}=\frac{1-\theta}{p_0}+\frac{\theta}{p_1}.\]
We split the argument into several steps.

\noindent\textbf{Step 1:} First of all, we need to define the compatible couples we wish to interpolate between. Given a postitve, measurable function $w:[1,\infty)\to(0,\infty)$ and $q\geq1$, we set 
\[A_q(w):=L^q_w([1,\infty)).\]
In particular, 
\[\|f\|_{A_q(w)}=\left(\int_1^\infty w(v)|f|^{q}(v)\,dv\right)^{1/q}.\]
Then, given positive, measurable functions $w_0,w_1:[1,\infty)\to(0,\infty)$,  define
\begin{eqnarray*}
A_0=A_{p_0}(w_0)\hspace{7mm}  \mbox{ and }\hspace{7mm}  A_1=A_{p_1}(w_1)
\end{eqnarray*}
so
\begin{eqnarray*}
\|f\|_{A_j}&=&\left(\int_1^\infty w_j(v)|f|^{p_j}(v)\,dv\right)^{1/p_j}, \hspace{3mm} j=0,1.
\end{eqnarray*}

Similarly for a postitve, measurable function $\omega:[1,\infty)\to(0,\infty)$ and $q\geq1$, we set 
\[B_q(\omega):=L^q_\omega([1,\infty)),\]
so
\[\|f\|_{B_q(\omega)}=\left(\int_1^\infty \omega(v)|f|^{q}(v)\,dv\right)^{1/q}.\]
Then, given  positive, measurable functions $\omega_0,\omega_1:[1,\infty)\to(0,\infty)$,  define
\begin{eqnarray*}
B_0=B_{p_0}(\omega_0) \hspace{7mm}\mbox{ and } \hspace{7mm} B_1=B_{p_1}(\omega_1)
\end{eqnarray*}
so
\begin{eqnarray*}
\|f\|_{B_j}&=&\left(\int_1^\infty \omega_j(v)|f|^{p_j}(v)\,dv\right)^{1/p_j}, \hspace{3mm} j=0,1.
\end{eqnarray*}
Finally, let $\overline{A}$ and $\overline{B}$ denote the compatible couples $\overline{A}=(A_0,A_1)$ and $\overline{B}=(B_0,B_1)$.

\noindent \textbf{Step 2:} 
Suppose $S:\overline{A}\to\overline{B}$ is a linear operator. Recall from Section \ref{INS} this means (with slight abuse of notation) that, $S:A_0+A_1\to B_0+B_1$ and moreover, $S:A_j\to B_j$ and is bounded, with norm $\|S\|_j$ say, for $j=0,1$.
Recalling that
\[\frac{1}{p}=\frac{1-\theta}{p_0}+\frac{\theta}{p_1}=\frac{1-\theta}{1}+\frac{\theta}{2},\]
it follows from Theorem \ref{kmethod} (with $q=p$ and $\theta=\theta$)  that
\[S:\overline{A}_{\theta,p}\to\overline{B}_{\theta,p} \hspace{3mm}\mbox{ is bounded with norm }\hspace{3mm}\|S\|_{\theta,p}\leq\|S\|_0^{1-\theta}\|S\|_1^{\theta}.\]
In order to make use of this result, we will need to identify appropriate weights $w_0,w_1,\omega_0$ and $\omega_1$ and a bounded operator $T: A_j\to B_j$. This is where the $L^1$ and $L^2$ estimates proved in the previous sections will come in. We will also need to understand the spaces $\overline{A}_{\theta,p}$ and $\overline{B}_{\theta,p}$ and this is our next task.

\noindent\textbf{Step 3:} 
We show that
\[\overline{A}_{\theta,p}  :=(A_0,A_1)_{\theta,p}=(A_{p_0}(w_0), A_{p_1}(w_1))_{\theta,p}= A_p(w),\]
with equivalent norms, where
\[w=w_0^{p\frac{1-\theta}{p_0}}w_1^{p\frac{\theta}{p_1}}.\]

The proof of this fact is sketched in Theorem 5.5.1 of \cite{BL}, though we present the proof in full detail. Indeed, set $\rho_0=p_0$, $\rho_1=p_1$, $\eta=\frac{\theta p}{p_1}\in(0,1)$ and $r=1$. Now $1-\eta=1-\frac{\theta p}{p_1}=(\frac{1}{p}-\frac{\theta}{p_1})p =\frac{1-\theta}{p_0}p$, and it follows that $\rho:=(1-\eta)\rho_0+\eta\rho_1=p$ and $q:=\rho r=p$. Then by the Power Theorem (Theorem \ref{power}) we have
\begin{equation}\label{power2}
\left((A_0,A_1)_{\theta,p}\right)^{p}=(A_0^{p_0},A_1^{p_1})_{\eta,1}
\end{equation}
with equivalent quasinorms.
For brevity we set $X=(A_0^{p_0},A_1^{p_1})_{\eta,1}$. For $f\in A_0^{p_0}+A_1^{p_1}$, we have

\begin{eqnarray}
\|f\|_X &=&\int_0^\infty\left(t^{-\eta}K(t,f;A_0^{p_0},A_1^{p_1}   )\right)^1\,\frac{dt}{t}\notag\\
&=&\int_0^\infty t^{-\eta} (\inf_{\substack{f=f_0+f_1\\f_i\in L^{p_i}_{w_i}([1,\infty))}}\left[\|f_0\|_{A_0^{p_0}}+t\|f_1\|_{A_1^{p_1}}\right])\frac{dt}{t} \notag\\
&=&\int_0^\infty t^{-\eta} \inf_{\substack{f=f_0+f_1\\f_i\in L^{p_i}_{w_i}([1,\infty))}}\left[\int_1^\infty w_0|f_0|^{p_0}(v)\,dv +t\left(  \int_1^\infty w_1|f_1|^{p_1}(v)\,dv \right)\right]\,\frac{dt}{t}\notag\\
&=&\int_1^\infty\left(\int_0^\infty t^{-\eta}\inf_{\substack{f=f_0+f_1\\f_i\in L^{p_i}_{w_i}}}(w_0|f_0|^{p_0}(v)+tw_1|f_1|^{p_1}(v))\,\frac{dt}{t}\right)\,dv.\label{switch}
\end{eqnarray}
The justification for \eqref{switch} is quite technical and is given in Lemma \ref{interplem} in Appendix \ref{Switch}.

We may write the  right hand side of \eqref{switch} as 
\begin{eqnarray}
\int_1^\infty\left(\int_0^\infty t^{-\eta}\inf_{\substack{f=f_0+f_1\\f_i\in L^{p_i}_{w_i}}}(w_0|f_0|^{p_0}(v)+tw_1|f_1|^{p_1}(v))\,\frac{dt}{t}\right)\,dv\hspace{30mm}\notag\\
=\int_{\{f(v)\not=0\}}\left(\int_0^\infty t^{-\eta}\inf_{\substack{f=f_0+f_1\\f_i\in L^{p_i}_{w_i}}}(w_0|f_0|^{p_0}(v)+tw_1|f_1|^{p_1}(v))\,\frac{dt}{t}\right)\,dv\notag\\
+\int_{\{f(v)=0\}}\left(\int_0^\infty t^{-\eta}\inf_{\substack{f=f_0+f_1\\f_i\in L^{p_i}_{w_i}}}(w_0|f_0|^{p_0}(v)+tw_1|f_1|^{p_1}(v))\,\frac{dt}{t}\right)\,dv.\label{split'}
\end{eqnarray}
But for $v$ such that $f(v)=0$, 
\[\inf_{\substack{f=f_0+f_1\\f_i\in L^{p_i}_{w_i}}}(w_0|f_0|^{p_0}(v)+tw_1|f_1|^{p_1}(v))=0,\]
this infimum being attained for any $f_0\in L^{p_0}_{w_0}$ and $f_1\in L^{p_1}_{w_1}$ with $f_0+f_1=f$ and $f_0(v)= f_1(v)= 0$. Now set 
\begin{equation}\label{F}
F(x)=\inf_{y_0+y_1=1}(|y_0|^{p_0}+x|y_1|^{p_1}).
\end{equation}
Then for $v$ such that $f(v)\not=0$,
\begin{eqnarray*}
\inf_{\substack{f=f_0+f_1\\f_i\in L^{p_i}_{w_i}}}(w_0|f_0|^{p_0}(v)+tw_1|f_1|^{p_1}(v))&=&\inf_{\substack{1=\frac{f_0}{f}+\frac{f_1}{f}\\f_i\in L^{p_i}_{w_i}}}(w_0|f_0|^{p_0}(v)+tw_1|f_1|^{p_1}(v))\\
&=&\inf_{\substack{1=\frac{f_0}{f}+\frac{f_1}{f}\\f_i\in L^{p_i}_{w_i}}}w_0|f|^{p_0}\left(\big|\frac{f_0}{f}\big|^{p_0}+t\frac{w_1}{w_0}|f|^{p_1-p_0}\big|\frac{f_1}{f}\big|^{p_1}\right)(v)\\
&=&w_0|f|^{p_0}(v)\cdot F\left(t\frac{w_1}{w_0}|f|^{p_1-p_0}(v)\right).
\end{eqnarray*}
Thus
\begin{eqnarray*}
\int_{\{f(v)\not=0\}}\left(\int_0^\infty t^{-\eta}\inf_{\substack{f=f_0+f_1\\f_i\in L^{p_i}_{w_i}}}(w_0|f_0|^{p_0}(v)+tw_1|f_1|^{p_1}(v))\,\frac{dt}{t}\right)\,dv\hspace{34mm}\\
=\int_{\{f(v)\not=0\}}w_0|f|^{p_0}(v)\int_0^\infty t^{-\eta}F\left(t\frac{w_1}{w_0}|f|^{p_1-p_0}(v)\right)\,\frac{dt}{t}\,dv\hspace{14mm}\\
=\int_{\{f(v)\not=0\}}w_0|f|^{p_0}(v)\int_0^\infty s^{-\eta}F(s)\left( \frac{w_0}{w_1}|f|^{p_0-p_1}(v)\right)^{-\eta}\,\frac{ds}{s}\,dv\hspace{6mm}\\
=\int_0^\infty s^{-\eta}F(s)\,\frac{ds}{s}\cdot \int_{\{f(v)\not=0\}}|f|^{p_0(1-\eta)+p_1\eta}w_0^{1-\eta}w_1^\eta(v)\,dv\hspace{14mm}\\
= \int_0^\infty s^{-\eta}F(s)\,\frac{ds}{s}\cdot \int_{\{f(v)\not=0\}}|f|^{p} w_0^{p\frac{(1-\theta)}{p_0}}w_1^{p\frac{\theta}{p_1}}\,dv\hspace{27mm}\\
=C\int_{\{f(v)\not=0\}}|f|^{p}w(v)\,dv\hspace{62.5mm}\\
=C\int_1^\infty w|f|^{p}(v)\,dv.\hspace{69mm}
\end{eqnarray*}
Note that $\int_0^\infty s^{-\eta}F(s)\,\frac{ds}{s}$ converges because $F(s)\leq 1$.
Thus, by \eqref{split'}
\[\int_1^\infty\left(\int_0^\infty t^{-\eta}\inf_{\substack{f=f_0+f_1\\f_i\in L^{p_i}_{w_i}}}(w_0|f_0|^{p_0}(v)+tw_1|f_1|^{p_1}(v))\,\frac{dt}{t}\right)\,dv=C\int_1^\infty w|f|^{p}(v)\,dv,\]
and so by \eqref{switch}
\[\|f\|_X=C\left(\int_1^\infty w|f|^{p}(v)\,dv\right)=\|f\|_{A_p(w)}^{p}.\]
In particular, $\|f\|_X<+\infty \Longleftrightarrow \|f\|_{A_p(w)^{p}}<+\infty$, so $X=A_p(w)^{p}$ and they have equivalent quasinorms. Thus, by \eqref{power2} 
\[\left((A_0,A_1)_{\theta,p}\right)^{p}=A_p(w)^{p},\]
with equivalent quasinorms, and so
\[(A_0,A_1)_{\theta,p}=A_p(w)\]
with equivalent norms, as desired.

\noindent\textbf{Step 4:} 
Using an argument identical to that in the previous step, we deduce
\[\overline{B}_{\theta,p}:=(B_0,B_1)_{\theta,p}=(B_{p_0}(\omega_0),B_{p_1}(\omega_1))_{\theta,p}=B_p(\omega)\]
with equivalent norms, where 
\[\omega=\omega_0^{p\frac{1-\theta}{p_0}}\omega_1^{p\frac{\theta}{p_1}}.\]
To avoid repetition, we omit the details.

\noindent \textbf{Step 5:} We now fix the weights by setting $w_0(v)=1$, $w_1(v)=v^\alpha$, $\omega_0(v)=1$ and $\omega_1(v)=v^{\alpha-\varepsilon}$, where $\alpha>1$ is a positive integer and $0<\varepsilon<1$. So 
\[A_0=L^1([1,\infty)), \hspace{7mm} A_1=L^2_{v^\alpha}([1,\infty))\]
and
\[B_0=L^1([1,\infty)), \hspace{7mm} B_1=L^2_{v^{\alpha-\varepsilon}}([1,\infty)).\]

 In order to construct a linear operator $T:\overline{A}\to\overline{B}$, we first construct two bounded linear operators
\[T_0:A_0\to B_0 \hspace{5mm}\mbox{ and } \hspace{5mm} T_1:A_1\to B_1.\]

Before defining these operators, we note that given smooth and  spherically symmetric data\\ $(\partial_v\phi\vert_{C_{-1}^{int}\cap\{v\geq1\}},\partial_u\phi\vert_{\underline{C}_1^{int}\cap\{u\leq-1\}})$ for the wave equation \eqref{waveequation}, the solution $\phi$ is determined up to a constant in $\{v\geq1\}\cap\{u\leq-1\}$ (as it is determined up to a constant on $C_{-1}^{int}\cap\{v\geq1\}$ and $\underline{C}_1^{int}\cap\{u\leq-1\}$), and hence $\partial_v\phi$ is uniquely determined on $\mathcal{H}^+\cap\{v\geq1\}$.

\noindent \textbf{Step 5a:} We define the operator $T_0:A_0\to B_0$. We first define $T_0$ for smooth functions $f\in A_0$ and then use a limiting process to extend $T_0$ to the entire space $A_0$. 

Indeed, given a smooth $f\in A_0$ we set $T_0f=\partial_v\varphi\vert_{\mathcal{H}^+\cap\{v\geq1\}}$ where $\varphi$ is a solution of the (spherically symmetric) wave equation with $\partial_v\varphi\vert_{C_{-1}\cap\{v\geq1\}}=f$ and $\partial_u\varphi\vert_{\underline{C}_1\cap\{u\leq-1\}}=0$. By the remarks directly above $T_0f$ is uniquely determined and by the $L^1$ estimates (Proposition \ref{intl1}  \footnote{We emphasize that we must assume the parameter restriction $\frac{2\sqrt e}{e+1}<Q<1$ in order to make use of the interior $L^1$ estimates, and this is the (only) reason for the parameter restriction in the interior $L^P$ estimates.})  applied to  $\varphi$,   $T_0f\in B_0$. Indeed,
\begin{eqnarray*}
\|T_0f\|_{B_0}=\int_1^\infty |T_0f|(v)\,dv=\int_{\mathcal{H}^+\cap\{v\geq1\}} |\partial_v\varphi|
&\leq&C\left(\int_1^\infty|\partial_v\varphi|(-1,v)\,dv+\int_{\underline{C}_1^{int}\cap\{u\leq-1\}}|\partial_u\varphi|\right).\\
&=&C\int_1^\infty |f(v)|\,dv =C\|f\|_{A_0}<\infty.
\end{eqnarray*}
Moreover, by uniqueness and linearity of solutions to the wave equation, $T_0$ acts linearly on the smooth functions in $A_0$. See Figure \ref{action} below for a diagramatic representation of the action of $T_0$ on smooth functions of $A_0$.

Now suppose $f\in A_0$ is not smooth. By density of smooth functions in $A_0=L^1$, there is a sequence of smooth functions $\{f_n\}$ in $A_0$ such that $f_n\to f$ in $A_0$. Thus $\{f_n\}$ is a Cauchy sequence in $A_0$, $T_0f_n\in B_0$ is defined for each $n$ by above and by the $L^1$ estimates (Proposition \ref{intl1}) $\{T_0f_n\}$ is a Cauchy sequence in $B_0$. But $B_0=L^1$ is complete and hence $\{T_0f_n\}$ has a limit in $B_0$. We define 
\[T_0f:=\lim_{n\to\infty}T_0f_n \hspace{2mm}\mbox{ in } B_0.\]

It is easy to check that $T_0f$ is well-defined. Indeed, if $\{f_n\}$ and $\{g_n\}$ are both sequences of smooth functions in $A_0$ which converge to $f$ in $A_0$, then by Proposition \ref{intl1}
\[
\|T_0f_n-T_0g_n\|_{B_0}=\|T_0(f_n-g_n)\|_{B_0}\leq C\|f_n-g_n\|_{A_0}\to0,
\]
so it follows that $T_0f$ is uniquely determined. So $T_0$ is well-defined. It is trivial to check that $T_0$ is linear. Furthermore, $T_0$ is bounded: if $f\in A_0$ is smooth, it follows immediately from Proposition \ref{intl1} that $\|T_0f\|_{B_0}\leq C\|f\|_{A_0}$, and if $f\in A_0$ is not smooth, then for any sequence of smooth functions $\{f_n\}$ in $A_0$ converging to $f$ in $A_0$,
\[
\|T_0f\|_{B_0}=\lim_{n\to\infty}\|T_0f_n\|_{B_0}\leq \lim_{n\to\infty} C\|f_n\|_{A_0} =C\|f\|_{A_0}.\]
So we have constructed the desired bounded linear operator $T_0:A_0\to B_0$ and moreover $\|T_0\|\leq C$, where $C$ is the constant from Proposition \ref{intl1}.

\begin{figure}[h]
\centering
\includegraphics[scale=0.5]{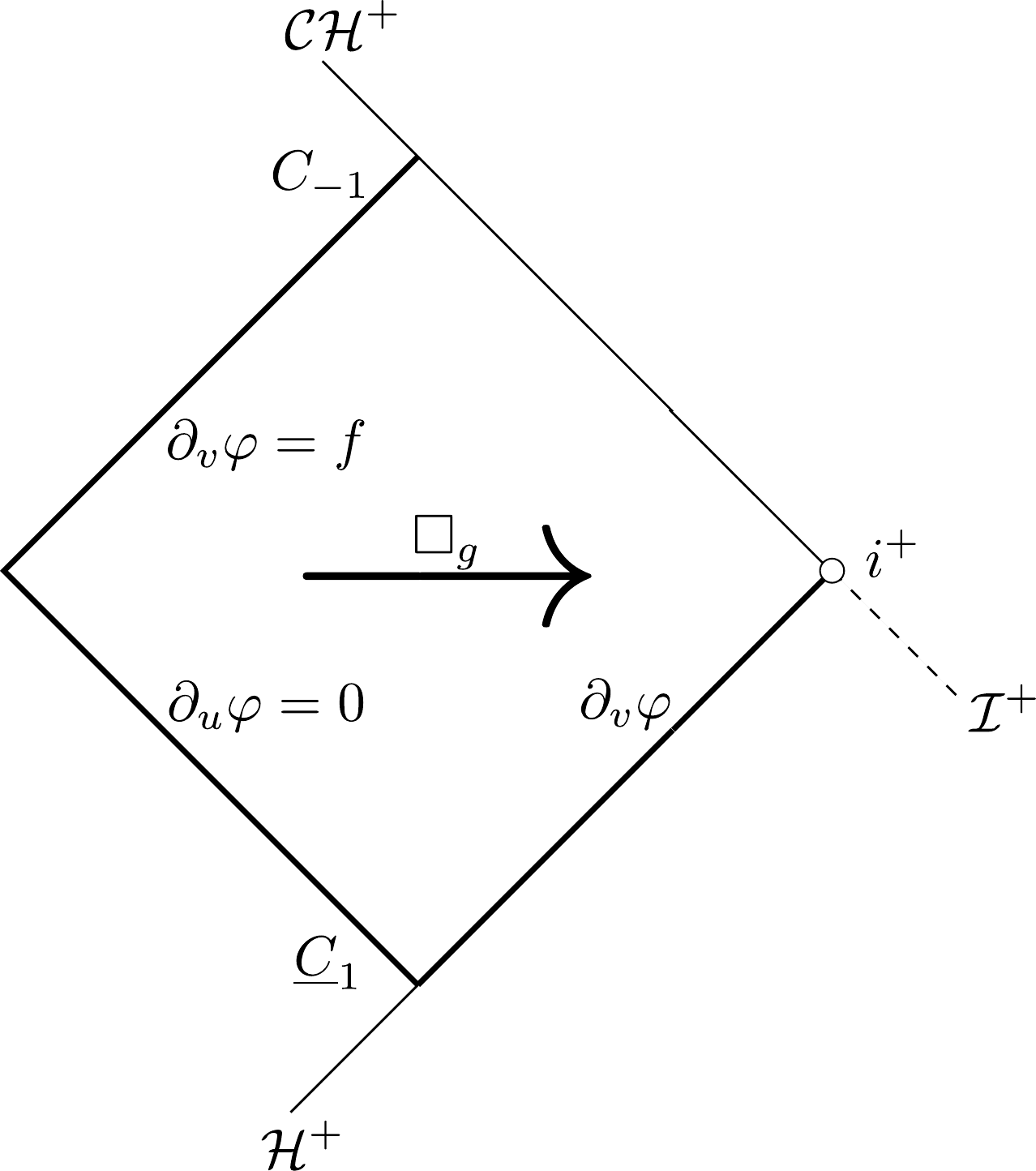}
\caption{Action of the operators $T_0$ and $T_1$ on smooth functions in their domains.}
\label{action}
\end{figure}

\noindent \textbf{Step 5b:} We similarly define the operator $T_1:A_1\to B_1$. Again, we first define $T_1$ for smooth functions $f\in A_1$ and then use a limiting process to extend $T_1$ to the entire space $A_1$. 

Given a smooth $f\in A_1$, we set $T_1f=\partial_v\varphi\vert_{\mathcal{H}^+\cap\{v\geq1\}}$, where $\varphi$ is a solution of the (spherically symmetric) wave equation with $\partial_v\varphi\vert_{C_{-1}\cap\{v\geq1\}}=f$ and $\partial_u\varphi\vert_{\underline{C}_1\cap\{u\leq-1\}}=0$. By the remarks at the start of Step 5 $T_1f$ is uniquely determined and by the $L^2$ estimates (Proposition \ref{intl2}) $T_1f\in B_1$. Indeed, applying Proposition \ref{intl2} to $\varphi$ yields
\begin{eqnarray*}
\|T_1f\|_{B_1}^2=\int_1^\infty v^{\alpha-\varepsilon} |T_1f|^2(v)\,dv &=&\int_{\mathcal{H}^+\{v\geq1\}} v^{\alpha-\varepsilon}(\partial_v\varphi)^2\\
&\leq& C(\alpha,\varepsilon)\left(\int_1^\infty v^\alpha(\partial_v\varphi)^2(-1,v)\,dv+\int_{\underline{C}_1^{int}\cap\{u\leq -1\}}  \big|\frac{\partial_u\varphi}{\Omega}\big|^2\right)\\
&=&C(\alpha,\varepsilon)\int_1^\infty v^\alpha |f(v)|^2\,dv =C(\alpha,\varepsilon)\|f\|_{A_1}^2<\infty.
\end{eqnarray*}
Moreover, by uniqueness and linearity of solutions to the wave equation, $T_1$ acts linearly on the smooth functions in $A_1$.

Now suppose $f\in A_1$ is not smooth. By density of smooth functions in $A_1=L^2_{v^\alpha}([1,\infty))$, there is a sequence of smooth functions $\{f_n\}$ in $A_1$ such that $f_n\to f$ in $A_1$. Thus $\{f_n\}$ is a Cauchy sequence in $A_1$, $T_1f_n\in B_1$ is defined for each $n$ by above and by the $L^2$ estimates (Proposition \ref{intl2}) $\{T_1f_n\}$ is a Cauchy sequence in $B_1$. But $B_1=L^2_{v^{\alpha-\varepsilon}}([1,\infty))$ is complete and hence $\{T_1f_n\}$ has a limit in $B_1$. We define 
\[T_1f:=\lim_{n\to\infty}T_1f_n \hspace{2mm}\mbox{ in } B_1.\]

As before $T_1f$ is well-defined. Indeed, if $\{f_n\}$ and $\{g_n\}$ are both sequences of smooth functions in $A_1$ which converge to $f$ in $A_1$, then by Proposition \ref{intl2}
\[
\|T_1f_n-T_1g_n\|_{B_1}=\|T_1(f_n-g_n)\|_{B_1}\leq C(\alpha,\varepsilon)^{1/2}\|f_n-g_n\|_{A_1}\to0,
\]
so it follows that $T_1f$ is uniquely determined. So $T_1$ is well defined. It is trivial to check that $T_1$ is linear. Furthermore, $T_1$ is bounded: if $f\in A_1$ is smooth, it follows immediately from Proposition \ref{intl2} that $\|T_1f\|_{B_1}\leq C(\alpha,\varepsilon)^{1/2}\|f\|_{A_1}$, and if $f\in A_1$ is not smooth, then for any sequence of smooth functions $\{f_n\}$ in $A_1$ converging to $f$ in $A_1$,
\[
\|T_1f\|_{B_1}=\lim_{n\to\infty}\|T_1f_n\|_{B_1}\leq \lim_{n\to\infty} C(\alpha,\varepsilon)^{1/2}\|f_n\|_{A_1} =C(\alpha,\varepsilon)^{1/2}\|f\|_{A_1}.\]
So we have constructed the desired bounded linear operator $T_1:A_1\to B_1$ and moreover $\|T_1\|\leq C(\alpha,\varepsilon)^{1/2}$, where $C(\alpha,\varepsilon)$ is the constant from Proposition \ref{intl2}.

\noindent \textbf{Step 5c:} We show that $T_0=T_1$ on $A_0\cap A_1$. To see this, we note that $T_0\vert_{A_0\cap A_1}:A_0\cap A_1 \to B_0+B_1$ is continuous. If $f\in A_0\cap A_1$, by continuity of $T_0:A_0\to B_0$ and by definition of $\|\cdot\|_{B_0+B_1}$ and $\|\cdot\|_{A_0\cap A_1}$ (see Appendix \ref{Interpolation Theory}), we have
\[\|T_0f\|_{B_0+B_1}\leq \|T_0f\|_{B_0}\leq C\|f\|_{A_0}\leq C\|f\|_{A_0\cap A_1},\]
so $T_0\vert_{A_0\cap A_1}:A_0\cap A_1 \to B_0+B_1$ is indeed continuous.
In exactly the same way,  $T_1\vert_{A_0\cap A_1}:A_0\cap A_1 \to B_0+B_1$ is continuous.

Now, the space of smooth and compactly supported functions $C^\infty_c$ satisfies $C^\infty_c\subset A_1$ and $(C_c^\infty,\|\cdot\|_{A_1})$ is dense in $(A_1,\|\cdot\|_{A_1})$. Also, because $\alpha>1$, given $f\in A_1$ by H\"{o}lder's inequality
\begin{eqnarray}\label{equivnorm}
\|f\|_{A_0} =\int_1^\infty |f|(v)\,dv\leq \left(\int_1^\infty v^\alpha |f|^2(v)\,dv\right)^{1/2}\left(\int_1^\infty v^{-\alpha}\,dv\right)^{1/2}=C\|f\|_{A_1}<\infty,
\end{eqnarray}
and hence $f\in A_0$. Thus $ A_1\subset A_0$. In particular, $A_0\cap A_1=A_1$. Moreover, (as $\|\cdot\|_{A_1}\leq \|\cdot\|_{A_0\cap A_1}$ and by \eqref{equivnorm}) $\|\cdot\|_{A_0\cap A_1}=\max(\|\cdot\|_{A_0},\|\cdot\|_{A_1})$ is equivalent to $\|\cdot\|_{A_1}$. It follows that $(C^\infty_c,\|\cdot\|_{A_0\cap A_1})$ is dense in $(A_0\cap A_1,\|\cdot\|_{A_0\cap A_1})$. Furthermore, it follows immediately from the definitions of $T_0$ and $T_1$ that $T_0=T_1$ on $C_c^\infty$ (as their actions on smooth functions are defined in the same way).

Now recall the classical result that if $X$, $Y$ are metric spaces, $S,T:X\to Y$ continuous maps and $A\subset X$ a dense subset of $X$ such that $S=T$ on $A$, then $S=T$ (on the entire space $X$). Applying this result with $X=A_0\cap A_1$, $Y=B_0+B_1$, $A=C^\infty_c$, $S=T_0\vert_{A_0\cap A_1}$ and $T=T_1\vert_{A_0\cap A_1}$, we deduce that $T_0=T_1$ on $A_0\cap A_1$.

\noindent \textbf{Step 5d:} We define the linear operator $T:A_0+A_1\to B_0+B_1$ as follows. For  $f\in A_0+A_1$ with $f=f_0+f_1$, $f_i\in A_i$, we set $Tf=T_0f_0+T_1f_1$. We need to check that $T$ is well-defined and linear.
\begin{itemize}
\item
To see that $T$ is well-defined, suppose $f\in A_0+A_1$ with $f=f_0+f_1=g_0+g_1$, where $f_0,g_0\in A_0$ and $f_1,g_1\in A_1$. Then $f_0-g_0=g_1-f_1\in A_0\cap A_1$ and hence (by Step 5c)
\begin{eqnarray*}
T_0(f_0-g_0)=T_0(g_1-f_1)=T_1(g_1-f_1)\hspace{3mm}&\Longrightarrow& \hspace{3mm} T_0f_0-T_0g_0=T_1g_1-T_1f_1\\
\hspace{3mm}&\Longrightarrow&\hspace{3mm}T_0f_0+T_1f_1=T_0g_0+T_1g_1,
\end{eqnarray*}
so $Tf$ is independent of the representation of $f$ and so $T$ is well-defined.

\item To see that $T$ is linear, suppose $\alpha,\beta\in\mathbb{R}$ and $f=f_0+f_1,g=g_0+g_1\in A_0+A_1$ (where $f_0,g_0\in A_0$, $f_1,g_1\in A_1$). Then by linearity of $T_0:A_0\to B_0$ and $T_1:A_1\to B_1$
\begin{eqnarray*}
T(\alpha f+ \beta g)=T((\alpha f_0+\beta g_0)+(\alpha f_1+\beta g_1)) &=& T_0(\alpha f_0+\beta g_0)+T_1(\alpha f_1+\beta g_1)\\
&=&\alpha (T_0f_0+T_1f_1)+\beta(T_0g_0+T_1g_1)\\
&=&\alpha Tf+\beta Tg, 
\end{eqnarray*}
so $T$ is indeed linear.
\end{itemize}

\noindent \textbf{Step 6:} Notice it follows immediately from the definition of the linear operator $T:A_0+A_1\to B_0+B_1$ that $T\vert_{A_0}=T_0:A_0\to B_0$ and $T\vert_{A_1}=T_1:A_1\to B_1$ and both of these maps are bounded. In other words, using the notation of Appendix \ref{Interpolation Theory}, $T:\overline{A}\to \overline{B}$. Hence, by Step 2, 
\[T:(A_0,A_1)_{\theta,p}\to(B_0,B_1)_{\theta,p}\hspace{3mm} \mbox{ with norm } \hspace{3mm}\|T\|_{\theta,p}\leq \|T_0\|^{1-\theta}\|T_1\|^\theta.\]
Recalling that we fixed $p_0=1$, $p_1=1$
$w_0(v)=1$, $w_1(v)=v^\alpha$, $\omega_0(v)=1$ and $\omega_1(v)=v^{\alpha-\varepsilon}$, where $\alpha>1$ is an integer and $0<\varepsilon<1$, Steps 3 and 4 allow us to compute $(A_0,A_1)_{\theta,p}$ and $(B_0,B_1)_{\theta,p}$. By Step 3
\[(A_0,A_1)_{\theta,p}=A_p(w)=L^p_w([1,\infty)),\]
where  
\[w(v)=w_0(v)^{p\frac{1-\theta}{p_0}}w_1(v)^{p\frac{\theta}{p_1}}=1^{p(1-\theta)/1}\cdot (v^\alpha)^{p\theta/2}=v^\frac{\alpha p\theta}{2},\]
so 
\[(A_0,A_1)_{\theta,p}=L^p_{v^\frac{\alpha p\theta}{2}}([1,\infty)).\]
Similarly, by Step 4
\[(B_0,B_1)_{\theta,p}=B_p(\omega)=L^p_\omega([1,\infty)),\]
where  
\[\omega(v)=\omega_0(v)^{p\frac{1-\theta}{p_0}}\omega_1(v)^{p\frac{\theta}{p_1}}=1^{p(1-\theta)/1}\cdot (v^{\alpha-\varepsilon})^{p\theta/2}=v^{(\alpha-\varepsilon)\frac{ p\theta}{2}},\]
so 
\[(B_0,B_1)_{\theta,p}=L^p_{v^{(\alpha -\varepsilon)\frac{p\theta}{2}}}([1,\infty)).\]
Thus, we conclude
\begin{equation}\label{range}
T:L^p_{v^\frac{\alpha p\theta}{2}}([1,\infty))\to L^p_{v^{(\alpha -\varepsilon)\frac{p\theta}{2}}}([1,\infty))
\end{equation}
is a bounded linear operator and by Steps 5a and 5b its norm $\|T\|_{\theta,p}$  is bounded by
\begin{equation}\label{bop}
\|T\|_{\theta,p}\leq \|T_0\|^{1-\theta}\|T_1\|^\theta \leq C^{1-\theta}\cdot C'(\alpha,\varepsilon)^{\theta/2}=\tilde{C}(\alpha,\varepsilon,p),
\end{equation}
where $C$ and $C'$ are the constants from Propositions \ref{intl1} and \ref{intl2} respectively.

\noindent \textbf{Step 7:} We complete the proof. Recall that given a solution $\phi$ of the wave equation with smooth, spherically symmetric data on $S$, we need to show \eqref{phi1}, where we denoted by $\phi_1$ a smooth, spherically symmetric solution of the wave equation such that 
\[\partial_u\phi_1\vert_{\underline{C}_1^{int}\cap\{u\leq-1\}}=0, \hspace{5mm} \partial_v\phi_1\vert_{C_{-1}^{int}\cap\{v\geq1\}}=\partial_v\phi\vert_{C_{-1}^{int}\cap\{v\geq1\}}.\]

Note that if 
\[\int_1^\infty v^\frac{\alpha p\theta}{2}|\partial_v\phi|^p(-1,v)\,dv=+\infty,\]
then \eqref{phi1} trivially holds. So we assume 
\[\int_1^\infty v^\frac{\alpha p\theta}{2}|\partial_v\phi|^p(-1,v)\,dv<+\infty,\]
namely we assume $f:=\partial_v\phi\vert_{C_{-1}^{int}\cap\{v\geq1\}}\in L^p_{v^\frac{\alpha p\theta}{2}}([1,\infty))=(A_0,A_1)_{\theta,p}$. But this implies $Tf\in (B_0,B_1)_{\theta,p}=L^p_{v^{(\alpha -\varepsilon)\frac{p\theta}{2}}}([1,\infty))$ (using \eqref{range}), and moreover by \eqref{bop}
\[\int_1^\infty v^{(\alpha-\varepsilon)\frac{p\theta}{2}}|Tf|^p(v)\,dv \leq \tilde{C}(\alpha,\varepsilon,p)^p\int_1^\infty v^\frac{\alpha p\theta}{2}|f|^p(v)\,dv=\tilde{C}(\alpha,\varepsilon,p)^p\int_1^\infty v^\frac{\alpha p\theta}{2}|\partial_v\phi|^p(-1,v)\,dv.\]
Comparing the last inequality to \eqref{phi1}, we see that in order to prove \eqref{phi1} it remains only to prove $Tf=\partial_v\phi_1\vert_{\mathcal{H}^+\cap\{v\geq1\}}$.

To see this, note that $f\in (A_0,A_1)_{\theta,p}\subseteq A_0+A_1\subseteq A_0+A_0$ since $A_1\subseteq A_0$ (as $\alpha>1$). But $A_0+A_0=A_0$ (with equal norms), so in particular $f\in A_0$. Moreover, as $\phi$ is a smooth solution of the wave equation, $f=\partial_v\phi\vert_{C_{-1}^{int}\cap\{v\geq1\}}$ is smooth. Hence, by construction of $T$ (Step 5d), $Tf=T_0f$ and by definition of $T_0$ (Step 5a), $T_0f=\partial_v\phi_1\vert_{\mathcal{H}^+\cap\{v\geq1\}}$. So $Tf=\partial_v\phi_1\vert_{\mathcal{H}^+\cap\{v\geq1\}}$ and the proof is complete.
\end{proof}

\subsection{Proof of Theorem \ref{condinstab2}}\label{condrel2}
Now that we have seen the proof of Theorem \ref{interior reduction}, and in particular the proof of the $L^1$ estimates, we give the proof of the conditional instability result, Theorem \ref{condinstab2}.

\begin{proof}[Proof of Theorem \ref{condinstab2}.]
For the parameter range $\frac{2\sqrt e}{e+1}<Q<1$, the result trivially follows from the instability result Theorem \ref{mtv2}. So we need only consider the parameter range $0<Q\leq \frac{2\sqrt e}{e+1}$.

Now, as noted in Section \ref{condrel}, Theorem \ref{mtv2}  holds for the  range $0<Q\leq \frac{2\sqrt e}{e+1}$ if Theorem \ref{mainthm} is valid for that range. Moreover, Theorem \ref{mainthm}  holds for $0<Q\leq \frac{2\sqrt e}{e+1}$ if Theorem \ref{interior reduction} holds in this range for all smooth, spherically symmetric solutions $\phi$ of \eqref{waveequation} satisfying the initial conditions \eqref{ic1} and \eqref{ic2} and the assumption \eqref{limit}. In particular, instability holds in $0<Q\leq \frac{2\sqrt e}{e+1}$ if Theorem \ref{interior reduction} holds in this range for all smooth, spherically symmetric solutions $\phi$ satisfying
\eqref{ic2}.
 However, arguing as in the proof of Theorem \ref{interior reduction} we see that the estimate \eqref{phi2} holds regardless of the parameter range, and therefore Theorem \ref{interior reduction} holds if we can establish the estimate \eqref{phi1} for the range $0<Q\leq \frac{2\sqrt e}{e+1}$. 
 
 By the interpolation argument in the proof of Theorem \ref{interior reduction}, \eqref{phi1} holds for $0<Q\leq \frac{2\sqrt e}{e+1}$ so long as the $L^1$ and $L^2$ estimates of Propositions \ref{intl1} and \ref{intl2} hold in this range for smooth, spherically symmetric solutions $\phi$ satisfying $\partial_u\phi=0$ on $\underline{C}^{int}_1\cap\{u\leq-1\}$. This is trivially true in the $L^2$ case by Proposition \ref{intl2}. Thus, the instability result Theorem \ref{mtv2} holds for $0<Q\leq \frac{2\sqrt e}{e+1}$, provided there exists some $C=C>0$ such that
\begin{equation*}
\int_{\mathcal{H}^+\cap\{v\geq1\}}|\partial_v\phi|\leq C\int_{C_{-1}^{int}\cap\{v\geq1\}}   |\partial_v\phi|
\end{equation*}
for all smooth, spherically symmetric solutions $\phi$ satisfying $\partial_u\phi=0$ on $\underline{C}_1\cap\{u\leq-1\}$.

It follows that the instability result Theorem \ref{mtv2} holds unless there is a solution $\phi$ as above such that 
\[\int_{\mathcal{H}^+\cap\{v\geq1\}}|\partial_v\phi|\] cannot be bounded by a constant multiple of \[\int_{C_{-1}^{int}\cap\{v\geq1\}}   |\partial_v\phi|,\]
or in other words unless there is a sequence $\{\phi_n\}$ of smooth, spherically symmetric solutions to \eqref{waveequation}  such that in the black hole interior
\begin{enumerate}
\item
$\partial_u\phi_n=0$ on $\underline{C}_1\cap\{u\leq-1\}$, 
\item
$\int_{\mathcal{H}^+\cap\{v\geq1\}}|\partial_v\phi_n|=1$, and
\item
$\int_{C_{-1}^{int}\cap\{v\geq1\}}   |\partial_v\phi_n|\searrow 0$, 
\end{enumerate}
as desired.
\end{proof}

\newpage
\section{Estimates in the Black Hole Exterior}\label{Exterior}

In this section, we prove the results relating to the black hole exterior needed for  the proof of Theorem \ref{reduction}, namely Theorem \ref{lp exterior estimate}, Proposition \ref{alt2} and Propositon \ref{alt1}. The proofs of the two propositions are straightforward and we defer them to the end of the section and focus first on the much more involved proof of Theorem \ref{lp exterior estimate}. Note that for the rest of this section we work in the coordinate system for the exterior region introduced in Section \ref{exterior coordinates}.

\subsection{Proof of Theorem \ref{lp exterior estimate}}

Recall from Section \ref{Outline of proof} that Theorem \ref{lp exterior estimate} asserts that if $\alpha-(1+2\varepsilon)>0$, then for a solution $\phi$ of the wave equation \eqref{waveequation} with smooth, spherically symmetric data on $S$ and  $R>r_+$ sufficiently large, there exists $C=C(R,\alpha,\varepsilon,p)>0$ such that
\[\int_{\gamma_R\cap\{v\geq1\}} v^{(\alpha-(1+2\varepsilon))\frac{p\theta}{2}}|\partial_v\phi|^p \leq C \left(\int_{\mathcal{H}^+\cap\{v\geq1\}} v^{\frac{\alpha p\theta}{2}}|\partial_v\phi|^p+\int_{\underline{C}_1^{ext}\cap\{u\geq u_R(1)\}} \Omega^{-p\theta}|\partial_u\phi|^p\right) \]
and
\[\int_{\gamma_R\cap\{v\geq1\}} v^{(\alpha-(1+2\varepsilon))\frac{p\theta}{2}}|\partial_u\phi|^p \leq C \left(\int_{\mathcal{H}^+\cap\{v\geq1\}} v^{\frac{\alpha p\theta}{2}}|\partial_v\phi|^p+\int_{\underline{C}_1^{ext}\cap\{u\geq u_R(1)\}} \Omega^{-p\theta}|\partial_u\phi|^p\right), \]
where
\[\frac{1}{p}=\frac{1-\theta}{1}+\frac{\theta}{2}, \hspace{10mm}0<\theta<1.\]
As in the case of the black hole interior, we prove these weighted $L^p$-type estimates by proving the analogous estimates for $p=1$ and $p=2$ and interpolating to deduce the estimates for $1<p<2$. We begin with the $L^1$-type estimates.

\subsubsection{$L^1$ Estimates}

In this section, we aim to show that we can control $\int_{\gamma_R\cap\{v\geq1\}}|\partial_v\phi|$ and $\int_{\gamma_R\cap\{v\geq1\}}|\partial_u\phi|$ by terms of the form $\int_{\mathcal{H}^+\cap\{v\geq1\}}|\partial_v\phi|$ and $\int_{\underline{C}_1^{ext}\cap\{u\geq u_R(1)\}}|\partial_u\phi|$ for any $R>r_+$. However, we cannot prove this directly. Rather, we must partition the region $\{r_+\leq r\leq R\}$ into subregions where $\sup_u\int|\lambda|\,dv$ and $\sup_v\int|\nu|\,du$ are small, obtain an estimate for each region separately and then patch these estimates together.

The next proposition establishes the desired estimates for a constant $r$-hypersurfaces sufficiently close to the event horizon. The proof is similar to the proof of Proposition \ref{intl1} (so again uses ideas from Proposition 13.1 of \cite{D}).

\begin{proposition}\label{1vprop}
Suppose $\phi$ is a solution of the wave equation \eqref{waveequation} with smooth, spherically symmetric data on $S$. Let $r_1>r_+$ be such that $r_1-r_+<\frac{r_+}{2}$. Then there exists $C>0$ (independent of $r_1$) such that
\begin{equation}\label{1v}
\int_{\gamma_{r_1}\cap\{v\geq 1\}}|\partial_v\phi|\,dv\leq C\left(\int_{\mathcal{H}^+\cap\{v\geq1\}}|\partial_v\phi|\,dv+\int_{u_{r_1}(1)}^\infty|\partial_u\phi|(u,1)\,du\right),
\end{equation}
and
\begin{equation}\label{1u}
\int_{\gamma_{r_1}\cap\{u\geq u_{r_1}(1)\}}|\partial_u\phi|\,du\leq C\left(\int_{\mathcal{H}^+\cap\{v\geq1\}}|\partial_v\phi|\,dv+\int_{u_{r_1}(1)}^\infty|\partial_u\phi|(u,1)\,du\right).
\end{equation}
\end{proposition}

\begin{proof}
We work in the shaded region shown in the figure below, namely $\{r_+\leq r\leq r_1\}\cap\{v\geq1\}$.
Recall from \eqref{ln1} that $\lambda=\partial_vr=\Omega^2>0$ and $\nu=\partial_ur=-\Omega^2<0$ in this region.
\begin{figure}[H]
\centering
\includegraphics[scale=.9]{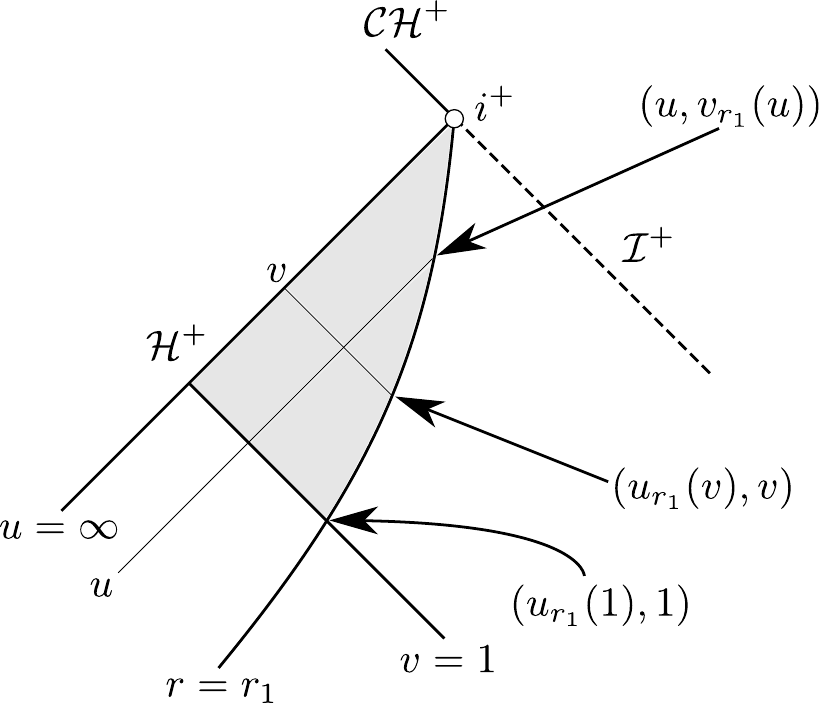}
\end{figure}

\noindent The wave equation \eqref{extwe} may be written as 
\[\partial_u(r\partial_v\phi)=-\lambda\partial_u\phi.\] 
For $v\geq 1$ and $u\geq u_{r_1}(1)$, integrating over $u'\in[u,\infty]$ yields
\[|r\partial_v\phi|(u,v)|\leq|r\partial_v\phi|(\infty,v)|+\int_u^\infty |\lambda\partial_u\phi|\,du',\]
and hence (as $\frac{r(\infty,v)}{r(u,v)}=\frac{r_+}{r(u,v)}\leq 1$ and $\frac{1}{r(u,v)}\leq \frac{1}{r_+}$),
\begin{equation}\label{weint'}
|\partial_v\phi|(u,v)|\leq|\partial_v\phi|(\infty,v)|+\frac{1}{r_+}\int_u^\infty |\lambda\partial_u\phi|\,du'.
\end{equation}
Thus, setting
\[\theta(u_1,v):=\sup_{u_1\leq u\leq \infty}|\partial_v\phi|(u,v),\]
and taking the supremum of \eqref{weint'} over $u\in[u_{r_1}(v),\infty]$ and then integrating over $v\in [1,\infty)$, we get
\begin{equation}\label{esst8}
\int_1^\infty\theta(u_{r_1}(v),v)\,dv\leq\int_1^\infty|\partial_v\phi|(\infty,v)\,dv+\frac{1}{r_+}\int_1^\infty\int_{u_{r_1}(v)}^\infty|\lambda\partial_u\phi|\,du\,dv.
\end{equation}
Now,  define 
\[Z(u,v):=\sup_{1\leq \overline{v}\leq v}|\partial_u\phi|(u,\overline{v}).\]
Then, since $|\lambda|=\partial_vr$,
\begin{eqnarray*}
\int_1^\infty\int_{u_{r_1}(v)}^\infty|\lambda\partial_u\phi|\,du\,dv&=& \int_{u_{r_1}(1)}^\infty\int_1^{v_{r_1}(u)}|\lambda\partial_u\phi|\,dv\,du\\
&\leq&\sup_{u\in[u_{r_1}(1),\infty]}\int_1^{v_{r_1}(u)}|\lambda|\,dv\int_{u_{r_1}(1)}^\infty\sup_{v\in[1,v_{r_1}(u)]}|\partial_u\phi|\,du\\
&=&\sup_{u\in[u_{r_1}(1),\infty]}(r_1-r(u,1))\int_{u_{r_1}(1)}^\infty Z(u,v_{r_1}(u))\,du\\
&=& (r_1-r_+)\int_{u_{r_1}(1)}^\infty Z(u,v_{r_1}(u))\,du.
\end{eqnarray*}
So \eqref{esst8} becomes 
\begin{equation}\label{esst99}
\int_1^\infty\theta(u_{r_1}(v),v)\,dv\leq\int_1^\infty|\partial_v\phi|(\infty,v)\,dv+\frac{r_1-r_+}{r_+}\int_{u_{r_1}(1)}^\infty Z(u,v_{r_1}(u))\,du.
\end{equation}

To estimate the rightmost term, we rewrite the wave equation \eqref{extwe} as
\[\partial_v(r\partial_u\phi)=-\nu\partial_v\phi\] 
and integrate  over $v\in[1,\overline{v}]$ to get
\[|r\partial_u\phi|(u,\overline{v})\leq |r\partial_u\phi|(u,1)+\int_1^{\overline{v}}|\nu\partial_v\phi|\,dv,\] and hence
\[|\partial_u\phi|(u,\overline{v})\leq |\partial_u\phi|(u,1)+\frac{1}{r_+}\int_1^{\overline{v}}|\nu\partial_v\phi|\,dv\]
because $\frac{r(u,1)}{r(u,\overline{v})}\leq 1$ since $r$ is increasing in $v$.
Taking the supremum over $\overline{v}\in [1,v_{r_1}(u)]$ and integrating over $u\in[u_{r_1}(1),\infty]$ yields
\[\int_{u_{r_1}(1)}^\infty Z(u,v_{r_1}(u))\,du\leq \int_{u_{r_1}(1)}^\infty |\partial_u\phi|(u,1)\,du+\frac{1}{r_+}\int_{u_{r_1}(1)}^\infty\int_1^{v_{r_1}(u)}|\nu\partial_v\phi|\,dv\,du.\]
But $|\nu|=-\partial_ur$, so
\begin{eqnarray*}
\int_{u_{r_1}(1)}^\infty\int_1^{v_{r_1}(u)}|\nu\partial_v\phi|\,dv\,du&=&\int_1^\infty\int_{u_{r_1}(v)}^\infty|\nu\partial_v\phi|\,du\,dv\\
&\leq&\sup_{v\in[1,\infty)}\int_{u_{r_1}(v)}^\infty|\nu|\,du\int_1^\infty\sup_{u\in[u_{r_1}(v),\infty]}|\partial_v\phi|\,dv\\
&=&(r_1-r_+)\int_1^\infty \theta(u_{r_1}(v),v)\,dv,
\end{eqnarray*}
and hence 
\[ \int_{u_{r_1}(1)}^\infty Z(u,v_{r_1}(u))\,du\leq \int_{u_{r_1}(1)}^\infty |\partial_u\phi|(u,1)\,du+\frac{r_1-r_+}{r_+}\int_1^\infty \theta(u_{r_1}(v),v)\,dv.\]

Substituting the last equation into \eqref{esst99} gives
\begin{eqnarray*}
\int_1^\infty\theta(u_{r_1}(v),v)\,dv&\leq&\int_1^\infty|\partial_v\phi|(\infty,v)\,dv+\frac{r_1-r_+}{r_+}\int_{u_{r_1}(1)}^\infty |\partial_u\phi|(u,1)\,du\\
& &+\left(\frac{r_1-r_+}{r_+}\right)^2\int_1^\infty \theta(u_{r_1}(v),v)\,dv,
\end{eqnarray*}
and since $\displaystyle\left(\frac{r_1-r_+}{r_+}\right)<1/2$, we deduce
\[\int_1^\infty\theta(u_{r_1}(v),v)\,dv\leq C\left(\int_1^\infty|\partial_v\phi|(\infty,v)\,dv+\int_{u_{r_1}(1)}^\infty |\partial_u\phi|(u,1)\,du\right).\]
Thus,
\begin{eqnarray*}
\int_1^\infty|\partial_v\phi|(u_{r_1}(v),v)\,dv &\leq&\int_1^\infty\theta(u_{r_1}(v),v)\,dv\\
&\leq&C\left(\int_1^\infty|\partial_v\phi|(\infty,v)\,dv+\int_{u_{r_1}(1)}^\infty |\partial_u\phi|(u,1)\,du\right),
\end{eqnarray*}
which proves \eqref{1v}. The proof of the second statement \eqref{1u} is similar.
\end{proof}

Next, we show that if two constant $r$-hypersurfaces, $\gamma_{r_j}$ and $\gamma_{r_{j+1}}$  with $r_j<r_{j+1}$ say, are sufficiently close, we can control
\[\int_{\gamma_{r_{j+1}}\cap\{v\geq 1\}}|\partial_v\phi|\,dv \hspace{5mm}\mbox{ and } \hspace{5mm} \int_{\gamma_{r_{j+1}}\cap\{u\geq u_{r_{j+1}}(1)\}}|\partial_u\phi|\,du \]
in terms of the corresponding integrals on $\gamma_{r_j}$ and $\int_{u_{r_{j+1}}(1)} ^{u_{r_j}(1)}|\partial_u\phi|(u,1)\,du$. The proof is very similar to the proof of the previous proposition.

\begin{proposition}\label{jvprop}
Suppose $\phi$ is a solution of the wave equation \eqref{waveequation} with smooth, spherically symmetric data on $S$. Let $r_{j+1}>r_j>r_+$ be such that $r_{j+1}-r_j <\frac{r_+}{2}$. Then
\begin{eqnarray}
\int_{\gamma_{r_{j+1}}\cap\{v\geq 1\}}|\partial_v\phi|\,dv &\leq& C\left(\int_{\gamma_{r_{j}}\cap\{v\geq 1\}}|\partial_v\phi|\,dv+\int_{\gamma_{r_{j}}\cap\{u\geq u_{r_j}(1)\}}|\partial_u\phi|\,du\right.\notag\\
& & \left.+\int_{u_{r_{j+1}}(1)} ^{u_{r_j}(1)}|\partial_u\phi|(u,1)\,du\right) \label{jv}
\end{eqnarray}
and
\begin{eqnarray}
\int_{\gamma_{r_{j+1}}\cap\{u\geq u_{r_{j+1}}(1)\}}|\partial_u\phi|\,du &\leq& C\left(\int_{\gamma_{r_{j}}\cap\{v\geq 1\}}|\partial_v\phi|\,dv+\int_{\gamma_{r_{j}}\cap\{u\geq u_{r_j}(1)\}}|\partial_u\phi|\,du\right.\notag\\
& & \left.+\int_{u_{r_{j+1}}(1)} ^{u_{r_j}(1)}|\partial_u\phi|(u,1)\,du\right), \label{ju}
\end{eqnarray}
where $C>0$ is independent of $r_j,r_{j+1}$. 
\end{proposition}

\begin{proof}
The proof of this proposition runs along the same lines as the previous proposition but is a little more involved. This time the analysis takes place in the shaded region shown below.
\begin{figure}[h]
\centering 
\includegraphics[width=0.5\textwidth]{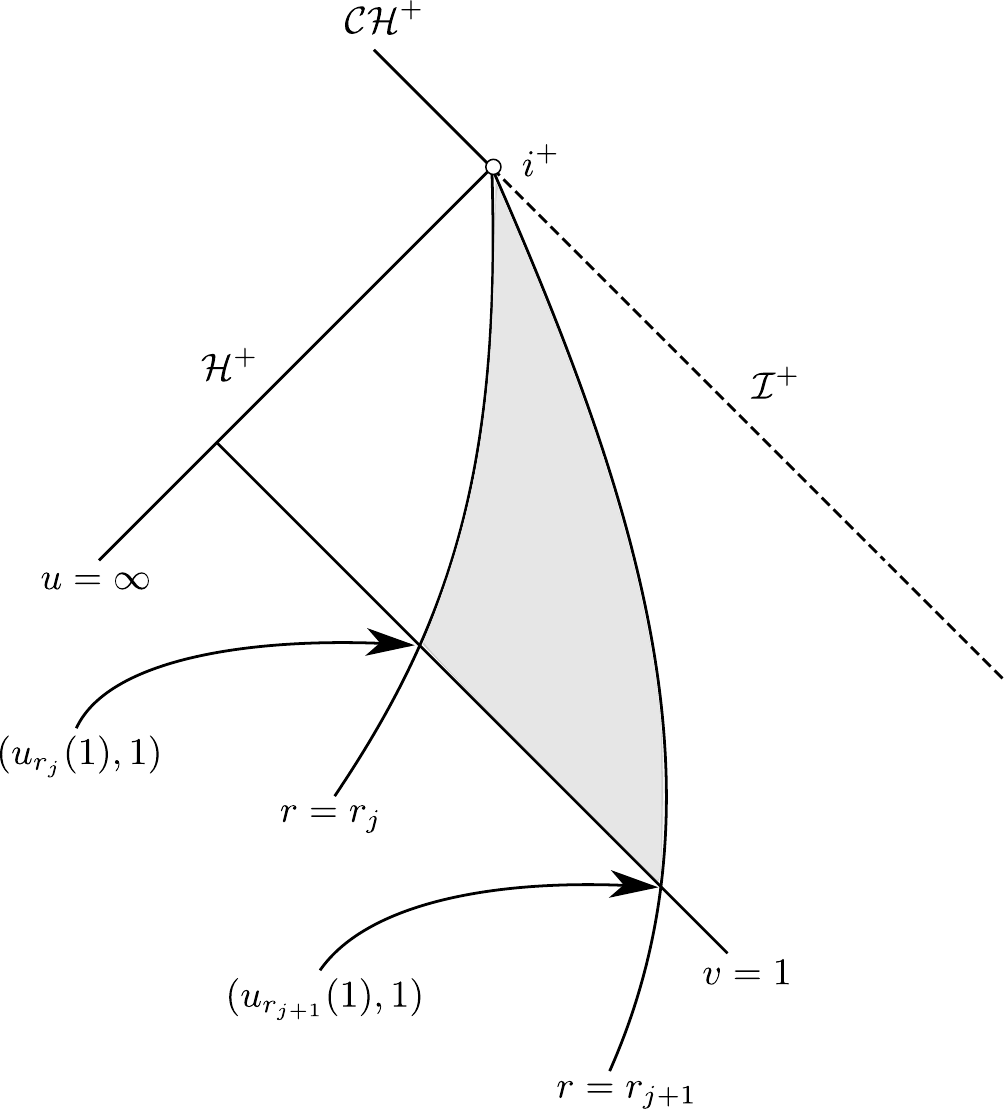}
\end{figure}
As noted previously, the wave equation \eqref{extwe} may be written as 
\[\partial_u(r\partial_v\phi)=-\lambda\partial_u\phi.\]
Thus for $v\geq 1$ and $u\in[u_{r_{j+1}}(v),u_{r_j}(v)]$, we have
\begin{eqnarray}
r|\partial_v\phi|(u,v)\leq r|\partial_v\phi|(u_{r_j}(v),v) +\int_{u}^{u_{r_j}(v)} |\lambda\partial_u\phi|(u',v)\,du'\notag\\
\Longrightarrow |\partial_v\phi|(u,v)\leq |\partial_v\phi|(u_{r_j}(v),v) +\frac{1}{r_+}\int_{u}^{u_{r_j}(v)} |\lambda\partial_u\phi|(u',v)\,du'\label{esst4}
\end{eqnarray}
since $r(u_{r_j}(v),v)\leq r(u,v)$ because $r$ is decreasing in $u$.
Define
\begin{equation*}
\theta(u,v) :=\sup_{u\leq u'\leq u_{r_j}(v)}|\partial_v\phi|(u',v) \hspace{3mm}\mbox{ for }u\leq u_{r_j}(v).
\end{equation*}
Take the supremum of \eqref{esst4} over $u\in [u_{r_{j+1}}(v),u_{r_j}(v)]$ and integrate over $v\in[1,\infty)$ to get
\begin{equation}\label{esst5}
\int_1^\infty\theta(u_{r_{j+1}}(v),v)\,dv\leq \int_1^\infty|\partial_v\phi|(u_{r_j}(v),v)\,dv+\frac{1}{r_+}\int_1^\infty \int_{u_{r_{j+1}}(v)}^{u_{r_j}(v)} |\lambda\partial_u\phi|(u,v)\,du\,dv.
\end{equation}

We shall estimate the rightmost term. Indeed, setting $M(u)=\max(1,v_{r_j}(u))$ for $u\in [u_{r_{j+1}}(1),\infty]$, we have
\begin{eqnarray}
\int_1^\infty\int_{u_{r_{j+1}}(v)}^{u_{r_j}(v)} |\lambda\partial_u\phi|(u,v)\,du\,dv &=&\int_{u_{r_{j+1}}(1)}^\infty\int_{M(u)}^{v_{r_{j+1}}(u)}|\lambda\partial_u\phi|\,dv\,du\notag\\
&\leq&\sup_{u\in[u_{r_{j+1}}(1),\infty]}\int_{M(u)}^{v_{r_{j+1}}(u)}|\lambda|\,dv \int_{u_{r_{j+1}}(1)}^\infty\sup_{v\in[M(u),v_{r_{j+1}}(u)]}|\partial_u\phi|\,du\notag\\
&\leq&(r_{j+1}-r_j)\int_{u_{r_{j+1}}(1)}^\infty\sup_{v\in[M(u),v_{r_{j+1}}(u)]}|\partial_u\phi|\,du\label{dagga}
\end{eqnarray}
because $|\lambda|=\partial_vr$, and so
\begin{eqnarray*}
\sup_{u\in[u_{r_{j+1}}(1),\infty]}\int_{M(u)}^{v_{r_{j+1}}(u)}|\lambda|\,dv&=&\sup_{u\in[u_{r_{j+1}}(1),\infty]} [r(u,v_{r_{j+1}}(u))-r(u,M(u))]\\
&=&\sup_{u\in[u_{r_{j+1}}(1),\infty]}\left\{\begin{array}{lc} r_{j+1}-r_j& \mbox{ if } v_{r_j}(u)\geq1 \Leftrightarrow u\geq u_{r_j}(1)\\ r_{j+1}-r(u,1)& \mbox{ if }v_{r_j}(u)<1 \Leftrightarrow u<u_{r_j}(1)\end{array}\right.\\
&=&\max\left(r_{j+1}-r_j,\sup_{u\in [u_{r_{j+1}}(1),u_{r_j}(1)]} (r_{j+1}-r(u,1))\right)\\
&=& r_{j+1}-r_{j} \hspace{3mm}\mbox{ since } r \mbox{ is decreasing in } u.
\end{eqnarray*}

Now define 
\begin{equation*}
Z(u,v):=\sup_{M(u)\leq \overline{v}\leq v}|\partial_u\phi|(u,\overline{v}) \hspace{4mm}\mbox{ for } v\geq M(u). 
\end{equation*}
Then by \eqref{esst5} and \eqref{dagga} we have
\begin{equation}\label{esst7}
\int_1^\infty\theta(u_{r_{j+1}}(v),v)\,dv\leq \int_1^\infty|\partial_v\phi|(u_{r_j}(v),v)\,dv+\left(\frac{r_{j+1}-r_j} {r_+}\right)\int_{u_{r_{j+1}}(1)}^\infty Z(u,v_{r_{j+1}}(u))\,du.
\end{equation}
Now rewriting the wave equation \eqref{extwe} as 
\[\partial_v(r\partial_u\phi)=-\nu\partial_v\phi\]
and integrating over $v\in[M(u),\overline{v}]$, we deduce
\begin{equation}\label{esst6}
|\partial_u\phi|(u,\overline{v})\leq  |\partial_u\phi|(u,M(u))+\frac{1}{r_+}\int_{M(u)}^{\overline{v}}|\nu\partial_v\phi|\,dv
\end{equation}
for $\overline{v}\geq M(u)=\max(1,v_{r_J}(u))$. Now, $r_{j+1}>r_j\Longrightarrow  v_{r_{j+1}}(u)\geq v_{r_j}(u)\, \forall u$. Moreover, for $u\geq u_{r_{j+1}}(1)$, $v_{r_{j+1}}(u)\geq 1$. Thus $u\in[u_{r_{j+1}}(1),\infty] \Longrightarrow v_{r_{j+1}}(u)\geq M(u)$.  Thus $[M(u),v_{r_{j+1}}(u)]$ is nonempty $\forall u\in[u_{r_{j+1}}(1),\infty]$, so for $u\geq u_{r_{j+1}}(1)$ we can thus take the supremum of \eqref{esst6} over $\overline{v}\in[M(u),v_{r_{j+1}}(u)]$  to deduce
\[Z(u,v_{r_{j+1}}(u))\leq |\partial_u\phi|(u,M(u))+\frac{1}{r_+}\int_{M(u)}^{v_{r_{j+1}}(u)}|\nu\partial_v\phi|\,dv,\]
and hence
\begin{eqnarray*}
\int_{u_{r_{j+1}}(1)}^\infty Z(u,v_{r_{j+1}}(u))\,du
&\leq&\int_{u_{r_{j+1}}(1)}^\infty |\partial_u\phi|(u,M(u))\,du+\frac{1}{r_+}\int_{u_{r_{j+1}}(1)}^\infty\int_{M(u)}^{v_{r_{j+1}}(u)}|\nu\partial_v\phi|\,dv\,du\\
&=&\int_{u_{r_{j+1}}(1)}^\infty |\partial_u\phi|(u,M(u))\,du+\frac{1}{r_+}\int_1^\infty\int_{u_{r_{j+1}}(v)}^{u_{r_j}(v)} |\nu\partial_v\phi|\,du\,dv.
\end{eqnarray*}
But
\begin{eqnarray*}
\int_1^\infty\int_{u_{r_{j+1}}(v)}^{u_{r_j}(v)} |\nu\partial_v\phi|\,du\,dv&\leq&
\sup_{v\in[1,\infty)} \int_{u_{r_{j+1}}(v)}^{u_{r_j}(v)}|\nu|\,du \int_1^\infty \sup_{u\in[u_{r_{j+1}}(v),u_{r_j}(v)]}|\partial_v\phi|\,dv\\
&=& \sup_{v\in[1,\infty)} \int_{u_{r_{j+1}}(v)}^{u_{r_j}(v)}(-\partial_ur)\,du \int_1^\infty \theta(u_{r_{j+1}}(v),v)\,dv\\
&=&(r_{j+1}-r_j) \int_1^\infty \theta(u_{r_{j+1}}(v),v)\,dv,
\end{eqnarray*}
so
\begin{eqnarray*}
\int_{u_{r_{j+1}}(1)}^\infty Z(u,v_{r_{j+1}}(u))\,du \leq
\int_{u_{r_{j+1}}(1)}^\infty |\partial_u\phi|(u,M(u))\,du+\left(\frac{r_{j+1}-r_j}{r_+}\right) \int_1^\infty \theta(u_{r_{j+1}}(v),v)\,dv.\hspace{12mm}
\end{eqnarray*}
Substituting this equation into \eqref{esst7} gives
\begin{eqnarray*}
\int_1^\infty\theta(u_{r_{j+1}}(v),v)\,dv&\leq& \int_1^\infty|\partial_v\phi|(u_{r_j}(v),v)\,dv+\left(\frac{r_{j+1}-r_j}{r_+}\right) \int_{u_{r_{j+1}}(1)}^\infty |\partial_u\phi|(u,M(u))\,du\\
& &+\left(\frac{r_{j+1}-r_j}{r_+}\right)^2 \int_1^\infty \theta(u_{r_{j+1}}(v),v)\,dv.
\end{eqnarray*}
Since $\displaystyle\left(\frac{r_{j+1}-r_j}{r_+}\right) <\displaystyle\frac{1}{2}$, we have
\[\int_1^\infty\theta(u_{r_{j+1}}(v),v)\,dv\leq C\left( \int_1^\infty|\partial_v\phi|(u_{r_j}(v),v)\,dv+ \int_{u_{r_{j+1}}(1)}^\infty |\partial_u\phi|(u,M(u))\,du\right),\]
or in other words
\begin{eqnarray}
\int_1^\infty\sup_{u_{r_{j+1}}(v)\leq u\leq u_{r_j}(v)}|\partial_v\phi|(u,v)\,dv\hspace{90mm}\notag \\
\leq C\left(\int_1^\infty|\partial_v\phi|(u_{r_j}(v),v)\,dv+\int_{u_{r_{j+1}}(1)}^\infty|\partial_u\phi|(u,M(u))\,du\right)\hspace{40mm}\notag\\
\leq C\left( \int_{\gamma_{r_j}\cap\{v\geq 1\}}|\partial_v\phi|\,dv +\int_{u_{r_{j+1}}(1)}^{u_{r_j}(1)}|\partial_u\phi|(u,1)\,du+\int_{u_{r_j}(1)}^\infty|\partial_u\phi|(u,v_{r_j}(u))\,du\right)\hspace{10mm}\notag\\
=C\left( \int_{\gamma_{r_j}\cap\{v\geq 1\}}|\partial_v\phi|\,dv +\int_{u_{r_{j+1}}(1)}^{u_{r_j}(1)}|\partial_u\phi|(u,1)\,du +\int_{\gamma_{r_j}\cap\{u\geq u_{r_j}(1)\}}|\partial_u\phi|\,du\right).\hspace{11mm}\label{starr}
\end{eqnarray}
In particular, 
\[\int_{\gamma_{r_{j+1}}\cap\{v\geq 1\}}|\partial_v\phi|\,dv=\int_1^\infty|\partial_v\phi|(u_{r_{j+1}}(v),v)\,dv\]
can be controlled by the right hand side of \eqref{starr}, which proves \eqref{jv}. The second statement \eqref{ju} is proved similarly.\end{proof}

\begin{remark}
Note that the same constant $C$ will work for both Propostion \ref{1vprop} and Proposition \ref{jvprop}.
\end{remark}

Combining the previous two propositions allows us to deduce estimates for 
\[\int_{\gamma_{R}\cap\{v\geq 1\}}|\partial_v\phi|\,dv \hspace{5mm}\mbox{ and }\hspace{5mm}\int_{\gamma_{R}\cap\{v\geq 1\}}|\partial_u\phi|\,dv=\int_{\gamma_{R}\cap\{u\geq u_{R}(1)\}}|\partial_u\phi|\,du\]
for any $R>r_+$. If $R$ is sufficiently close to $r_+$, this is merely Proposition \ref{1vprop}. If $R$ is large, however, we partition the region $\{r_+\leq r\leq R\}\cap\{v\geq1\}$ into subregions separated by constant $r$-hypersurfaces such that in each subregion we can apply either Proposition \ref{1vprop} or Proposition \ref{jvprop}. Combining these estimates will yield an estimate for the $\gamma_R$ integral.

\begin{corollary}\label{extl11}
Suppose $\phi$ is a solution of the wave equation \eqref{waveequation} with smooth, spherically symmetric data on $S$. Let $R>r_+$. Then there exists $C=C(R)>0$ such that
\begin{equation}\label{jjv}
\int_{\gamma_{R}\cap\{v\geq 1\}}|\partial_v\phi|\leq C\left(\int_{\mathcal{H}^+\cap\{v\geq1\}}|\partial_v\phi|+\int_{\underline{C}_1^{ext}\cap\{u\geq u_R(1)\}} |\partial_u\phi|\right)
\end{equation}
and
\begin{eqnarray}
\int_{\gamma_{R}\cap\{v\geq 1\}}|\partial_u\phi|&=&\int_{\gamma_{R}\cap\{u\geq u_{R}(1)\}}|\partial_u\phi|\,du\notag\\
&\leq& C\left(\int_{\mathcal{H}^+\cap\{v\geq1\}}|\partial_v\phi|+\int_{\underline{C}_1^{ext}\cap\{u\geq u_R(1)\}}|\partial_u\phi|\right).\label{indarg2}
\end{eqnarray}
\end{corollary}

\begin{proof}
If $R-r_+<\frac{r_+}{2}$, we are done by Proposition \ref{1vprop} with $r_1:=R$. So assume $R-r_+\geq\frac{r_+}{2}$.
We may partition the region $\{r_+\leq r\leq R\}\cap\{v\geq1\}$ into $N=N(R)>1$ regions $\{r_i\leq r \leq r_{i+1}\}\cap\{v\geq1\}, i=0,\ldots,N-1$ such that $r_{i+1}-r_i<\frac{r_+}{2}$ for each $i$. So $r_0=r_+$ and $r_N=R$. This is illustrated in the Penrose diagram below.

\begin{figure}[h]
\centering
\includegraphics[width=0.6\textwidth]{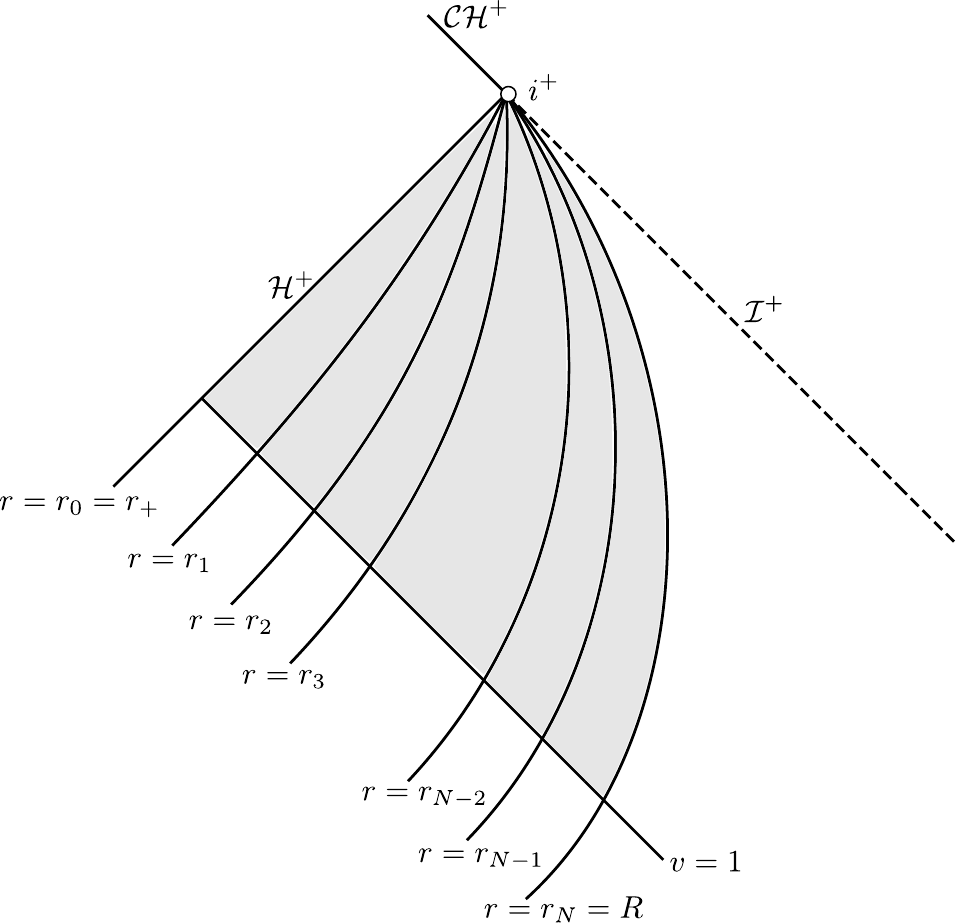}
\end{figure}

By Proposition \ref{jvprop} with $j=N-1$, we have
\begin{eqnarray}
\int_{\gamma_{R}\cap\{v\geq 1\}}|\partial_v\phi|\,dv &=&
\int_{\gamma_{r_{N}}\cap\{v\geq 1\}}|\partial_v\phi|\,dv\notag\\
 &\leq& C\left(\int_{\gamma_{r_{N-1}}\cap\{v\geq 1\}}|\partial_v\phi|\,dv+\int_{\gamma_{r_{N-1}}\cap\{u\geq u_{r_{N-1}}(1)\}}|\partial_u\phi|\,du\right.\notag\\
& & \left.+\int_{u_{r_{N}}(1)} ^{u_{r_{N-1}}(1)}|\partial_u\phi|(u,1)\,du\right).\label{unused}
\end{eqnarray}

If $N-1=1$, we use Proposition \ref{1vprop} to estimate the first two terms on the right hand side:
\begin{eqnarray*}
\int_{\gamma_{R}\cap\{v\geq 1\}}|\partial_v\phi|\,dv  &\leq& C\left(\int_{\gamma_{r_{1}}\cap\{v\geq 1\}}|\partial_v\phi|\,dv+\int_{\gamma_{r_{1}}\cap\{u\geq u_{r_{1}}(1)\}}|\partial_u\phi|\,du\right.\notag\\
& & \left.+\int_{u_{R}(1)} ^{u_{r_{1}}(1)}|\partial_u\phi|(u,1)\,du\right).\\
&\leq&2C^2\left(\int_{\mathcal{H}^+\cap\{v\geq1\}}|\partial_v\phi|\,dv+\int_{u_{r_1}(1)}^\infty|\partial_u\phi|(u,1)\,du\right)\\
& &+C\int_{u_{R}(1)} ^{u_{r_{1}}(1)}|\partial_u\phi|(u,1)\,du\\
&=&C_2\left(\int_{\mathcal{H}^+\cap\{v\geq1\}}|\partial_v\phi|\,dv+\int_{u_{R}(1)}^\infty|\partial_u\phi|(u,1)\,du\right),
\end{eqnarray*}
as required.

On the other hand, if $N-1>1$, we use \eqref{jv} and \eqref{ju} with $j=N-2$ to estimate the first two terms on the right hand side of \eqref{unused}:
\begin{eqnarray*}
\int_{\gamma_{R}\cap\{v\geq 1\}}|\partial_v\phi|\,dv&\leq& 2C^2\left(\int_{\gamma_{r_{N-2}}\cap\{v\geq 1\}}|\partial_v\phi|\,dv+\int_{\gamma_{r_{N-2}}\cap\{u\geq u_{r_{N-2}}(1)\}}|\partial_u\phi|\,du\right.\\
& & \left.+\int_{u_{r_{N-1}}(1)} ^{u_{r_{N-2}}(1)}|\partial_u\phi|(u,1)\,du\right)+C\int_{u_{r_{N}}(1)} ^{u_{r_{N-1}}(1)}|\partial_u\phi|(u,1)\,du\\
&\leq&C_2\left(\int_{\gamma_{r_{N-2}}\cap\{v\geq 1\}}|\partial_v\phi|\,dv+\int_{\gamma_{r_{N-2}}\cap\{u\geq u_{r_{N-2}}(1)\}}|\partial_u\phi|\,du\right.\\
& & \left.+\int_{u_{r_{N}}(1)} ^{u_{r_{N-2}}(1)}|\partial_u\phi|(u,1)\,du\right).
\end{eqnarray*}
Proceeding inductively yields
\begin{eqnarray*}
\int_{\gamma_{R}\cap\{v\geq 1\}}|\partial_v\phi|\,dv&\leq&C_{N-1}\left(\int_{\gamma_{r_{1}}\cap\{v\geq 1\}}|\partial_v\phi|\,dv+\int_{\gamma_{r_{1}}\cap\{u\geq u_{r_{1}}(1)\}}|\partial_u\phi|\,du\right.\\
& & \left.+\int_{u_{r_{N}}(1)} ^{u_{r_{1}}(1)}|\partial_u\phi|(u,1)\,du\right).
\end{eqnarray*}
And hence, by \eqref{1v} and \eqref{1u}, we have
\begin{eqnarray*}
\int_{\gamma_{R}\cap\{v\geq 1\}}|\partial_v\phi|\,dv&\leq& C_N\left(\int_{\mathcal{H}^+\cap\{v\geq1\}}|\partial_v\phi|\,dv+\int_{u_{r_1}(1)}^\infty|\partial_u\phi|(u,1)\,du\right.\\
& &\left.+\int_{u_{r_{N}}(1)} ^{u_{r_{1}}(1)}|\partial_u\phi|(u,1)\,du\right)\\
&=&C_N\left(\int_{\mathcal{H}^+\cap\{v\geq1\}}|\partial_v\phi|\,dv+\int_{u_{R}(1)}^\infty|\partial_u\phi|(u,1)\,du\right).
\end{eqnarray*}
Note that the constant $C_N$ depends on $N$, and hence on $R$. So we have shown \eqref{jjv}. The proof of the second statement \eqref{indarg2} is similar.
\end{proof}

\subsubsection{$L^2$ Estimates}
We now turn our attention to proving the $L^2$-type estimates for the derivatives of a smooth, spherically symmetric solution $\phi$ that we need in order to prove Theorem \ref{lp exterior estimate}.  Note that estimates for such quantities are obtained in Section 3 of \cite{LO}. However, we stress that the estimates of \cite{LO} are not suitable for our purposes due to our need to interpolate the $L^2$-type estimates with their $L^1$-type counterparts. As we wish to interpret the $L^2$ estimates as a statement about the boundedness of a certain operator between normed spaces, the structure of the estimate is vitally important to us: the left hand side of our estimate must correspond to the norm on the target space while the the right hand side must be a constant multiple of the norm on the domain of the operator. The estimates of \cite{LO} do not take this form (there it was only necessary to show that the left hand side could be bounded given certain assumptions on the initial data) and for this reason we must derive new (albeit similar) estimates to the $L^2$-type estimates established in \cite{LO}. 

As in the $L^1$ case, to obtain an estimate for an integral over $\gamma_R$ for $R>r_+$ large, we must again first deduce the estimate for curves $\gamma_r$ sufficiently close to the event horizon and then propagate this estimate to the curve $\gamma_R$.
We begin by estimating 
\[\int_{\gamma_{R_2}\cap\{v\geq1\}} v^\alpha(\partial_v\phi)^2\,dv,\]
where $R_2$ is sufficiently close to $r_+$.

\begin{proposition}\label{L21}
Suppose $\phi$ is a solution of the wave equation \eqref{waveequation} with smooth, spherically symmetric data on $S$ and let $\alpha>0$. If $R_2>r_+$ is  sufficiently close to $r_+$ then, for any $1\leq v_1<v_2$ and $u_2= u_{R_2}(v_2)$, we have
\[\int_{v_1}^{v_2}\sup_{u\in [u_{R_2}(v),u_2]}v^\alpha(\partial_v\phi)^2\,dv\leq C\left(\int_{v_1}^{v_2}v^\alpha(\partial_v\phi)^2(u_2,v)\,dv+\int_{u_{R_2}(v_1)}^{u_2}v_1^\alpha\left(\frac{\partial_u\phi}{\Omega}\right)^2 (u,v_1)\,du\right),  \]
where $C=C(R_2,\alpha)>0$. In particular,
\begin{eqnarray}
\int_{\gamma_{R_2}\cap\{v\geq1\}} v^\alpha(\partial_v\phi)^2\, &=&\int_1^\infty v^\alpha(\partial_v\phi)^2(u_{R_2}(v),v)\,dv \notag\\
&\leq&C\left(\int_{\mathcal{H}^+\cap\{v\geq1\}}v^\alpha(\partial_v\phi)^2+\int_{u_{R_2}(1)}^{\infty}\left(\frac{\partial_u\phi}{\Omega}\right)^2 (u,1)\,du\right),   \label{unlabelled}
\end{eqnarray}
where $C=C(R_2,\alpha)>0$.
\end{proposition}

 The Penorse diagram below depicts the scenario of interest. This proposition and its proof are reminiscent of Proposition 3.2 of \cite{LO}, though we emphasise that our estimate differs from the estimate obtained there. Indeed, the final term on the right hand side of \eqref{unlabelled} does not appear in the estimate in \cite{LO}, where a $\sup_u \left(\frac{\partial_u\phi}{\Omega^2}\right)^2$ term takes its place. Because we intend to interpolate, replacing the supremum term with the $L^2$-type term is essential. As mentioned above, it allows us to show boundedness of an appropriate operator in the $L^2$ endpoint case.

\begin{figure}[h]
\centering
\includegraphics[width=0.5\textwidth]{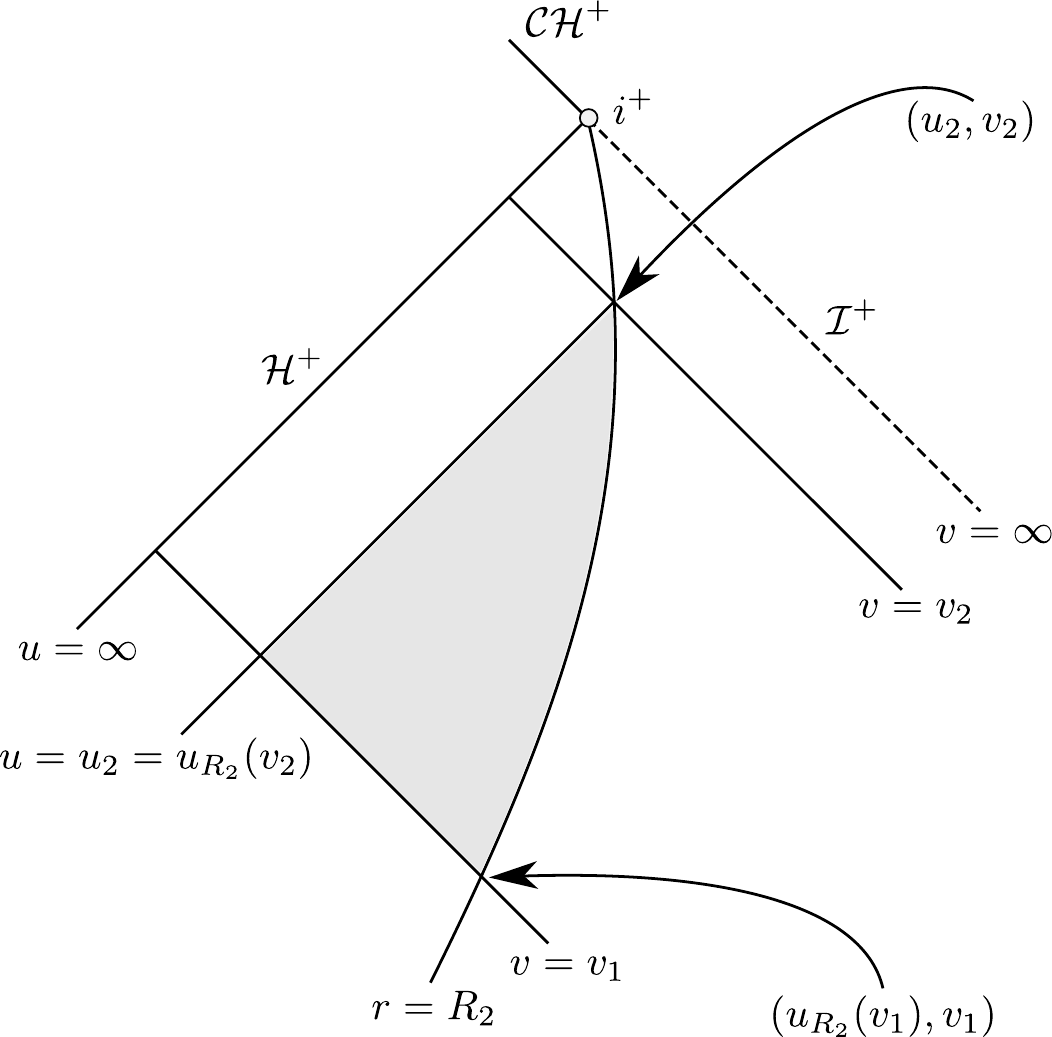}
\end{figure}

\begin{proof}
By the wave equation \eqref{exteriorwave}, we have
\begin{eqnarray*}
\partial_v\left(\frac{\partial_u\phi}{\Omega}\right) = \frac{\Omega\partial_v\partial_u\phi-\partial_u\phi\partial_v\Omega}{\Omega^2}&=& -\frac{\partial_u\phi}{\Omega} \left(\frac{\Omega^2}{r}+\frac{\partial_v\Omega}{\Omega}\right)+\frac{\Omega}{r}\partial_v\phi.
\end{eqnarray*}
Let $V>0$ be a large constant which we will choose later on, and multiply the previous equation by $(v+V)^\alpha \frac{\partial_u\phi}{\Omega}$ to get
\begin{eqnarray}
\frac{1}{2}\partial_v\left((v+V)^\alpha \left(\frac{\partial_u\phi}{\Omega}\right)^2\right) = (v+V)^\alpha\left(\frac{\partial_u\phi}{\Omega}\right)\partial_v\left(\frac{\partial_u\phi}{\Omega}\right) +\frac{\alpha}{2}(v+V)^{\alpha-1}\left(\frac{\partial_u\phi}{\Omega}\right)^2\notag\\
=
-\left(\frac{\partial_u\phi}{\Omega}\right)^2\left(\frac{\Omega^2}{r}+\frac{\partial_v\Omega}{\Omega}-\frac{\alpha}{2(v+V)}\right)(v+V)^\alpha +\partial_v\phi \partial_u\phi\frac{(v+V)^\alpha}{r}.\hspace{2mm}\label{l2i1'}
\end{eqnarray}
Now, $\frac{\partial_v\Omega}{\Omega}=\frac{\partial_v\Omega^2}{2\Omega^2}=\frac{1}{r_+^2}(M-\frac{e^2}{r_+})>0$ on $\mathcal{H}^+$. Thus, if $R_3>r_+$ is sufficiently close to $r_+$, there is a constant $c>0$ such that 
\[\frac{\Omega^2}{r}+\frac{\partial_v\Omega}{\Omega}\geq c \hspace{3mm}\mbox{ for } \hspace{3mm}r_+\leq r \leq R_3,\]
and choosing a smaller $R_3$ if necessary, we can ensure $R_3<R_1$, where $R_1$ is as in Theorem \ref{contrat}.
Thus by fixing $V=V(\alpha)>0$ large enough, we can ensure
\[\frac{\Omega^2}{r}+\frac{\partial_v\Omega}{\Omega}-\frac{\alpha}{2(v+V)}\geq \frac{c}{2} \hspace{3mm}\mbox{for }\hspace{3mm} r_+\leq r\leq R_3.\] 

Now let $r_+<R_2<R_3$ to be chosen later and suppose  $1\leq v_1<v_2$ and $u_2= u_{R_2}(v_2)$, as in the figure. Then integrating \eqref{l2i1'} over the spacetime region $\{(u,v): u_{R_2}(v)\leq u \leq u_2, v_1\leq v\leq v_2\}$ with respect to the measure $dv\,du$ and using the above lower bound together with the estimates $v\leq v+V\leq (V+1)v$ gives
\begin{eqnarray}
\int_{u_{R_2}(v_1)}^{u_2}\int_{v_1}^{v_{R_2}(u)} \partial_v((v+V)^\alpha \left(\frac{\partial_u\phi}{\Omega}\right)^2)\,dv\,du+c\int_{v_1}^{v_2}\int_{u_{R_2}(v)}^{u_2}v^\alpha\left(\frac{\partial_u\phi}{\Omega}\right)^2\,du\,dv\notag\\
\leq \frac{2(V+1)^\alpha}{r_+}\int_{v_1}^{v_2}\int_{u_{R_2}(v)}^{u_2}v^\alpha |\partial_v\phi\partial_u\phi|\,du\,dv.\hspace{5mm}\label{estim4}
\end{eqnarray}
By the Cauchy--Schwarz inequality,
\begin{eqnarray}
\int_{v_1}^{v_2}\int_{u_{R_2}(v)}^{u_2}v^\alpha |\partial_v\phi\partial_u\phi|\,du\,dv\hspace{90mm}\notag\\
\leq \int_{v_1}^{v_2}\left(\int_{u_{R_2}(v)}^{u_2}v^\alpha(\partial_v\phi)^2\Omega^2\,du\right)^{1/2} \left(\int_{u_{R_2}(v)}^{u_2}v^\alpha\left(\frac{\partial_u\phi}{\Omega}\right)^2\,du\right)^{1/2} \,dv\hspace{12mm}\notag\\
\leq \frac{1}{2}\left[\varepsilon^{-1}\int_{v_1}^{v_2}\int_{u_{R_2}(v)}^{u_2}v^\alpha(\partial_v\phi)^2\Omega^2\,du\,dv +\varepsilon\int_{v_1}^{v_2} \int_{u_{R_2}(v)}^{u_2}v^\alpha\left(\frac{\partial_u\phi}{\Omega}\right)^2\,du\,dv\right],\hspace{1mm}\notag
\end{eqnarray}
where $\varepsilon>0$ is chosen small enough that $\frac{\varepsilon(V+1)^\alpha}{r_+}< c$. Thus, by \eqref{estim4}
\begin{eqnarray}
\int_{u_{R_2}(v_1)}^{u_2}  (v_{R_2}(u))^\alpha \left(\frac{\partial_u\phi}{\Omega}\right)^2 (u,v_{R_2}(u))\,du+c\int_{v_1}^{v_2}\int_{u_{R_2}(v)}^{u_2}v^\alpha\left(\frac{\partial_u\phi}{\Omega}\right)^2\,du\,dv\hspace{15mm}\hspace{20mm}\notag\\
\leq (V+1)^\alpha\int_{u_{R_2}(v_1)}^{u_2} v_1^\alpha \left(\frac{\partial_u\phi}{\Omega}\right)^2(u,v_1)\,du + \frac{(V+1)^\alpha}{r_+}\left[\varepsilon^{-1}\int_{v_1}^{v_2}\int_{u_{R_2}(v)}^{u_2}v^\alpha(\partial_v\phi)^2\Omega^2\,du\,dv\right. \notag\\
\left.+\varepsilon\int_{v_1}^{v_2} \int_{u_{R_2}(v)}^{u_2}v^\alpha\left(\frac{\partial_u\phi}{\Omega}\right)^2\,du\,dv\right].\hspace{80mm}\notag
\end{eqnarray}
As $\varepsilon>0$ was chosen so that $\frac{\varepsilon(V+1)^\alpha}{r_+}< c$, we deduce
\begin{eqnarray}
\int_{u_{R_2}(v_1)}^{u_2}  (v_{R_2}(u))^\alpha \left(\frac{\partial_u\phi}{\Omega}\right)^2 (u,v_{R_2}(u))\,du+\int_{v_1}^{v_2}\int_{u_{R_2}(v)}^{u_2}v^\alpha\left(\frac{\partial_u\phi}{\Omega}\right)^2\,du\,dv\hspace{40mm}\notag\\
\leq C\left(\int_{u_{R_2}(v_1)}^{u_2} v_1^\alpha \left(\frac{\partial_u\phi}{\Omega}\right)^2(u,v_1)\,du+\int_{v_1}^{v_2}\int_{u_{R_2}(v)}^{u_2}v^\alpha(\partial_v\phi)^2\Omega^2\,du\,dv\right)\hspace{31mm}\notag\\
\leq C\left(\int_{u_{R_2}(v_1)}^{u_2} v_1^\alpha \left(\frac{\partial_u\phi}{\Omega}\right)^2(u,v_1)\,du+ \int_{v_1}^{v_2}\sup_{u\in[u_{R_2}(v),u_2]}v^\alpha(\partial_v\phi)^2\,dv\sup_{v\in[v_1,v_2]}\int_{u_{R_2}(v)}^{u_2}\Omega^2\,du\right)\notag\\
=C\left(\int_{u_{R_2}(v_1)}^{u_2} v_1^\alpha \left(\frac{\partial_u\phi}{\Omega}\right)^2(u,v_1)\,du+ (R_2-r_+)\int_{v_1}^{v_2}\sup_{u\in[u_{R_2}(v),u_2]}v^\alpha(\partial_v\phi)^2\,dv\right),\hspace{16mm} \label{estim5}
\end{eqnarray}
where $C=C(\alpha)>0$.

On the other hand, from the wave equation \eqref{extwe}, it follows that
\[\partial_u(r\partial_v\phi)=-\partial_vr\partial_u\phi.\]
Multiplying by $v^\alpha r\partial_v\phi$ and integrating with respect to $u\in[u',u_2]$ for any $u_2\geq u'\geq u_{R_2}(v)$ gives
\begin{eqnarray*}
\sup_{u'\in[u_{R_2}(v),u_2]}r^2v^\alpha(\partial_v\phi)^2(u',v)
\leq r^2v^\alpha(\partial_v\phi)^2(u_2,v)+2R_2\int_{u_{R_2}(v)}^{u_2}|v^\alpha\partial_v\phi\partial_u\phi\partial_vr|\,du.
\end{eqnarray*}
Integrating with respect to $v$ gives
\begin{eqnarray}
\int_{v_1}^{v_2}\sup_{u'\in[u_{R_2}(v),u_2]}v^\alpha(\partial_v\phi)^2(u',v)\,dv\hspace{68mm}\notag\\
\leq C\left(\int_{v_1}^{v_2}v^\alpha(\partial_v\phi)^2(u_2,v)\,dv+\int_{v_1}^{v_2}\int_{u_{R_2}(v)}^{u_2}|v^\alpha\partial_v\phi\partial_u\phi\partial_vr|\,du\,dv\right),\label{estim6}
\end{eqnarray}
where $C=C(R_2)>0$ increases in $R_2$.

But using the Cauchy--Schwarz inequality and the fact that $(\partial_vr)^2=\Omega^4\leq 1$ in the black hole exterior, we deduce
\begin{eqnarray*}
\int_{v_1}^{v_2}\int_{u_{R_2}(v)}^{u_2}|v^\alpha\partial_v\phi\partial_u\phi\partial_vr|\,du\,dv\hspace{100mm}\\
\leq \int_{v_1}^{v_2}\sup_{u\in[u_{R_2}(v),u_2]}v^{\alpha/2}|\partial_v\phi|\int_{u_{R_2}(v)}^{u_2}v^{\alpha/2}\left(\frac{|\partial_u\phi|}{\Omega}\right)\Omega|\partial_vr|\,du\,dv\hspace{52mm}\\
\leq \left(\int_{v_1}^{v_2}\sup_{u\in[u_{R_2}(v),u_2]}v^\alpha(\partial_v\phi)^2\,dv\right)^{1/2}\left(\int_{v_1}^{v_2}\left(\int_{u_{R_2}(v)}^{u_2}v^{\alpha/2}\left(\frac{|\partial_u\phi|}{\Omega}\right)\Omega|\partial_vr|\,du\right)^2\,dv\right)^{1/2}\hspace{12mm}\\
\leq \frac{\varepsilon_0}{2}\int_{v_1}^{v_2}\sup_{u\in[u_{R_2}(v),u_2]}v^\alpha(\partial_v\phi)^2\,dv+\frac{\varepsilon_0^{-1}}{2}\int_{v_1}^{v_2}\left(\int_{u_{R_2}(v)}^{u_2}v^{\alpha/2}\left(\frac{|\partial_u\phi|}{\Omega}\right)\Omega|\partial_vr|\,du\right)^2\,dv\hspace{18mm}\\
\leq \frac{\varepsilon_0}{2}\int_{v_1}^{v_2}\sup_{u\in[u_{R_2}(v),u_2]}v^\alpha(\partial_v\phi)^2\,dv +\frac{\varepsilon_0^{-1}}{2}\int_{v_1}^{v_2}\left(\int_{u_{R_2}(v)}^{u_2}v^\alpha\left(\frac{\partial_u\phi}{\Omega}\right)^2\,du\cdot\int_{u_{R_2}(v)}^{u_2}\Omega^2(\partial_vr)^2\,du\right)\,dv\\
\leq \frac{\varepsilon_0}{2}\int_{v_1}^{v_2}\sup_{u\in[u_{R_2}(v),u_2]}v^\alpha(\partial_v\phi)^2\,dv +\frac{\varepsilon_0^{-1}}{2}\int_{v_1}^{v_2}\left(\int_{u_{R_2}(v)}^{u_2}v^\alpha\left(\frac{\partial_u\phi}{\Omega}\right)^2\,du\cdot\int_{u_{R_2}(v)}^{u_2}\Omega^2\,du\right)\,dv\hspace{10mm}\\
\leq \frac{\varepsilon_0}{2}\int_{v_1}^{v_2}\sup_{u\in[u_{R_2}(v),u_2]}v^\alpha(\partial_v\phi)^2\,dv +\frac{\varepsilon_0^{-1}(R_2-r_+)}{2}\int_{v_1}^{v_2}\int_{u_{R_2}(v)}^{u_2}v^\alpha\left(\frac{\partial_u\phi}{\Omega}\right)^2\,du\,dv\hspace{23mm}
\end{eqnarray*}
for any $\varepsilon_0>0$.
So using \eqref{estim5} to estimate the second term on the right hand side, we have
\begin{eqnarray*}
\int_{v_1}^{v_2}\int_{u_{R_2}(v)}^{u_2}|v^\alpha\partial_v\phi\partial_u\phi\partial_vr|\,du\,dv\hspace{90mm}\\
\leq \frac{\varepsilon_0}{2}\int_{v_1}^{v_2}\sup_{u\in[u_{R_2}(v),u_2]}v^\alpha(\partial_v\phi)^2\,dv
+\, C_{R_2}\int_{u_{R_2}(v_1)}^{u_2} v_1^\alpha \left(\frac{\partial_u\phi}{\Omega}\right)^2(u,v_1)\,du\\
+ \,\tilde{C}\varepsilon_0^{-1}(R_2-r_+)^2\int_{v_1}^{v_2}\sup_{u\in[u_{R_2}(v),u_2]}v^\alpha(\partial_v\phi)^2\,dv,\hspace{31mm}
\end{eqnarray*}
for some universal constant $\tilde{C}$.
Thus, choosing $\varepsilon_0$ sufficiently small and $R_2$ sufficiently close to $r_+$ that
\[C(R_2)\left(\frac{\varepsilon_0}{2}+\tilde{C}\varepsilon_0^{-1}(R_2-r_+)^2\right)<1,\]
 it follows from \eqref{estim6} that
\begin{eqnarray*}
\int_{v_1}^{v_2}\sup_{u'\in[u_{R_2}(v),u_2]}v^\alpha(\partial_v\phi)^2(u',v)\,dv
\leq C\left(\int_{v_1}^{v_2}v^\alpha(\partial_v\phi)^2(u_2,v)\,dv+ \int_{u_{R_2}(v_1)}^{u_2} v_1^\alpha \left(\frac{\partial_u\phi}{\Omega}\right)^2(u,v_1)\,du\right),
\end{eqnarray*}
where $C=C(R_2,\alpha)>0$, as desired.

For the second statement \eqref{unlabelled}, set $v_1=1$ and $v_2=\infty$ so that $u_2=u_{R_2}(v_2)=\infty$. This gives
\begin{eqnarray*}
\int_{1}^{\infty}\sup_{u'\in[u_{R_2}(v),\infty]}v^\alpha(\partial_v\phi)^2(u',v)\,dv
\leq C\left(\int_{1}^{\infty}v^\alpha(\partial_v\phi)^2(\infty,v)\,dv+ \int_{u_{R_2}(1)}^{\infty}  \left(\frac{\partial_u\phi}{\Omega}\right)^2(u,1)\,du\right). 
\end{eqnarray*}
So 
\begin{eqnarray}
\int_{\gamma_{R_2}\cap\{v\geq1\}} v^\alpha(\partial_v\phi)^2\,dv &=&\int_1^\infty v^\alpha(\partial_v\phi)^2(u_{R_2}(v),v)\,dv \notag\\
&\leq&C\left(\int_{1}^{\infty}v^\alpha(\partial_v\phi)^2(\infty,v)\,dv+\int_{u_{R_2}(1)}^{\infty}\left(\frac{\partial_u\phi}{\Omega}\right)^2 (u,1)\,du\right),\notag  \hspace{7mm}
\end{eqnarray}
where $C=C(R_2,\alpha)>0$.
\end{proof}

In order to prove a corresponding estimate to Proposition \ref{L21} for $\partial_u\phi$, we first need to prove a similar estimate for $\partial_v\phi$ on a constant $u$-hypersurface instead of a constant $r$-curve. This is done in the next proposition. 

\begin{proposition}\label{similarprop}
Suppose $\phi$ is a solution of the wave equation \eqref{waveequation} with smooth, spherically symmetric data on $S$. Let $\alpha>0$. Then there exists $R_2>r_+$ sufficiently close to $r_+$ such that if $[u_1,u_2]\times[v_1,v_2]\subseteq \{r_+\leq r\leq R_2\}\cap\{v\geq1\}$, then there is a constant $C=C(R_2, \alpha)>0$ such that
\[\int_{v_1}^{v_2}v^\alpha(\partial_v\phi)^2(u_1,v)\,dv\leq C\left(\int_{v_1}^{v_2}v^\alpha(\partial_v\phi)^2(u_2,v)\,dv+\int_{u_1}^{u_2}v_1^\alpha\left(\frac{\partial_u\phi}{\Omega}\right)^2(u,v_1)\,du\right).\]
In particular, if $u'>u_{R_2}(1)$, then
\begin{equation}\label{vlih2}
\int_1^{v_{{R_2}(u')}}v^\alpha(\partial_v\phi)^2(u',v)\,dv\leq C\left(\int_{1}^{\infty}v^\alpha(\partial_v\phi)^2(\infty,v)\,dv+\int_{u_{R_2}(1)}^{\infty}\left(\frac{\partial_u\phi}{\Omega}\right)^2(u,1)\,du\right),
\end{equation}
where $C=C(R_2,\alpha)>0$.
\end{proposition}

\noindent The proof of this result is very similar to the proof of Proposition \ref{L21}, the only difference is the region of integration.

\begin{proof}
Arguing exactly as in the proof of Proposition \ref{L21}, we have that if $V>0$, then 
\begin{eqnarray}
\frac{1}{2}\partial_v\left((v+V)^\alpha \left(\frac{\partial_u\phi}{\Omega}\right)^2\right) = (v+V)^\alpha\left(\frac{\partial_u\phi}{\Omega}\right)\partial_v\left(\frac{\partial_u\phi}{\Omega}\right) +\frac{\alpha}{2}(v+V)^{\alpha-1}\left(\frac{\partial_u\phi}{\Omega}\right)^2\notag\\
=
-\left(\frac{\partial_u\phi}{\Omega}\right)^2\left(\frac{\Omega^2}{r}+\frac{\partial_v\Omega}{\Omega}-\frac{\alpha}{2(v+V)}\right)(v+V)^\alpha +\partial_v\phi \partial_u\phi\frac{(v+V)^\alpha}{r},\label{l2i1}
\end{eqnarray}
and fixing $V=V(\alpha)>0$ large enough, we can ensure
\[\frac{\Omega^2}{r}+\frac{\partial_v\Omega}{\Omega}-\frac{\alpha}{2(v+V)}\geq \frac{c}{2} \hspace{3mm}\mbox{for }\hspace{3mm} r_+\leq r\leq R_3,\] 
for $R_3<R_1$ sufficiently close to $r_+$.

Now let $r_+<R_2<R_3$ to be chosen later and let $[u_1,u_2]\times[v_1,v_2]\subseteq \{r_+\leq  r\leq R_2\}\cap\{v\geq1\}$, as shown in the diagram.
\begin{figure}[H]
\centering
\includegraphics[width=0.45\textwidth]{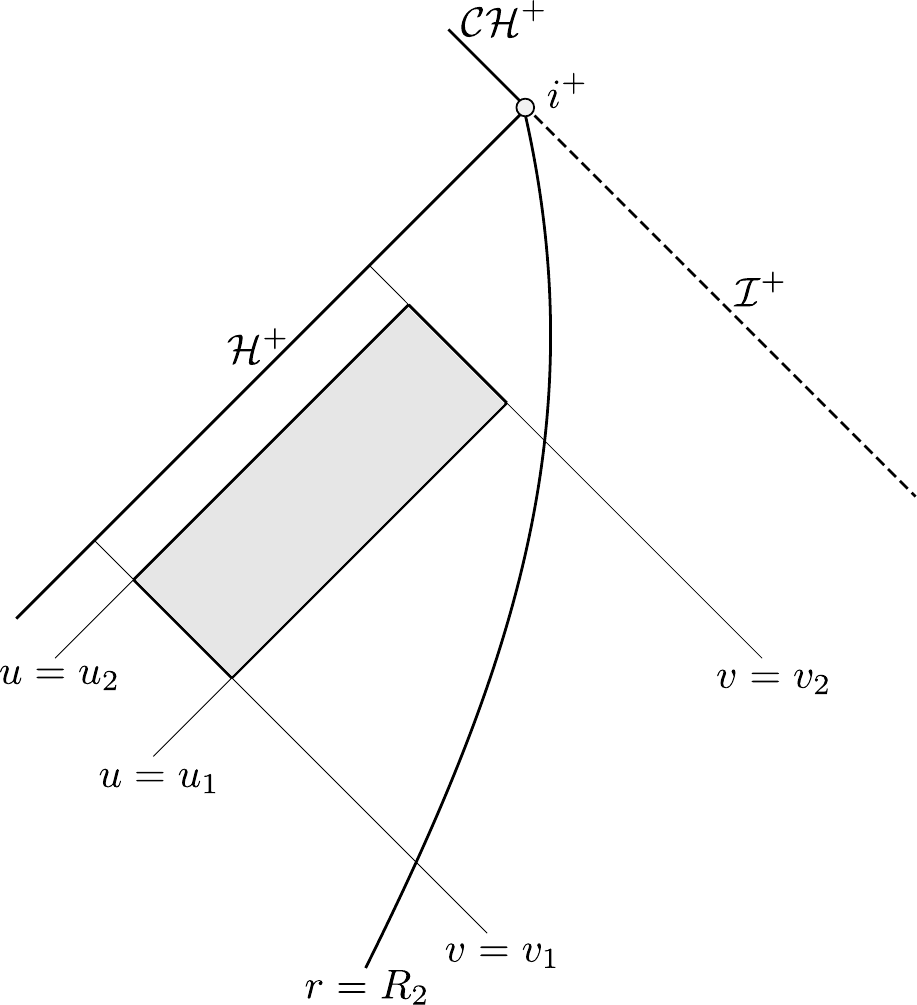}
\end{figure}
\noindent For $u\in[u_1,u_2]$, we integrate \eqref{l2i1} over $\{u\}\times[v_1,v_2]$ and use the inequality $v\leq v+V\leq (V+1)v$ for $v\in[v_1,v_2]$, together with the Cauchy--Schwarz inequality to deduce
\begin{eqnarray*}
v_2^\alpha \left(\frac{\partial_u\phi}{\Omega}\right)^2(u,v_2)+c\int_{v_1}^{v_2}v^\alpha \left(\frac{\partial_u\phi}{\Omega}\right)^2\,dv\hspace{85mm}\\
\leq 2(V+1)^\alpha\left[v_1^\alpha \left(\frac{\partial_u\phi}{\Omega}\right)^2(u,v_1)+\int_{v_1}^{v_2}\frac{v^\alpha\Omega}{r}\partial_v\phi\frac{\partial_u\phi}{\Omega}\,dv\right]\hspace{52mm}\\
\leq 2 (V+1)^\alpha\left[v_1^\alpha \left(\frac{\partial_u\phi}{\Omega}\right)^2(u,v_1) +\left(\int_{v_1}^{v_2}v^\alpha (\partial_v\phi)^2\frac{\Omega^2}{r^2}\,dv\right)^{1/2}   \left(\int_{v_1}^{v_2}v^\alpha\left(\frac{ \partial_u\phi}{\Omega}\right)^2\,dv\right)^{1/2}  \right]\hspace{2mm}\\
\leq 2 (V+1)^\alpha\left[v_1^\alpha \left(\frac{\partial_u\phi}{\Omega}\right)^2(u,v_1) + \frac{\varepsilon^{-1}}{2r_+^2}\int_{v_1}^{v_2}v^\alpha(\partial_v\phi)^2\Omega^2\,dv+ \frac{\varepsilon}{2}\int_{v_1}^{v_2}v^\alpha\left(\frac{\partial_u\phi}{\Omega}\right)^2\,dv \right],\hspace{8mm}
\end{eqnarray*}
where $\varepsilon>0$. Thus, choosing $\varepsilon$ small enough so that 
\[c>(V+1)^\alpha\varepsilon,\]
we have that
\begin{equation}\label{step}
v_2^\alpha \left(\frac{\partial_u\phi}{\Omega}\right)^2(u,v_2)+\int_{v_1}^{v_2}v^\alpha \left(\frac{\partial_u\phi}{\Omega}\right)^2\,dv \leq C\left(v_1^\alpha \left(\frac{\partial_u\phi}{\Omega}\right)^2(u,v_1) + \int_{v_1}^{v_2}v^\alpha(\partial_v\phi)^2\Omega^2(u,v)\,dv\right)
\end{equation}
for $u\in [u_1,u_2]$, where $C=C(\alpha)$.\newpage Integrating over $u\in[u_1,u_2]$ gives
\begin{eqnarray}
\int_{u_1}^{u_2} v_2^\alpha \left(\frac{\partial_u\phi}{\Omega}\right)^2(u,v_2)\,du+\int_{u_1}^{u_2}\int_{v_1}^{v_2}v^\alpha \left(\frac{\partial_u\phi}{\Omega}\right)^2\,dv\,du\hspace{57mm}\notag\\
 \leq C\left(\int_{u_1}^{u_2}v_1^\alpha \left(\frac{\partial_u\phi}{\Omega}\right)^2(u,v_1)\,du + \int_{u_1}^{u_2}\int_{v_1}^{v_2}v^\alpha(\partial_v\phi)^2\Omega^2\,dv\,du\right)\hspace{20mm}\notag\\
  \leq C\left(\int_{u_1}^{u_2}v_1^\alpha \left(\frac{\partial_u\phi}{\Omega}\right)^2(u,v_1)\,du + \int_{v_1}^{v_2}\sup_{u\in[u_1,u_2]} v^\alpha (\partial_v\phi)^2\int_{u_1}^{u_2}(-\partial_ur)\,du\,dv\right)\notag\\
  \leq C\left(\int_{u_1}^{u_2}v_1^\alpha \left(\frac{\partial_u\phi}{\Omega}\right)^2(u,v_1)\,du +(R_2-r_+) \int_{v_1}^{v_2}\sup_{u\in[u_1,u_2]} v^\alpha (\partial_v\phi)^2\,dv\right),\hspace{6mm}\label{star1}
\end{eqnarray}
where $C=C(\alpha)$.

As before, from the wave equation \eqref{extwe}, it follows that
\[\partial_u(r\partial_v\phi)=-\partial_vr\partial_u\phi.\]
Multiplying by $v^\alpha r\partial_v\phi$ and integrating with respect to $u\in[u',u_2]$ for any $ u'\geq u_1$, and then taking the supremum over $u'\in[u_1,u_2]$ gives 
\begin{eqnarray*}
\sup_{u'\in[u_1,u_2]}r^2v^\alpha(\partial_v\phi)^2(u',v)
\leq r^2v^\alpha(\partial_v\phi)^2(u_2,v)+2R_2\int_{u_1}^{u_2}|v^\alpha\partial_v\phi\partial_u\phi\partial_vr|\,du.
\end{eqnarray*}
Integrating with respect to $v$ then gives
\begin{eqnarray}
\int_{v_1}^{v_2}\sup_{u'\in[u_1,u_2]}v^\alpha(\partial_v\phi)^2(u',v)\,dv\hspace{68mm}\notag\\
\leq C'\left(\int_{v_1}^{v_2}v^\alpha(\partial_v\phi)^2(u_2,v)\,dv+\int_{v_1}^{v_2}\int_{u_1}^{u_2}|v^\alpha\partial_v\phi\partial_u\phi\partial_vr|\,du\,dv\right),\label{star2}
\end{eqnarray}
where $C'=C'(R_2)>0$ increases in $R_2$.
But arguing exactly as in Proposition \ref{L21}, we deduce
\begin{eqnarray*}
\int_{v_1}^{v_2}\int_{u_1}^{u_2}|v^\alpha\partial_v\phi\partial_u\phi\partial_vr|\,du\,dv\hspace{100mm}\\
\leq \frac{\varepsilon_0}{2}\int_{v_1}^{v_2}\sup_{u\in[u_1,u_2]}v^\alpha(\partial_v\phi)^2\,dv +\frac{\varepsilon_0^{-1}(R_2-r_+)}{2}\int_{v_1}^{v_2}\int_{u_1}^{u_2}v^\alpha\left(\frac{\partial_u\phi}{\Omega}\right)^2\,du\,dv\hspace{23mm}
\end{eqnarray*}
for any $\varepsilon_0>0$.
So using \eqref{star1} to estimate the second term on the right hand side of the previous equation, we have
\begin{eqnarray*}
\int_{v_1}^{v_2}\int_{u_1}^{u_2}|v^\alpha\partial_v\phi\partial_u\phi\partial_vr|\,du\,dv
&\leq& \frac{\varepsilon_0}{2}\int_{v_1}^{v_2}\sup_{u\in[u_1,u_2]}v^\alpha(\partial_v\phi)^2\,dv
+\, C''\int_{u_1}^{u_2} v_1^\alpha \left(\frac{\partial_u\phi}{\Omega}\right)^2(u,v_1)\,du\\
& &+ \,\tilde{C}\varepsilon_0^{-1}(R_2-r_+)^2\int_{v_1}^{v_2}\sup_{u\in[u_1,u_2]}v^\alpha(\partial_v\phi)^2\,dv,
\end{eqnarray*}
where $C''=C''(R_2,\alpha)$ and $\tilde{C}=\tilde{C}(\alpha)$.
Thus, choosing $\varepsilon_0$ sufficiently small and $R_2$ sufficiently close to $r_+$ that
\[C'(R_2)\left(\frac{\varepsilon_0}{2}+\tilde{C}\varepsilon_0^{-1}(R_2-r_+)^2\right)<1,\]
 it follows from \eqref{star2} that
\begin{eqnarray*}
\int_{v_1}^{v_2}\sup_{u'\in[u_1,u_2]}v^\alpha(\partial_v\phi)^2(u',v)\,dv
\leq C\left(\int_{v_1}^{v_2}v^\alpha(\partial_v\phi)^2(u_2,v)\,dv+ \int_{u_1}^{u_2} v_1^\alpha \left(\frac{\partial_u\phi}{\Omega}\right)^2(u,v_1)\,du\right),
\end{eqnarray*}
where $C=C(R_2,\alpha)>0$, as desired.

The second statement \eqref{vlih2} follows from the first by taking  $u_1=u'>u_{R_2}(1)$, $u_2=\infty$, $v_1=1$ and $v_2=v_{R_2}(u')$ (so $[u_1,u_2]\times[v_1,v_2]\subseteq \{r_+\leq r\leq R_2\}\cap\{v\geq1\}$), as shown in the Figure below.
\begin{figure}[H]
\centering
\includegraphics[width=0.45\textwidth]{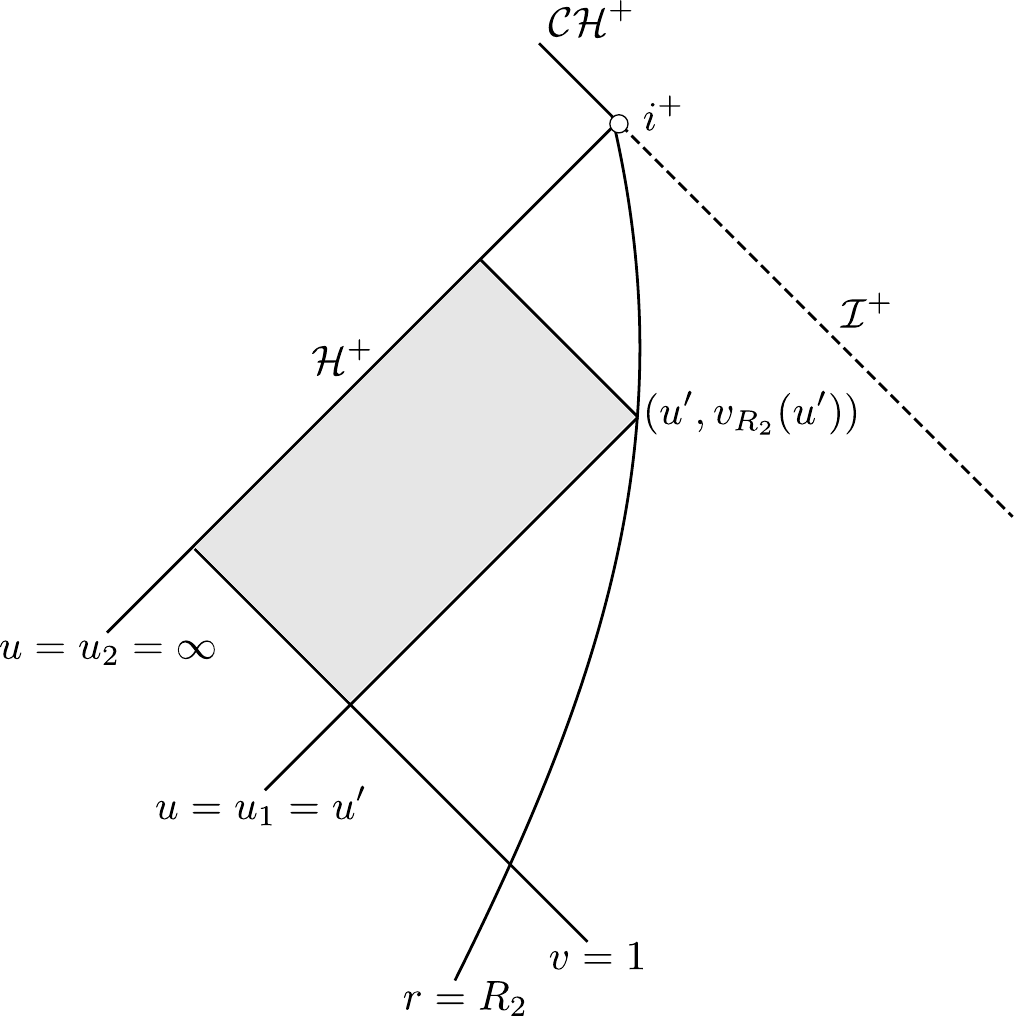}
\end{figure}
\end{proof}

We now use the previous proposition to deduce a  $L^2$-type estimate analogous to Proposition \ref{L21} for $\partial_u\phi$ on a curve $\gamma_{R_2}$ with $R_2$ sufficiently close to $r_+$.   (We remark that no corresponding estimate was deduced in \cite{LO} because there it was sufficient to estimate  $\sup_{r\leq R_2}|\partial_u\phi|$ instead.)

\begin{corollary}\label{l2u}
Suppose $\phi$ is a solution of the wave equation \eqref{waveequation} with smooth, spherically symmetric data on $S$. Let  $\alpha,\varepsilon>0$ be such that $\alpha-(1+\varepsilon)>0$. Then if $R_2>r_+$ is sufficiently close to $r_+$, there is a constant $C=C(R_2,\alpha, \varepsilon)>0$ such that
\[ \int_{\gamma_{R_2}\cap\{v\geq 1\}}v^{\alpha-(1+\varepsilon)}(\partial_u\phi)^2\leq C\left(\int_{\mathcal{H}^+\cap\{v\geq1\}}v^\alpha(\partial_v\phi)^2+\int_{u_{R_2}(1)}^{\infty}\left(\frac{\partial_u\phi}{\Omega}\right)^2(u,1)\,du\right).
\]
\end{corollary}

\begin{proof}
Let $R_2>r_+$ be sufficiently close to $r_+$ such that Proposition \ref{similarprop} (and its proof) hold. Let $u_1=u>u_{R_2}(1)$, $u_2=\infty$, $v_1=1$ and $v_2=v_{R_2}(u)$. Then $[u_1,u_2]\times[v_1,v_2]\subseteq \{r_+\leq r\leq R_2\}\cap\{v\geq1\}$  and so from the proof of Proposition \ref{similarprop} (equation \eqref{step})
\[(v_{R_2}(u))^\alpha \left(\frac{\partial_u\phi}{\Omega}\right)^2(u',v_{R_2}(u)) \leq C(\alpha)\left( \left(\frac{\partial_u\phi}{\Omega}\right)^2(u',1) + \int_{1}^{v_{R_2}(u)}v^\alpha(\partial_v\phi)^2\Omega^2(u',v)\,dv\right)\]
for all $u'\in[u_1,u_2]=[u,\infty]$. In particular, since $\Omega^2\leq1$ in the exterior region, with $u'=u$ we have
\begin{equation}\label{mistake}
(v_{R_2}(u))^\alpha \left(\frac{\partial_u\phi}{\Omega}\right)^2(u,v_{R_2}(u)) \leq C(\alpha)\left( \left(\frac{\partial_u\phi}{\Omega}\right)^2(u,1) + \int_{1}^{v_{R_2}(u)}v^\alpha(\partial_v\phi)^2(u,v)\,dv\right).
\end{equation}
But by the second statement of Proposition \ref{similarprop}, the rightmost integral may be estimated by
\[C(R_2,\alpha)\left(\int_{1}^{\infty}v^\alpha(\partial_v\phi)^2(\infty,v)\,dv+\int_{u_{R_2}(1)}^{\infty}\left(\frac{\partial_u\phi}{\Omega}\right)^2(\overline{u},1)\,d\overline{u}\right).\]
Multiplying \eqref{mistake}  by $(v_{R_2}(u))^{-(1+\varepsilon)}\Omega^2(u,v_{R_2}(u))$ and integrating over $u\in [u_{R_2}(1),\infty]$ yields
\begin{eqnarray*}
\int_{u_{R_2}(1)}^\infty (v_{R_2}(u))^{\alpha-(1+\varepsilon)}(\partial_u\phi)^2(u,v_{R_2}(u))\,du
 \leq C(R_2,\alpha)\left[\int_{u_{R_2}(1)}^\infty \left(\frac{\partial_u\phi}{\Omega}\right)^2(u,1)\frac{\Omega^2(u,v_{R_2}(u))}{(v_{R_2}(u))^{1+\varepsilon}}\,du \right.\hspace{15mm} \\
\left.+\left(\int_1^{v_{R_2}(u)}v^\alpha(\partial_v\phi)^2(\infty,v)\,dv  +\int_{u_{R_2}(1)}^\infty \left(\frac{\partial_u\phi}{\Omega}\right)^2(u',1)\,du'\right)\left(\int_{u_{R_2}(1)}^\infty \frac{\Omega^2(u,v_{R_2}(u))}{(v_{R_2}(u))^{1+\varepsilon}}\,du\right)\right].
\end{eqnarray*}
Since $\displaystyle{\frac{\Omega^2(u,v_{R_2}(u))}{(v_{R_2}(u))^{1+\varepsilon}}}\leq 1$ and since 
\begin{eqnarray*}
 \int_{u_{R_2}(1)}^\infty \frac{\Omega^2(u,v_{R_2}(u))}{(v_{R_2}(u))^{1+\varepsilon}}\,du=\Omega^2_{R_2}\int_1^\infty \frac{1}{{v}^{1+\varepsilon}}\,dv=C(R_2, \varepsilon)<+\infty,
\end{eqnarray*}
we deduce
\begin{eqnarray*}
\int_{u_{R_2}(1)}^\infty (v_{R_2}(u))^{\alpha-(1+\varepsilon)}(\partial_u\phi)^2(u,v_{R_2}(u))\,du
 &\leq &C\left(\int_{u_{R_2}(1)}^\infty \left(\frac{\partial_u\phi}{\Omega}\right)^2(u,1)\,du+\int_1^\infty v^\alpha(\partial_v\phi)^2(\infty,v)\,dv \right),
 \end{eqnarray*}
with $C=C(R_2,\alpha,\varepsilon)$. So
 \begin{eqnarray*}
 \int_{\gamma_{R_2}\cap\{v\geq 1\}}v^{\alpha-(1+\varepsilon)}(\partial_u\phi)^2\,dv &=&  \int_{\gamma_{R_2}\cap\{u\geq u_{R_2}(1)\}}v^{\alpha-(1+\varepsilon)}(\partial_u\phi)^2\,du\\
  &\leq &C\left(\int_{u_{R_2}(1)}^\infty \left(\frac{\partial_u\phi}{\Omega}\right)^2(u,1)\,du+\int_1^\infty v^\alpha(\partial_v\phi)^2(\infty,v)\,dv \right)
 \end{eqnarray*}
with $C=C(R_2,\alpha,\varepsilon)$, as desired.
\end{proof}

We now fix a value of $R_2$ small enough that $R_2<R_1$ and that both Proposition \ref{L21} and Corollary \ref{l2u} hold. What we actually need is to propagate the bounds we have established for $\gamma_{R_2}$ to the surface $\gamma_R$ for any $R>R_2$. This is done in the next proposition, albeit at the expense of some polynomial power. 

\begin{proposition}\label{l2ext22}
Suppose $\phi$ is a solution of the wave equation \eqref{waveequation} with smooth, spherically symmetric data on $S$. Let $\alpha,\varepsilon>0$ be such that $\alpha-(1+2\varepsilon)>0$. Then, for any $R>R_2$, there exist $C=C(R_2,R, \alpha,\varepsilon)>0$ such that
\begin{eqnarray*}
\int_{\gamma_R\cap\{v\geq 1\}}v^{\alpha-(1+2\varepsilon)}((\partial_u\phi)^2+(\partial_v\phi)^2)
\leq C\left(\int_{\mathcal{H}^+\cap\{v\geq1\}}v^\alpha(\partial_v\phi)^2+\int_{\underline{C}_1^{ext}\cap\{u\geq u_R(1)\}}\left(\frac{\partial_u\phi}{\Omega}\right)^2\right).
\end{eqnarray*}
\end{proposition}

\noindent The proof below follows the proof of Proposition 3.3 of \cite{LO}, though uses the estimates we have proved instead of the similar estimates obtained in \cite{LO} in order to yield an estimate with which we can interpolate.

\begin{proof} Fix $R>R_2$ and for $r_+<r<R$ define
\[\gamma_{r^*}^{(v_0)}=\gamma_{r^*}\cap\{v_0-(R^*-r^*)\leq v\leq 2v_0\}.\]
We begin by listing four facts we shall need during the course of the proof.
\begin{enumerate}
\item
During the proof of Proposition 3.3 of \cite{LO}, it is shown that for any $v_0\geq 1$,
\begin{eqnarray}
\sup_{r^*\in[R_2^*,R^*]}\left(\int_{\gamma_{r^*}^{(v_0)}}(\partial_v\phi)^2+(\partial_u\phi)^2\right) + \sup_{v\in [v_0,2v_0]}\int_{\underline{C}_v(R_2^*,R^*)}(\partial_u\phi)^2 
\leq C \left(\int_{\gamma_{R_2^*}^{(v_0)}}(\partial_v\phi)^2+(\partial_u\phi)^2\right),\label{dyadic}
\end{eqnarray}
where
\[\underline{C}_v(R_2^*,R^*)=\{(u,v):u_{R^*}(v)\leq u\leq u_{R_2^*}(v)\}=\underline{C}_v^{ext}\cap\{R_2\leq r \leq R\}\]
and $C=C(R_2,R)>0$. 
\item
For $v_0\geq V_1:=2(R^*-R_2^*)$, we have
\[\gamma_{R_2^*}^{(v_0)}=\gamma_{R_2^*}\cap\{v_0-(R^*-R_2^*)\leq v\leq 2v_0\}\subseteq \gamma_{R_2^*}\cap\{\frac{1}{2}v_0\leq v\leq 2v_0\}.\]
\item
Given $\varepsilon>0$, there exists $V_2=V_2(\varepsilon)>1$ such that 
\[
v\geq V_2 \hspace{3mm}\Longrightarrow \hspace{3mm}v^\varepsilon \geq \log^2(1+v).
\]
\item
Combining Proposition \ref{L21} and Corollary \ref{l2u} gives
\begin{eqnarray}
\int_{\gamma_{R_2^*}\cap\{v\geq 1\}}v^{\alpha-(1+\varepsilon)}((\partial_u\phi)^2+(\partial_v\phi)^2)\hspace{50mm}\notag\\
\leq C\left(\int_{\mathcal{H}^+\cap\{v\geq1\}}v^\alpha(\partial_v\phi)^2+\int_{u_{R_2}(1)}^{\infty}\left(\frac{\partial_u\phi}{\Omega}\right)^2(u,1)\,du\right)\hspace{1mm}\notag\\
\leq C\left(\int_{\mathcal{H}^+\cap\{v\geq1\}}v^\alpha(\partial_v\phi)^2+\int_{\underline{C}_1^{ext}\cap\{u\geq u_R(1)\}}\left(\frac{\partial_u\phi}{\Omega}\right)^2\right),\label{comb}
\end{eqnarray}
for some $C=C(R_2,\alpha,\varepsilon)>0$.
\end{enumerate}
With these facts in mind we set $V=\max(V_1,2V_2)$ and note $V/2>1$. Write
\begin{eqnarray}\label{split''}
\int_{\gamma_R\cap\{v\geq 1\}}v^{\alpha-(1+2\varepsilon)}((\partial_u\phi)^2+(\partial_v\phi)^2)
&=& \int_{\gamma_R\cap\{ 1\leq v\leq V\}}v^{\alpha-(1+2\varepsilon)}((\partial_u\phi)^2+(\partial_v\phi)^2)\notag\\
& & + \int_{\gamma_R\cap\{v\geq V\}}v^{\alpha-(1+2\varepsilon)}((\partial_u\phi)^2+(\partial_v\phi)^2).\hspace{8mm}\label{split}
\end{eqnarray}
We shall estimate the two terms on the right hand side separately.

To estimate  the second term, we use \eqref{dyadic} and the second fact above to deduce 
\begin{eqnarray}
\int_{\gamma_R\cap\{v\geq V\}}v^{\alpha-(1+2\varepsilon)}((\partial_u\phi)^2+(\partial_v\phi)^2)\hspace{58mm} \notag\\
= \sum_{k=0}^\infty\int_{\gamma_R\cap\{2^kV\leq v\leq 2^{k+1} V\}}v^{\alpha-(1+2\varepsilon)}((\partial_u\phi)^2+(\partial_v\phi)^2)\hspace{18mm}\notag\\
\leq \sum_{k=0}^\infty (2^{k+1}V)^{\alpha-(1+2\varepsilon)}\int_{\gamma_{R^*}^{(2^kV)}}((\partial_u\phi)^2+(\partial_v\phi)^2)\hspace{27mm}\notag\\
\leq C\sum_{k=0}^\infty (2^{k+1}V)^{\alpha-(1+2\varepsilon)}\int_{\gamma_{R_2^*}^{(2^kV)}}((\partial_v\phi)^2+(\partial_u\phi)^2)\hspace{24mm}\notag\\
\leq C\sum_{k=0}^\infty (2^{k+1}V)^{\alpha-(1+2\varepsilon)}\int_{\gamma_{R_2^*}\cap\{2^{k-1}V\leq v\leq 2^{k+1}V\}}((\partial_v\phi)^2+(\partial_u\phi)^2)\notag
\end{eqnarray}
since $2^kV\geq V\geq V_1$ for each $k\geq 0$. Thus, by the third fact above
\begin{eqnarray}
\int_{\gamma_R\cap\{v\geq V\}}v^{\alpha-(1+2\varepsilon)}((\partial_u\phi)^2+(\partial_v\phi)^2)\hspace{80mm} \notag\\
\leq  C4^{\alpha-(1+2\varepsilon)} \sum_{k=0}^\infty \int_{\gamma_{R_2^*}\cap\{2^{k-1}V\leq v\leq 2^{k+1}V\}}v^{\alpha-(1+2\varepsilon)}((\partial_v\phi)^2+(\partial_u\phi)^2)\hspace{5mm}\notag\\
\leq C \sum_{k=0}^\infty \int_{\gamma_{R_2^*}\cap\{2^{k-1}V\leq v\leq 2^{k+1}V\}}\frac{v^{\alpha-(1+\varepsilon)}}{\log^2(1+v)}((\partial_v\phi)^2+(\partial_u\phi)^2)\hspace{16mm}\notag\\
=C_V \int_{\gamma_{R_2^*}\cap\{V/2\leq v\leq 4V\}} v^{\alpha-(1+\varepsilon)}  ((\partial_v\phi)^2+(\partial_u\phi)^2) \hspace{32mm}\notag\\
+\,C \sum_{k=2}^\infty \int_{\gamma_{R_2^*}\cap\{2^{k-1}V\leq v\leq 2^{k+1}V\}}\frac{v^{\alpha-(1+\varepsilon)}}{\log^2(1+v)}((\partial_v\phi)^2+(\partial_u\phi)^2)\hspace{12mm},\label{logsum}
\end{eqnarray}
since $2^{k-1}V\geq V/2\geq V_2$. Furthermore, given $v$ such that $2^{k-1}V\leq v\leq 2^{k+1}V$, since $V>1$, we have
\[\log(1+v)\geq \log(1+2^{k-1}V)>\log(2^{k-1})=\frac{1}{\log_2e}(k-1).\]
Substituting this result into \eqref{logsum}, together with the fact $2^{k-1}V\geq V/2>1$ yields
\begin{eqnarray}
\int_{\gamma_R\cap\{v\geq V\}}v^{\alpha-(1+2\varepsilon)}((\partial_u\phi)^2+(\partial_v\phi)^2) \hspace{85mm}\notag\\
\leq C_V \int_{\gamma_{R_2^*}\cap\{V/2\leq v\leq 4V\}} v^{\alpha-(1+\varepsilon)}  ((\partial_v\phi)^2+(\partial_u\phi)^2)\hspace{35mm}\notag\\
+ C \sum_{k=2}^\infty \frac{1}{(k-1)^2} \int_{\gamma_{R_2^*}\cap\{2^{k-1}V\leq v\leq 2^{k+1}V\}}v^{\alpha-(1+\varepsilon)}((\partial_v\phi)^2+(\partial_u\phi)^2),\hspace{4mm}\notag\\
\leq C_V \int_{\gamma_{R_2^*}\cap\{v\geq1\}} v^{\alpha-(1+\varepsilon)}  ((\partial_v\phi)^2+(\partial_u\phi)^2)\hspace{44mm}\notag\\
+   C \sum_{k=2}^\infty \frac{1}{(k-1)^2} \int_{\gamma_{R_2^*}\cap\{v\geq 1\}}v^{\alpha-(1+\varepsilon)}((\partial_v\phi)^2+(\partial_u\phi)^2)\hspace{23mm}\notag\\
= C \int_{\gamma_{R_2^*}\cap\{v\geq 1\}}v^{\alpha-(1+\varepsilon)}((\partial_v\phi)^2+(\partial_u\phi)^2)\hspace{46mm}\notag\\
\leq C\left(\int_{\mathcal{H}^+\cap\{v\geq1\}}v^\alpha(\partial_v\phi)^2\,dv+\int_{u_{R}(1)}^{\infty}\left(\frac{\partial_u\phi}{\Omega}\right)^2(u,1)\,du\right),\hspace{21mm}\label{part2}
\end{eqnarray}
where we have used \eqref{comb} to deduce the last inequality. 

To estimate the first term on the right hand side of \eqref{split''}, note that by the wave equation \eqref{exteriorwave}
\begin{eqnarray}
\frac{1}{2}\partial_v(\partial_u\phi)^2-\frac{1}{2}\partial_u(\partial_v\phi)^2= \partial_u\phi\partial_v\partial_u \phi- \partial_v\phi\partial_u\partial_v\phi 
=(\partial_u\phi-\partial_v\phi)\partial_u\partial_v\phi
&=&-\frac{\Omega^2}{r}(\partial_v\phi-\partial_u\phi)^2\notag\\
&\leq&\frac{\Omega^2}{r}(\partial_v\phi-\partial_u\phi)^2\notag\\
&\leq&C(\partial_v\phi-\partial_u\phi)^2,\notag
\end{eqnarray}
for some $C>0$ since $\Omega^2$ is bounded and $r$ is bounded below in the exterior region.
Hence, in fact
 \begin{equation}\label{preintegration}
 \frac{1}{2}\partial_v(\partial_u\phi)^2-\frac{1}{2}\partial_u(\partial_v\phi)^2 \leq \tilde{C}( (\partial_v\phi)^2+(\partial_u\phi)^2).
 \end{equation}
Given $R>R_2$, we integrate \eqref{preintegration} over the region 
\[Y=\{R_2\leq r\leq R\}\cap\{v\geq 1\}.\]
We have
\begin{eqnarray}
\int_Y \partial_u(\partial_v\phi)^2\,du\,dv&=&\int_1^\infty \int_{u_{R}(v)}^{u_{R_2}(v)}\partial_u(\partial_v\phi)^2\,du\,dv\notag\\
&=&\int_1^\infty(\partial_v\phi)^2(u_{R_2}(v),v)\,dv-\int_1^\infty (\partial_v\phi)^2(u_R(v),v)\,dv\notag\\
&=&\int_{\gamma_{R_2^*}\cap\{v\geq 1\}} (\partial_v\phi)^2- \int_{\gamma_{R^*}\cap\{v\geq 1\}} (\partial_v\phi)^2\notag
\end{eqnarray}
and
\begin{eqnarray*}
\int_Y \partial_v(\partial_u\phi)^2\,dv\,du&=&\int_{u_R(1)}^\infty \int_{\max(1,v_{R_2}(u))}^{v_R(u)}\partial_v(\partial_u\phi)^2\,dv\,du\\
&=& \int_{u_R(1)}^\infty(\partial_u\phi)^2(u,v_R(u))\,du-\int_{u_R(1)}^\infty (\partial_u\phi)^2(u,\max(1,v_{R_2}(u))\,du\\
&=&\int_{u_R(1)}^\infty(\partial_u\phi)^2(u,v_R(u))\,du- \int_{u_{R}(1)}^{u_{R_2}(1)}(\partial_u\phi)^2(u,1)\,du\\
& & -\int_{u_{R_2}(1)}^\infty (\partial_u\phi)^2(u,v_{R_2}(u))\,du\\
&=&\int_{\gamma_{R^*}\cap\{v\geq 1\}} (\partial_u\phi)^2- \int_{u_{R}(1)}^{u_{R_2}(1)}(\partial_u\phi)^2(u,1)\,du-\int_{\gamma_{R_2^*}\cap\{v\geq 1\}} (\partial_u\phi)^2
\end{eqnarray*}
Thus, integrating \eqref{preintegration} over $Y$ yields
\begin{eqnarray*}
\int_{\gamma_{R^*}\cap\{v\geq 1\}} (\partial_u\phi)^2+(\partial_v\phi)^2\,dv - \int_{\gamma_{R_2^*}\cap\{v\geq 1\}} (\partial_u\phi)^2+(\partial_v\phi)^2\,dv-\int_{u_{R}(1)}^{u_{R_2}(1)}(\partial_u\phi)^2(u,1)\,du\\
\leq \tilde{C} \int_Y (\partial_u\phi)^2+(\partial_v\phi)^2\,du\,dv\hspace{28mm}\\
=C \int_{R_2^*}^{R^*}\left( \int_{\gamma_r*\cap\{v\geq1\}}(\partial_u\phi)^2+(\partial_v\phi)^2\,dv\right)\,dr^*,
\end{eqnarray*}
and hence
\begin{eqnarray*}
\int_{\gamma_{R^*}\cap\{v\geq 1\}}(\partial_u\phi)^2+(\partial_v\phi)^2\,dv&\leq &\int_{\gamma_{R_2^*}\cap\{v\geq 1\}}(\partial_u\phi)^2+(\partial_v\phi)^2\,dv + \int_{u_{R}(1)}^{u_{R_2}(1)}(\partial_u\phi)^2(u,1)\,du \\
& & +C \int_{R_2^*}^{R^*}\left( \int_{\gamma_r*\cap\{v\geq1\}}(\partial_u\phi)^2+(\partial_v\phi)^2\,dv\right)\,dr^*\\
&\leq&\int_{\gamma_{R_2^*}\cap\{v\geq 1\}}(\partial_u\phi)^2+(\partial_v\phi)^2\,dv + \int_{u_{R}(1)}^{u_{R_2}(1)}(\frac{\partial_u\phi}{\Omega})^2(u,1)\,du \\
& & +C \int_{R_2^*}^{R^*}\left( \int_{\gamma_r*\cap\{v\geq1\}}(\partial_u\phi)^2+(\partial_v\phi)^2\,dv\right)\,dr^*,
\end{eqnarray*}
since $\Omega^2\leq1$ in the black hole exterior. Hence, by Gr\"{o}nwall's inequality we have
\begin{eqnarray*}
\int_{\gamma_{R^*}\cap\{v\geq 1\}}(\partial_u\phi)^2+(\partial_v\phi)^2\,dv&\leq& C(R,R_2)\left(\int_{\gamma_{R_2^*}\cap\{v\geq 1\}}(\partial_u\phi)^2+(\partial_v\phi)^2\,dv +  \int_{u_{R}(1)}^{u_{R_2}(1)}(\frac{\partial_u\phi}{\Omega})^2(u,1)\,du  \right).
\end{eqnarray*}

In particular, it follows that
\begin{eqnarray*}
\int_{\gamma_{R^*}\cap\{1\leq v\leq V\}}v^{\alpha-(1+2\varepsilon)}((\partial_u\phi)^2+(\partial_v\phi)^2)\hspace{70mm}\\
\leq V^{\alpha-(1+2\varepsilon)}\int_{\gamma_{R^*}\cap\{v\geq 1\}}(\partial_u\phi)^2+(\partial_v\phi)^2\hspace{57mm}\\
\hspace{20mm}\leq C(R,R_2,V,\alpha) \left(\int_{\gamma_{R_2^*}\cap\{v\geq 1\}}(\partial_u\phi)^2+(\partial_v\phi)^2\,dv +  \int_{u_{R}(1)}^{u_{R_2}(1)}(\frac{\partial_u\phi}{\Omega})^2(u,1)\,du  \right)\\
\leq C\left(\int_{\gamma_{R_2^*}\cap\{v\geq 1\}}v^{\alpha-(1+\varepsilon)}((\partial_u\phi)^2+(\partial_v\phi)^2)\,dv +  \int_{u_{R}(1)}^\infty(\frac{\partial_u\phi}{\Omega})^2(u,1)\,du  \right),\hspace{4mm}
\end{eqnarray*}
and hence by \eqref{comb}, we have
\begin{eqnarray}
\int_{\gamma_{R^*}\cap\{1\leq v\leq V\}}v^{\alpha-(1+2\varepsilon)}((\partial_u\phi)^2+(\partial_v\phi)^2)\,dv\hspace{70mm}\notag\\
\leq C\left(\int_{1}^{\infty}v^\alpha(\partial_v\phi)^2(\infty,v)\,dv+\int_{u_{R}(1)}^{\infty}\left(\frac{\partial_u\phi}{\Omega}\right)^2(u,1)\,du\right),\label{first term}\hspace{15mm}
\end{eqnarray}
for some $C=C(R_2,R,\alpha,\varepsilon)>0$. So summing \eqref{part2} and \eqref{first term} gives the desired result.
\end{proof}

\newpage
\subsubsection{$L^p$ Estimates}\label{remlab}
 We are now ready to prove Theorem \ref{lp exterior estimate}, which we restate below for convenience.

\begin{theorem*}
Suppose $\phi$ is a solution of the wave equation \eqref{waveequation} with smooth, spherically symmetric data on $S$. Let $\alpha,\varepsilon>0$ be such that $\alpha-(1+2\varepsilon)>0$. Assume $1<p<2$ and $0<\theta<1$ with
\[\frac{1}{p}=\frac{1-\theta}{1}+\frac{\theta}{2}.\]
If $R>R_1$, then there exists $C=C(R,\alpha,\varepsilon,p)>0$ such that
\begin{equation}\label{vcase}\int_{\gamma_R\cap\{v\geq1\}} v^{(\alpha-(1+2\varepsilon))\frac{p\theta}{2}}|\partial_v\phi|^p \leq C \left(\int_{\mathcal{H}^+\cap\{v\geq1\}} v^{\frac{\alpha p\theta}{2}}|\partial_v\phi|^p+\int_{\underline{C}_1^{ext}\cap\{u\geq u_R(1)\}} \Omega^{-p\theta}|\partial_u\phi|^p\right) 
\end{equation}
and
\begin{equation}\label{ucase}\int_{\gamma_R\cap\{v\geq1\}} v^{(\alpha-(1+2\varepsilon))\frac{p\theta}{2}}|\partial_u\phi|^p \leq C \left(\int_{\mathcal{H}^+\cap\{v\geq1\}} v^{\frac{\alpha p\theta}{2}}|\partial_v\phi|^p+\int_{\underline{C}_1^{ext}\cap\{u\geq u_R(1)\}} \Omega^{-p\theta}|\partial_u\phi|^p\right). 
\end{equation}
\end{theorem*}

\begin{remark}\label{rem}
The above theorem could be proved in exactly the same way as Theorem \ref{interior reduction} provided we made the additional assumption that the solution $\phi$ satisfies \eqref{ic2}. Recall that including this assumption in Theorem \ref{interior reduction} allowed  us to split the solution into two parts, $\phi_1$ and $\phi_2$, and deduce the desired estimate for each of them separately (and in particular without using interpolation in the case of $\phi_2$ because by \eqref{ic2} it is well behaved). For Theorem \ref{interior reduction}, this approach had the advantage of allowing us to deduce the condition on solutions of the form $\phi_1$ (namely smooth, spherically symmetric solutions with $\partial_u\phi_1=0$ on $\underline{C}_1\cap\{u\leq-1\}$)  asserted in Theorem \ref{condinstab2} which may prevent instability. However, the assumption \eqref{ic2} is not actually needed to perform interpolation and deduce the estimate. Indeed, removing the assumption and not splitting the solution, we may deduce the same estimate but with the constant independent of $D$ and $\phi$. As we do not assume \eqref{ic2} in Theorem \ref{lp exterior estimate}, we follow this approach to prove it.  We note that an almost identical argument to the one given below could have been used for Theorem \ref{interior reduction}, though with the disadvantage that the condition of Theorem \ref{condinstab2} would not follow from the proof in this case.
\end{remark}

\begin{proof}
We prove the only the first estimate \eqref{vcase} as the proof of the second estimate \eqref{ucase} is analogous.

We have $R>R_1>R_2$ and so Proposition \ref{l2ext22} applies. The idea is to use the K-method of real interpolation to interpolate between the bounds achieved in Corollary \ref{extl11} and Proposition \ref{l2ext22} to deduce the desired estimate.  Again, for clarity, set $p_0=1$ and $p_1=2$, so that
\[\frac{1}{p}=\frac{1-\theta}{p_0}+\frac{\theta}{p_1}.\] 
We split the argument into a number of steps.

\noindent\textbf{Step 1:} First of all, we need to define the compatible couples we wish to interpolate between. Given  postitve, measurable functions $w:[1,\infty)\to(0,\infty)$, $\tilde{w}:[u_R(1),\infty)\to(0,\infty)$ and $q\geq1$, we set 
\[A_q(w,\tilde{w}):=L^q_w([1,\infty))\times L^q_{\tilde{w}}([u_R(1),\infty))\]
with norm
\[\|(f,g)\|_{A_q(w,\tilde{w})}=\left(\int_1^\infty w(v)|f|^{q}(v)\,dv+\int_{u_R(1)}^\infty \tilde{w}(u)|g|^{q}(u)\,du\right)^{1/q}.\]
Then, given  positive, measurable functions $w_0,w_1:[1,\infty)\to(0,\infty)$ and $\tilde{w}_0,\tilde{w}_1:[u_R(1),\infty)$,  define
\begin{eqnarray*}
A_0=A_{p_0}(w_0,\tilde{w}_0)\hspace{7mm}  \mbox{ and }\hspace{7mm}  A_1=A_{p_1}(w_1,\tilde{w}_1)
\end{eqnarray*}
so
\begin{eqnarray*}
\|(f,g)\|_{A_j}&=&\left(\int_1^\infty w_j(v)|f|^{p_j}(v)\,dv+\int_{u_R(1)}^\infty \tilde{w}_j(u)|g|^{p_j}(u)\,du\right)^{1/p_j}, \hspace{3mm} j=0,1.
\end{eqnarray*}

Similarly for a postitve, measurable function $\omega:[1,\infty)\to(0,\infty)$ and $q\geq1$, we set 
\[B_q(\omega):=L^q_\omega([1,\infty)),\]
so
\[\|f\|_{B_q(\omega)}=\left(\int_1^\infty \omega(v)|f|^{q}(v)\,dv\right)^{1/q}.\]
Then, given  positive, measurable functions $\omega_0,\omega_1:[1,\infty)\to(0,\infty)$,  define
\begin{eqnarray*}
B_0=B_{p_0}(\omega_0) \hspace{7mm}\mbox{ and } \hspace{7mm} B_1=B_{p_1}(\omega_1)
\end{eqnarray*}
so
\begin{eqnarray*}
\|f\|_{B_j}&=&\left(\int_1^\infty \omega_j(v)|f|^{p_j}(v)\,dv\right)^{1/p_j}, \hspace{3mm} j=0,1.
\end{eqnarray*}
Finally, let $\overline{A}$ and $\overline{B}$ denote the compatible couples $\overline{A}=(A_0,A_1)$ and $\overline{B}=(B_0,B_1)$.

\noindent \textbf{Step 2:} 
Suppose $S:\overline{A}\to\overline{B}$ is a linear operator. Recall from Section \ref{INS} this means (with slight abuse of notation) that, $S:A_0+A_1\to B_0+B_1$ and moreover, $S:A_j\to B_j$ and is bounded, with norm $\|S\|_j$ say, for $j=0,1$.
Recalling that
\[\frac{1}{p}=\frac{1-\theta}{p_0}+\frac{\theta}{p_1}=\frac{1-\theta}{1}+\frac{\theta}{2},\]
it follows from Theorem \ref{kmethod} (with $q=p$ and $\theta=\theta$) that
\[S:\overline{A}_{\theta,p}\to\overline{B}_{\theta,p} \hspace{3mm}\mbox{ is bounded with norm }\hspace{3mm}\|S\|_{\theta,p}\leq\|S\|_0^{1-\theta}\|S\|_1^{\theta}.\]
In order to make use of this result, we will need to identify appropriate weights $w_0,w_1,\tilde{w}_0,\tilde{w}_1,\omega_0$ and $\omega_1$ and a bounded operator $T: A_j\to B_j$. This is where the $L^1$ and $L^2$ estimates proved in the previous sections will come in. We will also need to understand the spaces $\overline{A}_{\theta,p}$ and $\overline{B}_{\theta,p}$ and this is our next task.

\noindent\textbf{Step 3:} 
We show that
\[\overline{A}_{\theta,p}  :=(A_0,A_1)_{\theta,p}=(A_{p_0}(w_0,\tilde{w}_0), A_{p_1}(w_1,\tilde{w}_1))_{\theta,p}= A_p(w,\tilde{w}),\]
with equivalent norms, where
\[w=w_0^{p\frac{1-\theta}{p_0}}w_1^{p\frac{\theta}{p_1}} \hspace{6mm}\mbox{ and }\hspace{6mm} \tilde{w}=\tilde{w}_0^{p\frac{1-\theta}{p_0}}\tilde{w}_1^{p\frac{\theta}{p_1}}.\]

The proof of this fact is almost identical to the proof of the corresponding statement in Theorem \ref{interior reduction} (again following the proof of Theorem 5.5.1 of \cite{BL}). As before, set $\rho_0=p_0$, $\rho_1=p_1$, $\eta=\frac{\theta p}{p_1}\in(0,1)$ and $r=1$. Now $1-\eta=1-\frac{\theta p}{p_1}=(\frac{1}{p}-\frac{\theta}{p_1})p =\frac{1-\theta}{p_0}p$, and it follows that $\rho:=(1-\eta)\rho_0+\eta\rho_1=p$ and $q:=\rho r=p$. Then by the Power Theorem (Theorem \ref{power}) we have
\begin{equation}\label{power2ext}
\left((A_0,A_1)_{\theta,p}\right)^{p}=(A_0^{p_0},A_1^{p_1})_{\eta,1}
\end{equation}
with equivalent quasinorms.
For brevity we set $X=(A_0^{p_0},A_1^{p_1})_{\eta,1}$. For $(f,g)\in A_0^{p_0}+A_1^{p_1}$, we have

\begin{eqnarray}
\|(f,g)\|_X &=&\int_0^\infty\left(t^{-\eta}K(t,(f,g);A_0^{p_0},A_1^{p_1}   )\right)^1\,\frac{dt}{t}\notag\\
&=&\int_0^\infty t^{-\eta} (\inf_{\substack{(f,g)=(f_0,g_0)+(f_1,g_1)\\f_i\in L^{p_i}_{w_i}([1,\infty))\\g_i\in L^{p_i}_{\tilde{w}_i}([u_R(1),\infty))}}\left[\|(f_0,g_0)\|_{A_0^{p_0}}+t\|(f_1,g_1)\|_{A_1^{p_1}}\right])\frac{dt}{t} \notag\\
&=&\int_0^\infty t^{-\eta} \inf_{\substack{(f,g)=(f_0,g_0)+(f_1,g_1)\\f_i\in L^{p_i}_{w_i}([1,\infty))\\g_i\in L^{p_i}_{\tilde{w}_i}([u_R(1),\infty))}}\left[\int_1^\infty w_0|f_0|^{p_0}(v)\,dv +\int_{u_R(1)}^{\infty}\tilde{w}_0|g_0|^{p_0}(u)\,du \right.\notag\\
& &\left.+t\left(  \int_1^\infty w_1|f_1|^{p_1}(v)\,dv +\int_{u_R(1)}^{\infty}\tilde{w}_1|g_1|^{p_1}(u)\,du\right)\right]\,\frac{dt}{t}\notag\\
&=&\int_0^\infty t^{-\eta}\left[\inf_{\substack{f=f_0+f_1\\f_i\in L^{p_i}_{w_i}}}
\left(\int_1^\infty w_0|f_0|^{p_0}(v) +t w_1|f_1|^{p_1}(v)\,dv\right) \right.\notag\\
& &\left.+\inf_{\substack{g=g_0+g_1\\g_i\in L^{p_i}_{\tilde{w}_i}}}\left( \int_{u_R(1)}^{\infty}\tilde{w}_0|g_0|^{p_0}(u)+t\tilde{w}_1|g_1|^{p_1}(u)\,du\right)\right] \,\frac{dt}{t}\notag\\
&=&\int_1^\infty\left(\int_0^\infty t^{-\eta}\inf_{\substack{f=f_0+f_1\\f_i\in L^{p_i}_{w_i}}}(w_0|f_0|^{p_0}(v)+tw_1|f_1|^{p_1}(v))\,\frac{dt}{t}\right)\,dv\notag\\
& &+\int_{u_R(1)}^{\infty}\left(\int_{0}^\infty t^{-\eta} \inf_{\substack{g=g_0+g_1\\g_i\in L^{p_i}_{\tilde{w}_i}}} \left(\tilde{w}_0|g_0|^{p_0}(u)+t\tilde{w}_1|g_1|^{p_1}(u)\right)\,\frac{dt}{t}\right)\,du.\label{switchext}
\end{eqnarray}
Note that \eqref{switchext} follows from an analogous argument to  Lemma \ref{interplem} of Appendix \ref{Switch}.

Now arguing exactly as in the proof of Theorem \ref{interior reduction} we see that
\[\int_1^\infty\left(\int_0^\infty t^{-\eta}\inf_{\substack{f=f_0+f_1\\f_i\in L^{p_i}_{w_i}}}(w_0|f_0|^{p_0}(v)+tw_1|f_1|^{p_1}(v))\,\frac{dt}{t}\right)\,dv=C\int_1^\infty w|f|^{p}(v)\,dv\]
and
\begin{eqnarray*}
\int_{u_R(1)}^{\infty}\left(\int_{0}^\infty t^{-\eta} \inf_{\substack{g=g_0+g_1\\g_i\in L^{p_i}_{\tilde{w}_i}}} \left(\tilde{w}_0|g_0|^{p_0}(u)+t\tilde{w}_1|g_1|^{p_1}(u)\right)\,\frac{dt}{t}\right)\,du=C\int_{u_R(1)}^{\infty}\tilde{w}|g|^{p}(u)\,du,
\end{eqnarray*}
with $C=\int_0^\infty s^{-\eta}F(s)\,\frac{ds}{s}<+\infty$ (where $F$ is defined by \eqref{F}).
Thus by \eqref{switchext}
\[\|(f,g)\|_X=C\left(\int_1^\infty w|f|^{p}(v)\,dv+\int_{u_R(1)}^{\infty}\tilde{w}|\partial_ug|^{p}(u)\,du.\right)=C\|(f,g)\|_{A_p(w,\tilde{w})}^{p}.\]
In particular, $\|(f,g)\|_X<+\infty \Longleftrightarrow \|(f,g)\|_{A_p(w,\tilde{w})^{p}}<+\infty$, so $X=A_p(w,\tilde{w})^{p}$ and they have equivalent quasinorms. Thus, by \eqref{power2ext} 
\[\left((A_0,A_1)_{\theta,p}\right)^{p}=A_p(w,\tilde{w})^{p}\]
with equivalent quasinorms, and so
\[(A_0,A_1)_{\theta,p}=A_p(w,\tilde{w})\]
with equivalent norms, as desired.

\noindent\textbf{Step 4:} 
Using an argument analogous to that in the previous step, we deduce
\[\overline{B}_{\theta,p}:=(B_0,B_1)_{\theta,p}=(B_{p_0}(\omega_0),B_{p_1}(\omega_1))=B_p(\omega)\]
with equivalent norms, where 
\[\omega=\omega_0^{p\frac{1-\theta}{p_0}}\omega_1^{p\frac{\theta}{p_1}}.\]
To avoid repetition, we omit the details.

\noindent \textbf{Step 5:} We now fix the weights by setting $w_0(v)=1$, $w_1(v)=v^\alpha$, $\tilde{w}_0(u)=1$, $\tilde{w}_1(u)=\Omega^{-2}(u,1)$, $\omega_0(v)=1$ and $\omega_1(v)=v^{\alpha-(1+2\varepsilon)}$, where $0<\varepsilon<1$ and $\alpha-(1+2\varepsilon)>0$. So 
\[A_0=L^1([1,\infty))\times L^1([u_R(1),\infty)), \hspace{7mm} A_1=L^2_{v^\alpha}([1,\infty))\times L^2_{\Omega^{-2}(u,1)}([u_R(1),\infty))\hspace{3mm}\mbox{ (as sets)}\]
and
\[B_0=L^1([1,\infty)), \hspace{7mm} B_1=L^2_{v^{\alpha-(1+2\varepsilon)}}([1,\infty)).\]

 In order to construct a linear operator $T:\overline{A}\to\overline{B}$, we first construct two bounded linear operators
\[T_0:A_0\to B_0 \hspace{5mm}\mbox{ and } \hspace{5mm} T_1:A_1\to B_1.\]
Before defining these operators however, we note that given smooth and  spherically symmetric data\\ $(\partial_v\phi\vert_{\mathcal{H}^+\cap\{v\geq1\}},\partial_u\phi\vert_{\underline{C}_1^{ext}\cap\{u\geq u_R(1)\}})$ for the wave equation \eqref{waveequation}, the solution $\phi$ is determined up to a constant in $\{r\leq R\}\cap\{v\geq1\}$ (as it is determined up to a constant on $\mathcal{H}^+\cap\{v\geq1\}$ and $\underline{C}_1^{ext}\cap\{u\geq u_R(1)\}$), and hence $\partial_v\phi$ is uniquely determined on $\gamma_R\cap\{v\geq1\}$.

\noindent \textbf{Step 5a:} We define the operator $T_0:A_0\to B_0$. We first define $T_0$ for smooth  $(f,g)\in A_0$ and then use a limiting process to extend $T_0$ to the entire space $A_0$. 

Indeed, given a smooth $(f,g)\in A_0$ we set $T_0(f,g)=\partial_v\varphi\vert_{\gamma_R\cap\{v\geq1\}}$ where $\varphi$ is a solution of the (spherically symmetric) wave equation with $\partial_v\varphi\vert_{\mathcal{H}^+\cap\{v\geq1\}}=f$ and $\partial_u\varphi\vert_{\underline{C}_1\cap\{u\geq u_R(1)\}}=g$. Then, by the remarks directly above $T_0(f,g)$ is uniquely determined and by the $L^1$ estimates (Proposition \ref{extl11}) $T_0(f,g)\in B_0$.
Indeed, 
\begin{eqnarray*}
\|T_0(f,g)\|_{B_0}=\int_1^\infty |T_0(f,g)(v)|\,dv &=&\int_1^\infty |\partial_v\varphi|(u_R(v),v)\,dv\\
&\leq& C(R)\left(\int_{\mathcal{H}^+\cap\{v\geq1\}}|\partial_v\varphi|+\int_{\underline{C}_1^{ext}\cap\{u\geq u_R(1)\}} |\partial_u\varphi|\right)\\
&=& C(R)\left(\int_1^\infty |f(v)|\,dv+\int_{u_R(1)}^\infty |g(u)|\,du\right)\\
&=&C(R) \|(f,g)\|_{A_0}<\infty.
\end{eqnarray*}
Moreover, by uniqueness and linearity of solutions to the wave equation, $T_0$ acts linearly on the smooth elements of $A_0$. See Figure \ref{action2} below for a diagramatic representation of the action of $T_0$ on smooth elements of $A_0$.

Now suppose $(f,g)\in A_0$ is not smooth. By density of smooth functions $C^\infty([1,\infty))\times C^\infty([u_R(1),\infty))$ in $A_0$, there is a sequence of smooth functions $\{(f_n,g_n)\}$ in $A_0$ such that $(f_n,g_n)\to (f,g)$ in $A_0$. Thus $\{(f_n,g_n)\}$ is a Cauchy sequence in $A_0$, $T_0(f_n,g_n)\in B_0$ is defined for each $n$ by above and by the $L^1$ estimates (Proposition \ref{extl11}) $\{T_0(f_n,g_n)\}$ is a Cauchy sequence in $B_0$. But $B_0=L^1$ is complete and hence $\{T_0(f_n,g_n)\}$ has a limit in $B_0$. We define 
\[T_0(f,g):=\lim_{n\to\infty}T_0(f_n,g_n) \hspace{2mm}\mbox{ in } B_0.\]

It is easy to check that $T_0(f,g)$ is well-defined. Indeed, if $\{(f_n,g_n)\}$ and $\{(f_n',g_n')\}$ are both sequences of smooth functions in $A_0$ which converge to $(f,g)$ in $A_0$, then by Proposition \ref{extl11}
\[
\|T_0(f_n,g_n)-T_0(f_n',g_n')\|_{B_0}=\|T_0(f_n-f_n',g_n-g_n')\|_{B_0}\leq C(R)\|f_n-f_n',g_n-g_n')\|_{A_0}\to0,
\]
so it follows that $T_0(f,g)$ is uniquely determined. So $T_0$ is well-defined. It is trivial to check that $T_0$ is linear. Furthermore, $T_0$ is bounded: if $(f,g)\in A_0$ is smooth, it follows immediately from Proposition \ref{extl11} that $\|T_0(f,g)\|_{B_0}\leq C(R)\|(f,g)\|_{A_0}$, and if $(f,g)\in A_0$ is not smooth, then for any sequence $\{(f_n,g_n)\}$ of smooth 
functions in $A_0$ converging to $(f,g)$ in $A_0$,
\[
\|T_0(f,g)\|_{B_0}=\lim_{n\to\infty}\|T_0(f_n,g_n)\|_{B_0}\leq \lim_{n\to\infty} C(R)\|(f_n,g_n)\|_{A_0} =C(R)\|(f,g)\|_{A_0}.\]
So we have constructed the desired bounded linear operator $T_0:A_0\to B_0$ and moreover $\|T_0\|\leq C(R)$, where $C(R)$ is the constant from Proposition \ref{extl11}.

\begin{figure}[h]
\centering
\includegraphics[scale=0.5]{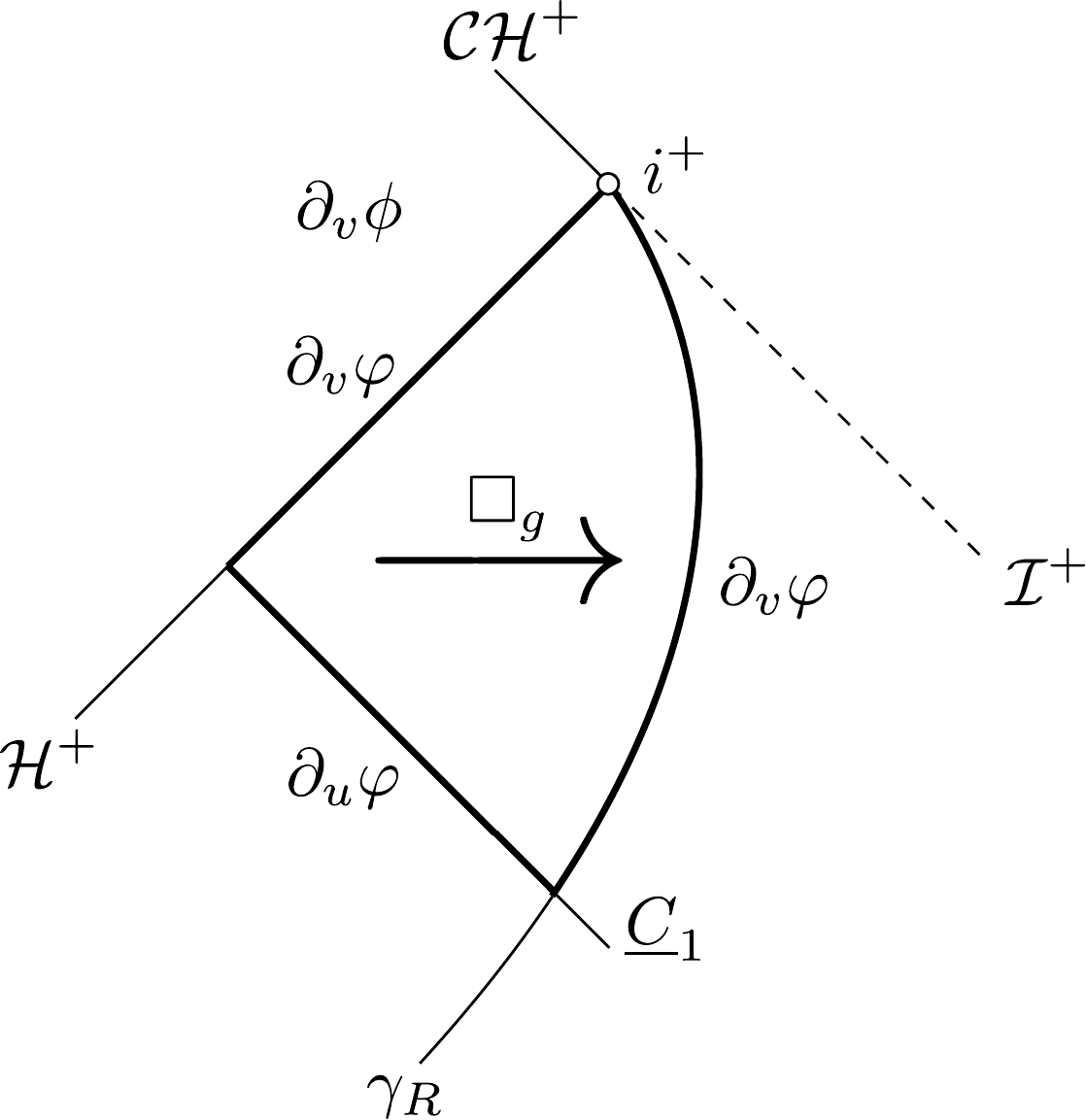}
\caption{Action of the operators $T_0$ and $T_1$ on smooth functions in their domains.}
\label{action2}
\end{figure}

\noindent \textbf{Step 5b:} We similarly define the operator $T_1:A_1\to B_1$. Again, we first define $T_1$ for smooth  $(f,g)\in A_1$ and then use a limiting process to extend $T_1$ to the entire space $A_1$. 

Given a smooth $(f,g)\in A_1$, we set $T_1(f,g)=\partial_v\varphi\vert_{\gamma_R\cap\{v\geq1\}}$, where $\varphi$ is a solution of the (spherically symmetric) wave equation with $\partial_v\varphi\vert_{\mathcal{H}^+\cap\{v\geq1\}}=f$ and $\partial_u\varphi\vert_{\underline{C}_1\cap\{u\geq u_R(1)\}}=g$. By the remarks at the start of Step 5 $T_1(f,g)$ is uniquely determined and by the $L^2$ estimates (Proposition \ref{l2ext22}) $T_1(f,g)\in B_1$.
Indeed,
\begin{eqnarray*}
\|T_1(f,g)\|_{B_1}^2&=&\int_1^\infty v^{\alpha-(1+2\varepsilon)}  |T_1(f,g)|^2(v)\,dv\\
 &=&\int_1^\infty v^{\alpha-(1+2\varepsilon)}|\partial_v\varphi|^2(u_R(v),v)\,dv\\
&\leq& C(R,\alpha,\varepsilon) \left(\int_{\mathcal{H}^+\cap\{v\geq1\}}v^\alpha(\partial_v\varphi)^2+\int_{\underline{C}_1^{ext}\cap\{u\geq u_R(1)\}}\left(\frac{\partial_u\varphi}{\Omega}\right)^2\right)\\
&=&C(R,\alpha,\varepsilon) \left(\int_1 ^\infty v^\alpha|f(v)|^2\,dv+\int_{u_R(1)}^\infty \frac{|g(u)|^2}{\Omega^2(u,1)}\,du\right)\\
&=& C(R,\alpha,\varepsilon)\|(f,g)\|_{A_1}^2<\infty.
\end{eqnarray*}
Moreover, by uniqueness and linearity of solutions to the wave equation, $T_1$ acts linearly on the smooth elements of $A_1$.

Now suppose $(f,g)\in A_1$ is not smooth. By density of smooth functions $C^\infty([1,\infty))\times C^\infty([u_R(1),\infty))$ in $A_1$, there is a sequence of smooth functions $\{(f_n,g_n)\}$ in $A_1$ such that $(f_n,g_n)\to (f,g)$ in $A_1$. Thus $\{(f_n,g_n)\}$ is a Cauchy sequence in $A_1$, $T_1(f_n,g_n)\in B_1$ is defined for each $n$ by above and by the $L^2$ estimates (Proposition \ref{l2ext22}) $\{T_1(f_n,g_n)\}$ is a Cauchy sequence in $B_1$. But $B_1=L^2_{v^{\alpha-(1+2\varepsilon)}}([1,\infty))$ is complete and hence $\{T_1(f_n,g_n)\}$ has a limit in $B_1$. We define 
\[T_1(f,g):=\lim_{n\to\infty}T_1(f_n,g_n) \hspace{2mm}\mbox{ in } B_1.\]

As before $T_1f$ is well-defined. Indeed, if $\{(f_n,g_n)\}$ and $\{(f_n',g_n')\}$ are both sequences of smooth functions in $A_1$ which converge to $(f,g)$ in $A_1$, then by Proposition \ref{l2ext22}
\[
\|T_1(f_n,g_n)-T_1(f_n',g_n')\|_{B_1}=\|T_1(f_n-f_n',g_n-g_n')\|_{B_1}\leq C(R,\alpha,\varepsilon)^{1/2}\|(f_n-f_n',g_n-g_n')\|_{A_1}\to0,
\]
so it follows that $T_1(f,g)$ is uniquely determined. So $T_1$ is well defined. It is trivial to check that $T_1$ is linear. Furthermore, $T_1$ is bounded: if $(f,g)\in A_1$ is smooth, it follows immediately from Proposition \ref{l2ext22} that $\|T_1(f,g)\|_{B_1}\leq C(R,\alpha,\varepsilon)^{1/2}\|(f,g)\|_{A_1}$, and if $(f,g)\in A_1$ is not smooth, then for any sequence of smooth functions $\{(f_n,g_n)\}$ in $A_1$ converging to $(f,g)$ in $A_1$,
\[
\|T_1(f,g)\|_{B_1}=\lim_{n\to\infty}\|T_1(f_n,g_n)\|_{B_1}\leq \lim_{n\to\infty} C(R,\alpha,\varepsilon)^{1/2}\|(f_n,g_n)\|_{A_1} =C(R,\alpha,\varepsilon)^{1/2}\|(f,g)\|_{A_1}.\]
So we have constructed the desired bounded linear operator $T_1:A_1\to B_1$ and moreover $\|T_1\|\leq C(R,\alpha,\varepsilon)^{1/2}$, where $C(R,\alpha,\varepsilon)$ is the constant from Proposition \ref{l2ext22}.

\noindent \textbf{Step 5c:} We show that $T_0=T_1$ on $A_0\cap A_1$. To see this, we note that $T_0\vert_{A_0\cap A_1}:A_0\cap A_1 \to B_0+B_1$ is continuous. Indeed if $(f,g)\in A_0\cap A_1$, by continuity of $T_0:A_0\to B_0$ and by definition of $\|\cdot\|_{B_0+B_1}$ and $\|\cdot\|_{A_0\cap A_1}$ (see Appendix \ref{Interpolation Theory}), we have
\[\|T_0(f,g)\|_{B_0+B_1}\leq \|T_0(f,g)\|_{B_0}\leq C(R)\|(f,g)\|_{A_0}\leq C(R)\|(f,g)\|_{A_0\cap A_1},\]
so $T_0\vert_{A_0\cap A_1}:A_0\cap A_1 \to B_0+B_1$ is indeed continuous.
In exactly the same way,  $T_1\vert_{A_0\cap A_1}:A_0\cap A_1 \to B_0+B_1$ is continuous.

Now, the space of smooth and compactly supported functions $\mathcal{C}:=C^\infty_c([1,\infty))\times C^\infty_c([u_R(1),\infty))$ satisfies $\mathcal{C}\subset A_1$ and $(\mathcal{C},\|\cdot\|_{A_1})$ is dense in $(A_1,\|\cdot\|_{A_1})$. Also  $\alpha>1+2\varepsilon$, so given $(f,g)\in A_1$, by H\"{o}lder's inequality we have
\begin{eqnarray}
\|(f,g)\|_{A_0} &=&\int_1^\infty |f|(v)\,dv+\int_{u_R(1)}^\infty |g|(u)\,du\notag\hspace{90mm}\\
&\leq& \left(\int_1^\infty v^\alpha |f|^2(v)\,dv\cdot \int_1^\infty v^{-\alpha}\,dv\right)^{1/2}+\left( \int_{u_R(1)}^\infty \frac{|g|^2(u)}{\Omega^2(u,1)}\,du\cdot\int_{u_R(1)}^\infty \Omega^2(u,1)\,du\right)^{1/2}\notag\\
&=&C_1(\alpha)\left(\int_1^\infty v^\alpha |f|^2(v)\,dv\right)^{1/2}+(R-r_+)^{1/2}\left( \int_{u_R(1)}^\infty \frac{|g|^2(u)}{\Omega^2(u,1)}\,du\right)^{1/2}\notag\\
&=&C(R,\alpha)\|(f,g)\|_{A_1}<\infty,\notag
\end{eqnarray}
and hence $(f,g)\in A_0$. So $A_1\subset A_0$. In particular, $A_0\cap A_1=A_1$ and (as $\|\cdot\|_{A_1}\leq \|\cdot\|_{A_0\cap A_1}$ and by the previous equation) $\|\cdot\|_{A_0\cap A_1}=\max(\|\cdot\|_{A_0},\|\cdot\|_{A_1})$ is equivalent to $\|\cdot\|_{A_1}$. It follows that $(\mathcal{C},\|\cdot\|_{A_0\cap A_1})$ is dense in $(A_0\cap A_1,\|\cdot\|_{A_0\cap A_1})$. Furthermore, it follows immediately from the definitions of $T_0$ and $T_1$ that $T_0=T_1$ on $\mathcal{C}$ (as their actions on smooth functions are defined in the same way).

Now recall the classical result that if $X$, $Y$ are metric spaces, $S,T:X\to Y$ continuous maps and $A\subset X$ a dense subset of $X$ such that $S=T$ on $A$, then $S=T$ (on the entire space $X$). Applying this result with $X=A_0\cap A_1$, $Y=B_0+B_1$, $A=\mathcal{C}$, $S=T_0\vert_{A_0\cap A_1}$ and $T=T_1\vert_{A_0\cap A_1}$, we deduce that $T_0=T_1$ on $A_0\cap A_1$.

\noindent \textbf{Step 5d:} We define the linear operator $T:A_0+A_1\to B_0+B_1$ as follows. For  $(f,g)\in A_0+A_1$ with $(f,g)=(f_0,g_0)+(f_1,g_1)$, $(f_i,g_i)\in A_i$, we set $T(f,g)=T_0(f_0,g_0)+T_1(f_1,g_1)$. We need to check that $T$ is well-defined and linear.
\begin{itemize}
\item
To see that $T$ is well-defined, suppose $(f,g)\in A_0+A_1$ with $(f,g)=(f_0,g_0)+(f_1,g_1)=(f_0',g_0')+(f_1',g_1')$, where $(f_i,g_i), (f_i',g_i')\in A_i$. Then $(f_0,g_0)-(f_0',g_0')=(f_1',g_1')-(f_1,g_1)\in A_0\cap A_1$ and hence (by Step 5c)
\begin{eqnarray*}
T_0((f_0,g_0)-(f_0',g_0'))=T_0((f_1',g_1')-(f_1,g_1))&=&T_1((f_1',g_1')-(f_1,g_1))\\
\Longrightarrow \hspace{40mm}T_0(f_0,g_0)-T_0(f_0',g_0')&=&T_1(f_1',g_1')-T_1(f_1,g_1)\\
\Longrightarrow \hspace{40mm}T_0(f_0,g_0)+T_1(f_1,g_1)&=&T_0(f_0',g_0')+T_1(f_1',g_1'),
\end{eqnarray*}
so $T(f,g)$ is independent of the representation of $(f,g)$ and so $T$ is well-defined.

\item To see that $T$ is linear, suppose $\alpha,\beta\in\mathbb{R}$ and $(f,g)=(f_0,g_0)+(f_1,g_1), (f',g')=(f_0',g_0')+(f_1',g_1')\in A_0+A_1$ (where $(f_i,g_i), (f_i',g_i')\in A_i$). Then by linearity of $T_0:A_0\to B_0$ and $T_1:A_1\to B_1$
\begin{eqnarray*}
T(\alpha (f,g)+ \beta (f',g'))&=&T[(\alpha (f_0,g_0)+\beta(f_0', g_0'))+(\alpha (f_1,g_1)+\beta (f_1',g_1'))]\\ &=& T_0( \alpha (f_0,g_0)+\beta(f_0', g_0')  )+T_1(\alpha (f_1,g_1)+\beta (f_1',g_1'))\\
&=&\alpha (T_0(f_0,g_0)+T_1(f_1,g_1))+\beta(T_0(f_0',g_0')+T_1(f_1',g_1'))\\
&=&\alpha T(f,g)+\beta T(f',g'), 
\end{eqnarray*}
so $T$ is indeed linear.
\end{itemize}

\noindent \textbf{Step 6:} Notice it follows immediately from the definition of the linear operator $T:A_0+A_1\to B_0+B_1$ that $T\vert_{A_0}=T_0:A_0\to B_0$ and $T\vert_{A_1}=T_1:A_1\to B_1$ and both of these maps are bounded. In other words, using the notation of Appendix \ref{Interpolation Theory}, $T:\overline{A}\to \overline{B}$. Hence, by Step 2, 
\[T:(A_0,A_1)_{\theta,p}\to(B_0,B_1)_{\theta,p}\hspace{3mm} \mbox{ with norm } \hspace{3mm}\|T\|_{\theta,p}\leq \|T_0\|^{1-\theta}\|T_1\|^\theta.\]
Recalling that we fixed $p_0=1$, $p_1=1$
$w_0(v)=1$, $w_1(v)=v^\alpha$, $\tilde{w}_0(u)=1$, $\tilde{w}_1(u)=\Omega^{-2}(u,1)$, $\omega_0(v)=1$ and $\omega_1(v)=v^{\alpha-(1+2\varepsilon)}$, where  $0<\varepsilon<1$ and $\alpha-(1+2\varepsilon)>0$, Steps 3 and 4 allow us to compute $(A_0,A_1)_{\theta,p}$ and $(B_0,B_1)_{\theta,p}$. By Step 3
\[(A_0,A_1)_{\theta,p}=A_p(w,\tilde{w})=L^p_w([1,\infty))\times L^p_{\tilde{w}}([u_R(1),\infty))\]
with norm
\[\|(f,g)\|_{(A_0,A_1)_{\theta,p}}=\left(\int_1^\infty w(v)|f|^{p}(v)\,dv+\int_{u_R(1)}^\infty \tilde{w}(u)|g|^{p}(u)\,du\right)^{1/p},\]
where  
\[w(v)=w_0(v)^{p\frac{1-\theta}{p_0}}w_1(v)^{p\frac{\theta}{p_1}}=1^{p(1-\theta)/1}\cdot (v^\alpha)^{p\theta/2}=v^\frac{\alpha p\theta}{2}\]
and
\[\tilde{w}(u)=\tilde{w}_0(u)^{p\frac{1-\theta}{p_0}}\tilde{w}_1(u)^{p\frac{\theta}{p_1}}=1^{p(1-\theta)/1}\cdot(\Omega^{-2}(u,1))^{p\theta}=\Omega^{-p\theta}(u,1).\]
So 
\[(A_0,A_1)_{\theta,p}=L^p_{v^\frac{\alpha p\theta}{2}}([1,\infty))\times L^p_{\Omega^{-p\theta}(u,1)}([u_R(1),\infty))\]
with norm
\[\|(f,g)\|_{(A_0,A_1)_{\theta,p}}=\left(\int_1^\infty v^{\frac{\alpha p\theta}{2}}|f|^{p}(v)\,dv+\int_{u_R(1)}^\infty \Omega^{-p\theta}(u,1) |g|^{p}(u)\,du\right)^{1/p}.\]
Similarly, by Step 4
\[(B_0,B_1)_{\theta,p}=B_p(\omega)=L^p_\omega([1,\infty)),\]
where  
\[\omega(v)=\omega_0(v)^{p\frac{1-\theta}{p_0}}\omega_1(v)^{p\frac{\theta}{p_1}}=1^{p(1-\theta)/1}\cdot (v^{\alpha-(1+2\varepsilon)})^{p\theta/2}=v^{(\alpha-(1+2\varepsilon))\frac{ p\theta}{2}},\]
so 
\[(B_0,B_1)_{\theta,p}=L^p_{v^{(\alpha -(1+2\varepsilon))\frac{p\theta}{2}}}([1,\infty))\]
with norm
\[\|f\|_{(B_0,B_1)_{\theta,p}}=\left(\int_1^\infty v^{(\alpha -(1+2\varepsilon))\frac{p\theta}{2}}|f|^p(v)\,dv\right)^{1/p}.\]
Thus, we conclude
\begin{equation}\label{range2}
T:L^p_{v^\frac{\alpha p\theta}{2}}([1,\infty))\times L^p_{\Omega^{-p\theta}(u,1)}([u_R(1),\infty))\to L^p_{v^{(\alpha -(1+2\varepsilon))\frac{p\theta}{2}}}([1,\infty))
\end{equation}
is a bounded linear operator and by Steps 5a and 5b its norm $\|T\|_{\theta,p}$  is bounded by
\begin{equation}\label{bop2}
\|T\|_{\theta,p}\leq \|T_0\|^{1-\theta}\|T_1\|^\theta \leq C(R)^{1-\theta}\cdot C(R,\alpha,\varepsilon)^{\theta/2}=\tilde{C}(R,\alpha,\varepsilon,p),
\end{equation}
where $C(R)$ and $C(R,\alpha,\varepsilon)$ are the constants from Propositions \ref{extl11} and \ref{l2ext22} respectively.

\noindent \textbf{Step 7:} We complete the proof. Let $\phi$ be a solution of the wave equation with smooth, spherically symmetric data on $S$.
Note that if 
\[\int_1^\infty v^\frac{\alpha p\theta}{2}|\partial_v\phi|^p(\infty,v)\,dv+\int_{u_R(1)}^\infty \Omega^{-p\theta}|\partial_u\phi|^p(u,1)\,du=+\infty,\]
then \eqref{vcase} trivially holds. So we assume 
\[\int_1^\infty v^\frac{\alpha p\theta}{2}|\partial_v\phi|^p(\infty,v)\,dv+\int_{u_R(1)}^\infty \Omega^{-p\theta}|\partial_u\phi|^p(u,1)\,du<+\infty,\]
namely we assume $(f,g):=(\partial_v\phi\vert_{\mathcal{H}^+\cap\{v\geq1\}}, \partial_u\phi\vert_{\underline{C}_1\cap\{u\geq u_R(1)\}})\in (A_0,A_1)_{\theta,p}$. But this implies $T(f,g)\in (B_0,B_1)_{\theta,p}=L^p_{v^{(\alpha -(1+2\varepsilon))\frac{p\theta}{2}}}([1,\infty))$ (using \eqref{range2}), and moreover by \eqref{bop2}
\begin{eqnarray*}
\int_1^\infty v^{(\alpha-(1+2\varepsilon))\frac{p\theta}{2}}|T(f,g)|^p(v)\,dv &\leq& \tilde{C}(R,\alpha,\varepsilon,p)^p\left(\int_1^\infty v^\frac{\alpha p\theta}{2}|f|^p(v)\,dv+\int_{u_R(1)}^\infty \Omega^{-p\theta}(u,1)|g|^p(u)\,du\right)\\
&=&\tilde{C}\left(\int_1^\infty v^\frac{\alpha p\theta}{2}|\partial_v\phi|^p(\infty,v)\,dv+\int_{u_R(1)}^\infty \Omega^{-p\theta}|\partial_u\phi|^p(u,1)|\,du\right).
\end{eqnarray*}
Comparing the previous inequality to \eqref{vcase}, we see that our proof of \eqref{vcase} will be complete if we show that $T(f,g)=\partial_v\phi\vert_{\gamma_R\cap\{v\geq1\}}$.

To see this, note that $(f,g)\in (A_0,A_1)_{\theta,p}\subseteq A_0+A_1\subseteq A_0+A_0$ since $A_1\subseteq A_0$ (as $\alpha>1$ and $\Omega^2(u,1)$ is integrable). But $A_0+A_0=A_0$ (with equal norms), so in particular $(f,g)\in A_0$. Moreover, as $\phi$ is a smooth solution of the wave equation, $(f,g):=(\partial_v\phi\vert_{\mathcal{H}^+\cap\{v\geq1\}}, \partial_u\phi\vert_{\underline{C}_1\cap\{u\geq u_R(1)\}})$ is smooth. Hence, by construction of $T$ (Step 5d), $T(f,g)=T_0(f,g)$ and by definition of $T_0$ (Step 5a), $T_0(f,g)=\partial_v\phi_1\vert_{\gamma_R\cap\{v\geq1\}}$. So $T(f,g)=\partial_v\phi\vert_{\gamma_R\cap\{v\geq1\}}$ and the proof is complete.
\end{proof}

\subsection{Proof of Propositions \ref{alt2} and \ref{alt1}}\label{easyres}

The remainder of this section is devoted to proving Propositions \ref{alt2} and \ref{alt1}, namely showing that for $R>r_+$ with $R^*\geq1$ and a smooth, spherically symmetric solution $\phi$ of the wave equation, we can control 
\[\int_{\gamma_{R}\cap\{v\geq 1\}}v^{3p-1}|\phi|^p\]
and 
\[\sup_{\gamma_{R}\cap\{v\geq R^*\}}v^3|\phi|\]
in terms of
\begin{equation}\label{vlih3}
I_{p,R}(\phi):=\left(\int_{\gamma_{R}\cap\{v\geq1\}}v^{4p}|\partial_v\phi|^p\right)^{1/p}+\left(\int_{\gamma_{R}\cap\{v\geq1\}}v^{4p}|\partial_u\phi|^p\right)^{1/p}.
\end{equation}
We first prove a preliminary lemma from which these estimates easily follow.

\begin{lemma}\label{Cauchy}Assume $p>1$.
Suppose $\phi$ is a solution of the wave equation \eqref{waveequation} with smooth, spherically symmetric data on $S$ satisfying \eqref{ic1} and \eqref{ic2}. Let $R\geq r_+$ be such that $R^*\geq 1$ and suppose $a>1-1/p$. Then for $v_0\geq R^*$, there exists $C=C(a,p)>0$ such that
\[|\phi|(u_{R^*}(v_0),v_0)\leq C   \left(\left(\int_{\gamma_R\cap\{v\geq1\}} v^{ap}|\partial_u\phi|^p\right)^{1/p}+\left(\int_{\gamma_R\cap\{v\geq1\}} v^{ap}|\partial_v\phi|^p\right)^{1/p}\right)  v_0^{1-a-1/p}.\]
\end{lemma}

\begin{proof}
Let $R$ be as above. It follows from the fundamental theorem of calculus–, together with the fact $\phi(t,R^*)\to0$ as $t\to \infty$ (see \cite{DR1}) that
\begin{eqnarray}
|\phi|(t_0,R^*)&\leq& \int_{t_0}^\infty|\partial_t\phi|(t,R^*)\,dt\notag\\
&=&\int_{t_0}^\infty  t^a|\partial_t\phi|(t,R^*)t^{-a}\,dt\notag\\
&\leq& \left(\int_{t_0}^\infty t^{ap}|\partial_t\phi|^p(t,R^*)\,dt\right)^{1/p}\left(\int_{t_0}^\infty t^{-ap'}\,dt\right)^{1/p'},\label{ecs}
\end{eqnarray}
where $1/p+1/p'=1$. Now, since $a>1-1/p$
\begin{eqnarray}
\left(\int_{t_0}^\infty t^{-ap'}\,dt\right)^{1/p'}=\left(\int_{t_0}^\infty t^{-\frac{ap}{p-1}}\,dt\right)^\frac{p-1}{p}
&=&C(t_0^{1-\frac{ap}{p-1}})^{\frac{p-1}{p}}
=Ct_0^{1-a-1/p},\hspace{8mm}\label{2ndcomp}
\end{eqnarray}
where $C=C(a,p)>0$.

On the other hand, since $v=\frac{t+R^*}{2}$ on $\gamma_R$ and $\partial_t=\frac{1}{2}(\partial_u+\partial_v)$, we have
\begin{eqnarray*}
\left(\int_{t_0}^\infty t^{ap}|\partial_t\phi|^p(t,R^*)\,dt\right)^{1/p}\hspace{130mm}\\
\leq\frac{1}{2}\left(\left(\int_{t_0}^\infty t^{ap}|\partial_u\phi|^p(t,R^*)\,dt\right)^{1/p}+\left(\int_{t_0}^\infty t^{ap}|\partial_v\phi|^p(t,R^*)\,dt\right)^{1/p}\right)\hspace{69mm}\\
= \frac{2^{1/p}}{2}\left(\left(\int_{\frac{t_0+R^*}{2}}^\infty (2v-R^*)^{ap}|\partial_u\phi|^p(u_{R^*}(v),v)\,dv\right)^{1/p}+\left(\int_{\frac{t_0+R^*}{2}}^\infty (2v-R^*)^{ap}|\partial_v\phi|^p(u_{R^*}(v),v)\,dv\right)^{1/p}\right)\hspace{20mm}\\
\leq C \left(\left(\int_{\frac{t_0+R^*}{2}}^\infty v^{ap}|\partial_u\phi|^p(u_{R^*}(v),v)\,dv\right)^{1/p}+\left(\int_{\frac{t_0+R^*}{2}}^\infty v^{ap}|\partial_v\phi|^p(u_{R^*}(v),v)\,dv\right)^{1/p}\right),\hspace{44mm}
\end{eqnarray*}
where $C=C(a,p)$.
Thus, if $v_0=\frac{t_0+R^*}{2}\geq R^*\geq1$, then 
\begin{eqnarray}
\left(\int_{t_0}^\infty t^{ap}|\partial_t\phi|^p(t,R^*)\,dt\right)^{1/p} \hspace{90mm}\notag\\
\leq C \left(\left(\int_{v_0}^\infty v^{ap}|\partial_u\phi|^p(u_{R^*}(v),v)\,dv\right)^{1/p}+\left(\int_{v_0}^\infty v^{ap}|\partial_v\phi|^p(u_{R^*}(v),v)\,dv\right)^{1/p}\right)\notag\\
\leq C\left(\left(\int_{\gamma_R\cap\{v\geq1\}} v^{ap}|\partial_u\phi|^p\right)^{1/p}+\left(\int_{\gamma_R\cap\{v\geq1\}} v^{ap}|\partial_v\phi|^p\right)^{1/p}\right),\hspace{21mm}\label{1stcomp}
\end{eqnarray}
where $C=C(a,p)$.
Thus for $v_0=\frac{1}{2}(t_0+R^*)\geq R^*$, substituting \eqref{2ndcomp} and \eqref{1stcomp} into \eqref{ecs} gives
\begin{eqnarray*}
|\phi|(u_{R^*}(v_0),v_0)&=&|\phi|(t_0,R^*)\\
&\leq&C \left(\left(\int_{\gamma_R\cap\{v\geq1\}} v^{ap}|\partial_u\phi|^p\right)^{1/p}+\left(\int_{\gamma_R\cap\{v\geq1\}} v^{ap}|\partial_v\phi|^p\right)^{1/p}\right)\cdot t_0^{1-a-1/p}\\
&\leq&C \left(\left(\int_{\gamma_R\cap\{v\geq1\}} v^{ap}|\partial_u\phi|^p\right)^{1/p}+\left(\int_{\gamma_R\cap\{v\geq1\}} v^{ap}|\partial_v\phi|^p\right)^{1/p}\right)\cdot v_0^{1-a-1/p},
\end{eqnarray*}
with $C=C(a,p)$, where we have used the fact that $t_0=2v_0-R^*=v_0+(v_0-R^*)\geq v_0\geq 0$ and $1-a-1/p<0$, and hence $t_0^{1-a-1/p}\leq v_0^{1-a-1/p}$.
\end{proof}

\subsubsection{Proof of Proposition \ref{alt2}}
From the above Lemma, we can conclude Proposition \ref{alt2}, which is restated below for convenience.
\begin{corollary*}Assume $p>1$.
Suppose $\phi$ is a solution of the wave equation \eqref{waveequation} with smooth, spherically symmetric data on $S$ satisfying \eqref{ic1} and \eqref{ic2} and let $R>r_+$ be such that $R^*\geq1$. Then there exist $C=C(R,\phi,p)>0$ and $\tilde{C}=\tilde{C}(R,p)>0$ such that
\begin{eqnarray*}
\left(\int_1^\infty v^{3p-1}|\phi|^p(u_{R^*}(v),v)\,dv\right)^{1/p}\leq C+\tilde{C}I_{p,R}(\phi),
\end{eqnarray*}
where $I_{p,R}(\phi)$ is as in \eqref{vlih3}.
\end{corollary*}

\begin{proof}
By Lemma \ref{Cauchy} with $a=4$, we have
\[|\phi|(u_{R^*}(v),v)\leq C_p\cdot I_{p,R}(\phi) v^{-(3+1/p)} \hspace{5mm}\mbox{for }v\geq R^*,\]
so
\begin{eqnarray*}
\left(\int_1^\infty v^{3p-1}|\phi|^p(u_{R^*}(v),v)\,dv\right)^{1/p}\hspace{130mm}\\
\leq 2\left(\int_1^{R^*} v^{3p-1}|\phi|^p(u_{R^*}(v),v)\,dv\right)^{1/p} +\left(\int_{R^*}^\infty v^{3p-1}|\phi|^p(u_{R^*}(v),v)\,dv\right)^{1/p}\hspace{54mm}\\
\leq C(R,\phi,p)+C_p\left[\int_{R^*}^\infty v^{3p-1}\left(I_{p,R}(\phi)\cdot v^{-(3+1/p)}\right)^p\,dv\right]^{1/p}\hspace{78mm}\\
\leq C(R,\phi,p)+C_p\cdot I_{p,R}(\phi)\left(\int_{R^*}^\infty v^{-2}\,dv\right)^{1/p}\hspace{101mm}\\
\leq C(R,\phi,p)+\tilde{C}(R,p)\cdot I_{p,R}(\phi),\hspace{120mm}
\end{eqnarray*}
as desired.
\end{proof}

\subsubsection{Proof of Proposition \ref{alt1}}
We now close the proof of Theorem \ref{reduction} by giving the proof of Proposition \ref{alt1}, which is stated below.
\begin{corollary*}Assume $p>1$.
Suppose $R>r_+$ is such that $R^*\geq1$. Then, given a solution $\phi$ of the wave equation \eqref{waveequation} with smooth, spherically symmetric data on $S$ satisfying \eqref{ic1} and \eqref{ic2}, there exists $C=C(R,p)>0$ such that
\[\sup_{\gamma_R\cap\{v\geq R^*\}}v^3|\phi|\leq C\cdot I_{p,R}(\phi).\]
\end{corollary*}

\begin{proof}
From Lemma \ref{Cauchy} with $a=4$, for $v_0\geq R^*$ we have
\begin{eqnarray*}
|\phi|(u_{R^*}(v_0),v_0)\leq C\cdot I_{p,R}(\phi)\cdot v_0^{-(3+1/p)}
\end{eqnarray*}
with $C=C(p)$,
and hence
\begin{eqnarray*}
v_0^3|\phi|(u_{R^*}(v_0),v_0)\leq C\cdot I_{p,R}(\phi)\cdot v_0^{-1/p}\leq C\cdot (R^*)^{-1/p}\cdot I_{p,R}(\phi),
\end{eqnarray*}
and the result follows by taking the supremum.
\end{proof}

\newpage
\appendix

\section{Interpolation Theory}\label{Interpolation Theory}

\subsection{Interpolation of Normed Spaces}\label{INS}

This section provides a review of some definitions, terminology and results from interpolation theory. This material presented here is covered in \cite{BL}.

Let $(A_0,\|\cdot\|_{A_0})$, $(A_1,\|\cdot\|_{A_1})$ be normed spaces. 

\begin{definition}[Compatible Couple]
We say that $\overline{A}=(A_0,A_1)$ is a \emph{compatible couple} of normed vector spaces if there is a Hausdorff topological vector space $\mathcal{U}$ such that $A_0$ and $A_1$ are subspaces of $\mathcal{U}$.
\end{definition}

\noindent If $(A_0,A_1)$ is a compatible couple, we define
\[A_0+A_1=\{a\in \mathcal{U}: a=a_0+a_1, a_i\in A_i\}.\]
\noindent Then  $(A_0\cap A_1, \|\cdot\|_{A_0\cap A_1})$ and $(A_0+A_1,\|\cdot\|_{A_0+A_1})$ are normed spaces, where
\[\|a\|_{A_0\cap A_1}=\max(\|a\|_{A_0},\|a\|_{A_1})\]
and
\[\|a\|_{A_0+A_1}=\inf\{\|a_0\|_{A_0}+\|a_1\|_{A_1}: a=a_0+a_1, a_i\in A_i\}.\]

\begin{definition}[Intermediate Space]
Let $\overline{A}=(A_0,A_1)$ be a a compatible couple of normed spaces. A normed space $A$ is said to be an \emph{intermediate space with respect to $\overline{A}$} if 
\[A_0\cap A_1\subset A\subset A_0+A_1\]
and these inclusions are continuous.
\end{definition}

Given a pair of compatible couples $\overline{A}=(A_0,A_1)$ and $\overline{B}=(B_0,B_1)$ and a linear operator $T:A_0+A_1\to B_0+B_1$, we write $T:\overline{A}\to\overline{B}$ if
\[T\vert_{A_0}:A_0\to B_0, \hspace{3mm} T\vert_{A_1}:A_1\to B_1\]
and these maps are bounded.

\begin{definition}[Interpolation Spaces]
Given compatible couples $\overline{A}$ and $\overline{B}$, we say intermediate spaces $A$ and $B$ with respect to $\overline{A}$ and $\overline{B}$ are \emph{interpolation spaces with respect to $\overline{A}$ and $\overline{B}$} if 
\[T:\overline{A}\to \overline{B} \hspace{5mm}\mbox{ implies } \hspace{5mm}T\vert_A:A\to B\]
and moreover this map is bounded. Furthermore, we say the interpolation spaces $A$ and $B$ are exact of exponent $\theta$ if 
\[\|T\vert_A\|\leq \|T\vert_{A_0}\|^{1-\theta}\cdot\|T\vert_{A_1}\|^\theta.\]
\end{definition}

\subsection{The K-Method of Real Interpolation}

Let $\overline{A}$ be a compatible couple of normed spaces. For $t>0$ and $a\in A_0+A_1$ we define
\[K(t,a):=K(t,a;\overline{A})=\inf\{\|a_0\|_{A_0}+t\|a_1\|_{A_1}:a=a_0+a_1, a_i\in A_i\}.\]
For each fixed $t>0$, $\|\cdot\|_{A_1}$ and $t\|\cdot\|_{A_1}$ are equivalent norms on $A_1$, and hence $\|\cdot\|_{A_0+A_1}$ and $K(t,\cdot)$ are equivalent norms on $A_0+A_1$. In fact,  
\begin{equation}\label{equiv}
K(t,a)\leq \max(1,t/s)K(s,a)\hspace{5mm}\forall\,s,t>0.
\end{equation}

Let $1\leq q \leq+\infty$ and $0<\theta<1$. We define a functional $\Phi_{\theta,q}$ by
\begin{equation}\label{Phi}
\Phi_{\theta,q}(\varphi):=\left\{\begin{array}{lc}\left ( \int_0^\infty(t^{-\theta}\varphi(t))^q\,\frac{dt}{t}\right)^{1/q},&\mbox{if } q<\infty\\ \sup_{t>0}|t^{-\theta}\varphi(t)|,&\mbox{if } q=\infty\end{array}\right.
\end{equation}
for non-negative functions $\varphi:(0,\infty)\to [0,\infty)$.
\begin{definition}
We define
\[\overline{A}_{\theta,q}:=\{a\in A_0+A_1:\Phi_{\theta,q}(K(\cdot,a))<+\infty\}\]
and, for $a\in A_0+A_1$, we set $\|a\|_{\overline{A}_{\theta,q}}=\Phi_{\theta,q}(K(\cdot,a))$.
\end{definition}

\noindent
Since $K(t,a)$ is a norm on $A_0+A_1$ for each $t>0$ and since $\Phi_{\theta,q}$ satisfies all the properties of a norm, it follows that $\|\cdot\|_{\overline{A}_{\theta,q}}$ is a norm on $\overline{A}_{\theta,q}$.

\begin{theorem}\label{kmethod}
Let $\overline{A}$ and $\overline{B}$ be two compatible couples of normed spaces and let $0<\theta<1$ and $1\leq q \leq+\infty$. Then $\overline{A}_{\theta,q}$ and $\overline{B}_{\theta,q}$ are interpolation spaces with respect to $\overline{A}$ and $\overline{B}$, namely
\[T:\overline{A}\to\overline{B}\hspace{5mm} \Longrightarrow\hspace{5mm} T:\overline{A}_{\theta,q}\to\overline{B}_{\theta,q}\]
and the map $T\vert_{\overline{A}_{\theta,q}}$ is bounded. Furthermore, they are exact interpolation spaces of exponent $\theta$, namely
\[\|T\vert_{\overline{A}_{\theta,q}}\|\leq \|T\vert_{A_0}\|^{1-\theta}\cdot\|T\vert_{A_1}\|^\theta.\]
\end{theorem}

\begin{proof}
First we prove
\begin{equation}\label{int}
K(s,a;\overline{A})\leq\gamma_{\theta,q}s^\theta \|a\|_{\overline{A}_{\theta,q}}
\end{equation}
for some constant $\gamma_{\theta,q}$ depending only on $\theta$ and $q$.
Indeed, by \eqref{equiv} we have
\[\min(1,t/s)K(s,a)\leq K(t,a),\]
and applying $\Phi_{\theta,q}$ to this inequality (where $t$ is the variable of integration) yields
\begin{equation}\label{inequ}\Phi_{\theta,q}(\min(1,t/s)K(s,a))= \Phi_{\theta,q}(\min(1,t/s))K(s,a) \leq \|a\|_{\overline{A}_{\theta,q}}.
\end{equation}
But for any non-negative function $\varphi$, if $q<\infty$
\begin{eqnarray}
\Phi_{\theta,q}(\varphi(t/s))&=&\left( \int_0^\infty (t^{-\theta}\varphi(t/s))^q\,\frac{dt}{t}\right)^{1/q}\notag\\
&=&s^{-\theta}\left(\int_0^\infty((t/s)^{-\theta}\varphi(t/s))^q\,\frac{d(t/s)}{(t/s)}\right)^{1/q}\notag\\
&=&s^{-\theta}\Phi_{\theta,q}(\varphi(t)),\label{st}
\end{eqnarray}
 and similarly if $q=\infty$.
Hence, by \eqref{inequ}
\[s^{-\theta}\Phi_{\theta,q}(\min(1,t))K(s,a)\leq \|a\|_{\overline{A}_{\theta,q}},\]
which is \eqref{int} with $\gamma_{\theta,q}=\Phi_{\theta,q}(min(1,t))^{-1}<\infty$.

Now that we have established \eqref{int}, the theorem follows in a straightforward manner. We first show that $\overline{A}_{\theta,q}$ and $\overline{B}_{\theta,q}$ are intermediate spaces with respect to $\overline{A}$ and $\overline{B}$ respectively. Indeed, taking $a\in \overline{A}_{\theta,q}$ and $s=1$ in \eqref{int} yields
\[\|a\|_{A_0+A_1}=K(1,a;\overline{A})\leq\gamma_{\theta,q}\|a\|_{\overline{A}_{\theta,q}},\]
and in particular the inclusion $\overline{A}_{\theta,q}\subset A_0+A_1$ is continuous. Furthermore, the inclusion $A_0\cap A_1\subset \overline{A}_{\theta,q}$ is continuous, since for $a\in A_0\cap A_1$ by definition
\[K(t,a)\leq \min(1,t)\|a\|_{A_0\cap A_1},\]
so applying the functional $\Phi_{\theta,q}$ gives
\[\|a\|_{\overline{A}_{\theta,q}}=\Phi_{\theta,q}(K(t,a))\leq \Phi_{\theta,q}(\min(1,t))\|a\|_{A_0+A_1}.\]
In exactly the same way the inclusions $B_0\cap B_1\subset \overline{B}_{\theta,q}\subset B_0+B_1$ are continuous. Thus $\overline{A}_{\theta,q}$ and $\overline{B}_{\theta,q}$ are intermediate spaces with respect to $\overline{A}$ and $\overline{B}$ respectively.

It remains now to show that $\overline{A}_{\theta,q}$ and $\overline{B}_{\theta,q}$ are exact interpolation spaces of exponent $\theta$ with respect to $\overline{A}$ and $\overline{B}$. Thus assume $T:\overline{A}\to \overline{B}$. Then
\[T\vert_{A_j}:A_j\to B_j \hspace{3mm}\mbox{ for }j=0,1\]
are bounded and for brevity we set
\[M_j:=\|T\vert_{A_j}\|=\|T\|_{A_j,B_j}, \hspace{3mm}j=0,1.\]
Then, for $a\in\overline{A}_{\theta,q}$, we have
\begin{eqnarray*}
K(t,Ta;\overline{B})&:=&\inf\{\|b_0\|_{B_0}+t\|b_1\|_{B_1}:Ta=b_0+b_1, b_j\in B_j\}\\
&\leq&\inf\{\|Ta_0\|_{B_0}+t\|Ta_1\|_{B_1}:a=a_0+a_1, a_j\in A_j\}\\
&\leq&\inf\{M_0\|a_0\|_{A_0}+tM_1\|a_1\|_{A_1}:a=a_0+a_1, a_j\in A_j\}\\
&=&M_0 \inf\{\|a_0\|_{A_0}+\frac{tM_1}{M_0}\|a_1\|_{A_1}:a=a_0+a_1, a_j\in A_j\}\\
&=&M_0 K\left(\frac{tM_1}{M_0},a;\overline{A}\right).
\end{eqnarray*}
Hence, applying $\Phi_{\theta,q}$ and using \eqref{st} with $s=M_0/M_1$, we see
\begin{equation}\label{exact}\|Ta\|_{\overline{B}_{\theta,q}}\leq M_0^{1-\theta}M_1^\theta\|a\|_{\overline{A}_{\theta,q}}.
\end{equation}
So $a\in \overline{A}_{\theta,q}\Longrightarrow \|a\|_{\overline{A}_{\theta,q}}<+\infty\Longrightarrow \|Ta\|_{\overline{B}_{\theta,q}}< +\infty\Longrightarrow Ta\in \overline{B}_{\theta,q}$, so $T:\overline{A}_{\theta,q}\to\overline{B}_{\theta,q}$, and moreover by \eqref{exact} this map is bounded. Thus $\overline{A}_{\theta,q}$ and $\overline{B}_{\theta,q}$ are interpolation spaces with respect to $\overline{A}$ and $\overline{B}$, and by \eqref{exact} they are exact of exponent $\theta$.

\end{proof}

\subsection{The Power Theorem}

So far, all of our interpolation results have been formulated for normed spaces, but in fact strict triangle inequalities are not necessary. The theory works just as well if we relax the assumptions on the various norms so that they satisfy quasi-triangle inequalities
\[\|a+b\|\leq C( \|a\|+\|b\|) \hspace{3mm}\mbox{ for some } C>0,\]
(where $C$ is independent of $a,b$) instead of strict triangle inequalities, and we refer to them as \emph{quasinorms}. Indeed, it is shown in \cite{BL} that, given a compatible couple $\overline{A}$ of quasinormed spaces, $\overline{A}_{\theta,q}$ is a quasinormed space. Moreover, by the same proof as for Theorem \ref{kmethod}, if $\overline{A}$ and $\overline{B}$ are compatible couples of quasinormed spaces, then $\overline{A}_{\theta,q}$ and $\overline{B}_{\theta,q}$ are exact interpolation spaces of exponent $\theta$ with respect to $\overline{A}$ and $\overline{B}$.

Given a quasinorm $\|\cdot\|$ on a vector space $A$, we note that for any $\rho>0$, $\|\cdot\|^\rho$ is also a quasinorm. We denote the quasinormed space $(A,\|\cdot\|^\rho)$ by $A^\rho$.

One of our main tools will be the Power Theorem, which, under certain assumptions on the exponents and indices, relates  interpolation spaces of powers of a compatible couple  to powers of interpolation spaces of that compatible couple.

\begin{theorem}[Power Theorem]\label{power}
Let $\rho_0,\rho_1>0$ be given. For $0<\eta<1$ and $0<r\leq+\infty$, set
\begin{eqnarray*}
\rho&=&(1-\eta)\rho_0+\eta \rho_1,\\
\theta&=&\frac{\eta\rho_1}{\rho},\\
q&=&\rho r.
\end{eqnarray*}
Then given a compatible couple $\overline{A}=(A_0,A_1)$ of normed spaces, we have
\[(A_0^{\rho_0},A_1^{\rho_1})_{\eta,r}=((A_0,A_1)_{\theta,q})^\rho=(\overline{A}_{\theta,q})^\rho\]
with equivalent quasinorms.
\end{theorem}
\noindent The proof below is sketched in \cite{BL}.

\begin{proof}
Given any compatible couple $\overline{A}=(A_0,A_1)$ of normed spaces, consider the functional
\[K_\infty(t,a)=K_\infty(t,a;\overline{A})=\inf_{\substack{a=a_0+a_1\\a_i\in A_i}}\max(\|a_0\|_{A_0},t\|a_1\|_{A_1}).\]
Since 
\begin{equation}\label{last5}K_\infty(t,a)\leq K(t,a)\leq 2K_\infty(t,a),
\end{equation}
the norm $\|a\|_{\theta,q}$ on $\overline{A}_{\theta,q}$ is equivalent to $\Phi_{\theta,q}(K_\infty(t,a))$, so it is sufficient to prove the Power Theorem using the functional $K_\infty$ instead of $K$.

First, we will show 
\begin{equation}\label{first}
K_\infty(s,a;A_0^{\rho_0},A_1^{\rho_1})=K_\infty(t,a;A_0,A_1)^{\rho_0},
\end{equation}
where
\begin{equation}\label{second}
s=t^{\rho_1}K_\infty(t,a;A_0,A_1)^{\rho_0-\rho_1}.
\end{equation}
For ease of notation, we shall write 
\[K_\infty(t)=K_\infty(t,a;A_0,A_1), \hspace{7mm}K_\infty(s)=K_\infty(s,a;A_0^{\rho_0},A_1^{\rho_1}).\]
Given $a\in A_0+A_1$ and $\varepsilon >0$, we may choose $a_0$ and $a_1$ such that
\[K_\infty(t)\leq \max(\|a_0\|_{A_0},t\|a_1\|_{A_1})\leq (1+\varepsilon)K_\infty(t).\]
Thus, since at least one of 
\[\frac{\|a_0\|_{A_0}}{K_\infty(t)}, \hspace{5mm} \frac{t\|a_1\|_{A_1}}{K_\infty(t)}\]
is larger than $1$, it follows that
\[1\leq \max\left(\left(\frac{\|a_0\|_{A_0}}{K_\infty(t)}\right)^{\rho_0},\left(\frac{t\|a_1\|_{A_1}}{K_\infty(t)}\right)^{\rho_1}\right)\leq 1+\varepsilon',\]
where $\varepsilon'\to0$ as $\varepsilon\to0$. Hence,
\begin{eqnarray*}
1&=&\inf_{\substack{a=a_0+a_1\\a_i\in A_i}} \max\left(\left(\frac{\|a_0\|_{A_0}}{K_\infty(t)}\right)^{\rho_0},\left(\frac{t\|a_1\|_{A_1}}{K_\infty(t)}\right)^{\rho_1}\right)  \\
&=&\inf_{\substack{a=a_0+a_1\\a_i\in A_i^{\rho_i}}} \max\left(\left(\frac{\|a_0\|_{A_0}}{K_\infty(t)}\right)^{\rho_0},\left(\frac{t\|a_1\|_{A_1}}{K_\infty(t)}\right)^{\rho_1}\right).
\end{eqnarray*}
and multiplying by $K_\infty(t)^{\rho_0}$ gives \eqref{first}.

Now, if $q=\infty$, then $r=\infty$ and so, by \eqref{first} and \eqref{second}
\begin{eqnarray}
\|a\|_{(A_0^{\rho_0},A_1^{\rho_1})_{\eta,r}} &\sim& \sup_{s>0}s^{-\eta}K_\infty(s)\notag\\
&=&\sup_{t>0}(t^{\rho_1}K_\infty(t,a;A_0,A_1)^{\rho_0-\rho_1})^{-\eta}K_\infty(t,a;A_0,A_1)^{\rho_0})\notag\\
&=& \sup_{t>0}t^{-\eta\rho_1}K_\infty(t)^{\rho_0(1-\eta)+\rho_1\eta}\notag\\
&=& \sup_{t>0}(t^{-\theta}K_\infty(t))^{\rho}\notag\\
&\sim& \|a\|_{(A_0,A_1)_{\theta,q=\infty}}^\rho,\notag
\end{eqnarray}
as required.

On the other hand, if $q<\infty$, then $r<\infty$ and by \eqref{last5}
\begin{equation}\label{basic4}
\int_0^\infty (s^{-\eta}K_\infty(s))^r\,\frac{ds}{s} \leq \|a\|_{(A_0^{\rho_0},A_1^{\rho_1})_{\eta,r}}^r= \int_0^\infty (s^{-\eta}K(s))^r\,\frac{ds}{s}\leq 2^r \int_0^\infty (s^{-\eta}K_\infty(s))^r\,\frac{ds}{s}.
\end{equation}
Now,
\begin{eqnarray}
 \int_0^\infty (s^{-\eta}K_\infty(s))^r\,\frac{ds}{s}
&=&\int_0^\infty s^{-\eta r-1}K_\infty(s)^r\,ds\notag\\
&=&\frac{-1}{\eta r}\int_0^\infty K_\infty(s)^r\,d(s^{-\eta r})\notag\\
&=& \frac{-1}{\eta r}\left[\frac{K_\infty(s)^r}{s^{\eta r}}\right]_0^\infty +\frac{1}{\eta r}\int_0^\infty s^{-\eta r}\,d(K_\infty(s)^r).\hspace{3mm}\label{third}
\end{eqnarray}
But by \eqref{first} and \eqref{second} we have
\begin{eqnarray}
\int_0^\infty s^{-\eta r}\,d(K_\infty(s)^r)&=&\int_0^\infty (t^{\rho_1}K_\infty(t)^{\rho_0-\rho_1})^{-\eta r}\,d(K_\infty(t)^{\rho_0 r})\notag\\
&=&\int_0^\infty t^{-\rho_1\eta r}K_\infty(t)^{\rho_1\eta r-\rho_0\eta r}\,d(K_\infty(t)^{\rho_0 r})\notag\\
&=&\rho_0 r\int_0^\infty t^{-\rho_1\eta r}K_\infty(t)^{\rho_1\eta r+\rho_0r(1-\eta) -1}\,dK_\infty(t)\notag\\
&=&\rho_0 r\int_0^\infty t^{-\theta\rho r} K_\infty(t)^{\rho r-1}\,dK_\infty(t)\notag\\
&=&\rho_0 r\int_0^\infty t^{-\theta q}K_\infty(t)^{q-1}\,dK_\infty(t)\notag\\
&=&\frac{\rho_0 r}{q}\int_0^\infty t^{-\theta q}\,d(K_\infty(t)^q).\notag
\end{eqnarray}
So, integration by parts gives
\begin{eqnarray}
\int_0^\infty s^{-\eta r}\,d(K_\infty(s)^r)&=&\frac{\rho_0 r}{q}\left[\frac{K_\infty(t)^q}{t^{\theta q}}\right]_0^\infty -\frac{\rho_0 r}{q}\int_0^\infty K_\infty(t)^q \,d(t^{-\theta q})\notag\\
&=&\frac{\rho_0 r}{q} \left[\frac{K_\infty(t)^q}{t^{\theta q}}\right]_0^\infty +\frac{\rho_0 r\theta q}{q}\int_0^\infty t^{-\theta q-1}K_\infty(t)^q\,dt\notag\\
&=&\frac{\rho_0 r}{q}\left[\frac{K_\infty(t)^q}{t^{\theta q}}\right]_0^\infty +\rho_0 r\theta \int_0^\infty t^{-\theta q-1}K_\infty(t)^q\,dt.\notag
\end{eqnarray}
Combining the last equation with \eqref{third} yields
\begin{eqnarray}
 \int_0^\infty (s^{-\eta}K_\infty(s))^r\,\frac{ds}{s}&=&\frac{-1}{\eta r}\left[\frac{K_\infty(s)^r}{s^{\eta r}}\right]_0^\infty +\frac{\rho_0}{\eta q}\left[\frac{K_\infty(t)^q}{t^{\theta q}}\right]_0^\infty +\frac{\rho_0\theta}{\eta }\int_0^\infty t^{-\theta q-1}K_\infty(t)^q\,dt.\hspace{10mm}\label{forth}
\end{eqnarray}
But by \eqref{first} and \eqref{second}
\begin{equation}\label{last}
\frac{K_\infty(s)^r}{s^{\eta r}}=\left(\frac{K_\infty(t)^{\rho_0}}{(t^{\rho_1}K_\infty(t)^{\rho_0-\rho_1})^\eta}\right)^r=\frac{K_\infty(t)^{\rho r}}{t^{\rho_1\eta r}} =\frac{K_\infty(t)^q}{t^{\theta q}}.
\end{equation}

Suppose $a\in(A_0,A_1)_{\theta,q}^\rho$. Then by \eqref{last5}
\begin{equation}\label{last2}
\int_0^\infty (t^{-\theta}K_\infty(t))^q\,\frac{dt}{t}\leq \|a\|_{(A_0,A_1)_{\theta,q}}^q<+\infty,
\end{equation}
and it follows that 
\begin{equation}\label{star3}
\lim_{t\to0}(t^{-\theta}K_\infty(t))^q=0 \hspace{5mm}\mbox{ and } \hspace{5mm}\lim_{t\to\infty} (t^{-\theta}K_\infty(t))^q=0.
\end{equation}
Indeed, these limits must exist in $[0,\infty]$ because $K_\infty$ is an increasing function of $t$, and if either is positve then $\|a\|_{(A_0,A_1)_{\theta,q}}$ would diverge. For instance, if
\[L:= \lim_{t\to\infty}( t^{-\theta}K_\infty(t))^q>0,\]
then for $\varepsilon=L/2>0$, there exists $N>0$ sufficiently large such that
\[|(t^{-\theta}K_\infty(t))^q-L|< L/2 \hspace{3mm}\mbox{ for } t>N,\]
and in particular 
\[L/2<(t^{-\theta}K_\infty(t))^q \hspace{3mm}\mbox{ for }t>N.\]
Thus
\[\int_N^\infty (t^{-\theta}K_\infty(t))^q\,\frac{dt}{t}>\int_N^\infty \frac{L}{2t}\,dt=\infty,\]
a contradiction with \eqref{last2}. So combining \eqref{star3} and \eqref{last}, we see that if  $a\in(A_0,A_1)_{\theta,q}^\rho$, then the boundary terms in \eqref{forth} vanish and so by \eqref{forth} and \eqref{last5}
\[ \int_0^\infty (s^{-\eta}K_\infty(s))^r\,\frac{ds}{s}=C\int_0^\infty t^{-\theta q-1}K_\infty(t)^q\,dt\leq C\int_0^\infty (t^{-\theta }K(t))^q\,\frac{dt}{t}=C\|a\|_{(A_0,A_1)_{\theta,q}}^q.\]
So from \eqref{basic4}
\[\|a\|_{(A_0^{\rho_0},A_1^{\rho_1})_{\eta,r}}^r \leq 2^rC\|a\|_{(A_0,A_1)_{\theta,q}}^q\]
for some $C>0$.  Hence 
\begin{equation*}
\|a\|_{(A_0^{\rho_0},A_1^{\rho_1})_{\eta,r}}\leq C' \|a\|_{(A_0,A_1)_{\theta,q}}^{q/r}=C'\|a\|_{(A_0,A_1)_{\theta,q}}^\rho.
\end{equation*}
It follows that $a\in(A_0,A_1)_{\theta,q}^\rho \Longrightarrow a\in(A_0^{\rho_0},A_1^{\rho_1})_{\eta,r}$.

On the other hand, if $a\in (A_0^{\rho_0},A_1^{\rho_1})_{\eta,r}$, then
\[\int_0^\infty (s^{-\eta}K_\infty(s))^r\,\frac{ds}{s}\leq \|a\|_{(A_0^{\rho_0},A_1^{\rho_1})_{\eta,r}}^r<\infty.\]
Arguing exactly as before, we find that
\[\lim_{s\to0}(s^{-\eta}K_\infty(s))^r=0 \hspace{5mm}\mbox{ and } \hspace{5mm}\lim_{s\to\infty} (s^{-\eta}K_\infty(s))^r=0,\]
 so   it follows from \eqref{last} that all the boundary terms in \eqref{forth} must vanish. 
So by \eqref{forth} and \eqref{last5} we get 
\[\int_0^\infty (s^{-\eta}K_\infty(s))^r\,\frac{ds}{s}=C\int_0^\infty t^{-\theta q-1}K_\infty(t)^q\,dt\geq 2^{-q}C\int_0^\infty (t^{-\theta }K(t))^q\,\frac{dt}{t}=c\|a\|_{(A_0,A_1)_{\theta,q}}^q.\]
So by \eqref{basic4}
\[\|a\|_{(A_0^{\rho_0},A_1^{\rho_1})_{\eta,r}}^r\geq c\|a\|_{(A_0,A_1)_{\theta,q}}^q.\]
Thus
\begin{equation*}
\|a\|_{(A_0^{\rho_0},A_1^{\rho_1})_{\eta,r}}\geq c' \|a\|_{(A_0,A_1)_{\theta,q}}^\rho.
\end{equation*}
It follows that $a\in(A_0^{\rho_0},A_1^{\rho_1})_{\eta,r} \Longrightarrow a\in(A_0,A_1)_{\theta,q}^\rho$. 

Thus, we have shown that $a\in(A_0^{\rho_0},A_1^{\rho_1})_{\eta,r} \Longleftrightarrow a\in(A_0,A_1)_{\theta,q}^\rho$, and that the quasinorms on the spaces are equivalent, namely 
\[(A_0,A_1)_{\theta,q}^\rho=(A_0^{\rho_0},A_1^{\rho_1})_{\eta,r}\]
with equivalent quasinorms, as required.
\end{proof}

\newpage
\section{A Technical Lemma}\label{Switch}

The proof of Theorem \ref{interior reduction} will be complete once we prove the following lemma.

\begin{lemma}\label{interplem}
Let $f\in A_0^{p_0}+A_1^{p_1}$, where $A_0=A_{p_0}(w_0)$, $A_1=A_{p_1}(w_1)$, $p_0$ and $p_1$ are as in the proof of Theorem \ref{interior reduction}. (In particular, $w_0,w_1:[1,\infty)\to (0,\infty)$ are assumed to be positive, measurable functions.) Then
\begin{equation}\label{switchf}
\inf_{\substack{f=f_0+f_1\\f_i\in L^{p_i}_{w_i}([1,\infty))}} \int_1^\infty w_0|f_0|^{p_0}(v)+tw_1|f_1|^{p_1}(v)\,dv = \int_1^\infty\inf_{\substack{f=f_0+f_1\\f_i\in L^{p_i}_{w_i}([1,\infty))}} w_0|f_0|^{p_0}(v)+tw_1|f_1|^{p_1}(v)\,dv
\end{equation}
for every $t>0$.
\end{lemma}

\begin{proof}
It is clear that the LHS of \eqref{switchf} $\geq$ the RHS for every $t>0$, so we need only show
\begin{equation}\label{star7}
\inf_{\substack{f=f_0+f_1\\f_i\in L^{p_i}_{w_i}}} \int_1^\infty w_0|f_0|^{p_0}(v)+tw_1|f_1|^{p_1}(v)\,dv \leq \int_1^\infty\inf_{\substack{f=f_0+f_1\\f_i\in L^{p_i}_{w_i}}} w_0|f_0|^{p_0}(v)+tw_1|f_1|^{p_1}(v)\,dv
\end{equation}
for  every $t>0$.

Let $t>0$ and set
\begin{eqnarray*}
F_t(v)&:=&\inf_{\substack{f=f_0+f_1\\f_i\in L^{p_i}_{w_i}}} w_0|f_0|^{p_0}(v)+tw_1|f_1|^{p_1}(v).
\end{eqnarray*}
If
\[\int_1^\infty F_t(v)\,dv=+\infty\]
then \eqref{star7} certainly holds, so we assume
\begin{equation}\label{star4}
\int_1^\infty F_t(v)\,dv<+\infty.
\end{equation}

Suppose  there are fucntions $h_0\in L^{p_0}_{w_0}$ and $h_1\in L^{p_1}_{w_1}$ (which may depend on $t$) with $h_0+h_1=f$ such that at each point $v\geq 1$, the pointwise infimum is attained at $h_0, h_1$, namely for each $v\geq 1$
\[\inf_{\substack{f=f_0+f_1\\f_i\in L^{p_i}_{w_i}}} w_0|f_0|^{p_0}(v)+tw_1|f_1|^{p_1}(v)= w_0|h_0|^{p_0}(v)+tw_1|h_1|^{p_1}(v).\]
Then 
\begin{eqnarray*}
\int_1^\infty\inf_{\substack{f=f_0+f_1\\f_i\in L^{p_i}_{w_i}}} w_0|f_0|^{p_0}(v)+tw_1|f_1|^{p_1}(v)\,dv&=&\int_1^\infty w_0|h_0|^{p_0}(v)+tw_1|h_1|^{p_1}(v)\,dv \\
&\geq& \inf_{\substack{f=f_0+f_1\\f_i\in L^{p_i}_{w_i}}} \int_1^\infty w_0|f_0|^{p_0}(v)+tw_1|f_1|^{p_1}(v)\,dv,
\end{eqnarray*}
 as desired. So we are done once we prove the existence of the functions $h_0$ and $h_1$.

If a measurable function $\alpha(v)=\alpha_t(v)$ satisfies 
\begin{equation}\label{choosea}
F_t(v)=w_0|\alpha(v)|^{p_0}+tw_1(v)|f(v)-\alpha(v)|^{p_1},
\end{equation}
then integrating gives 
\begin{eqnarray*}
\int_1^\infty w_0|\alpha(v)|^{p_0}+tw_1(v)|f(v)-\alpha(v)|^{p_1}\,dv&=&\int_1^\infty F_t(v)\,dv<\infty,
\end{eqnarray*}
by \eqref{star4}.
 Thus $\alpha\in L^{p_0}_{w_0}$ and $f-\alpha\in L^{p_1}_{w_1}$ (note $f-\alpha$ is measurable as $f\in L^{p_0}_{w_0}+L^{p_1}_{w_1}$ so $f$ and $\alpha$ are both measurable) and we may choose $h_0=\alpha$ and $h_1=f-\alpha$. So it remains only to show that for each $v\geq1$ we can choose a measurable function $\alpha(v)$ satisfying \eqref{choosea}. 
 
Note that for each $v\geq1$
 \begin{eqnarray*}
 \inf_{\substack{f=f_0+f_1\\f_i\in L^{p_i}_{w_i}}} w_0|f_0|^{p_0}(v)+tw_1|f_1|^{p_1}(v)&=& tw_1(v)  \inf_{\substack{f_0\in L^{p_0}_{w_0}\\f-f_0\in L^{p_1}_{w_1}}}  \frac{w_0}{tw_1}(v) |f_0(v)|^{p_0}+|f(v)-f_0(v)|^{p_1}\\
 &=&tw_1(v)  \inf_{x\in\mathbb{R}}  \frac{w_0}{tw_1}(v) |x|^{p_0}+|f(v)-x|^{p_1}\\
 &=&tw_1(v) \inf_{x\in\mathbb{R}} G_v(x),
 \end{eqnarray*}
where
\[G_v(x)=\frac{w_0}{tw_1}(v)|x|^{p_0}+|f(v)-x|^{p_1}.\]
We claim that for each $v\geq1$, $G_v$ has a (unique) minimiser, $x_0=x_0(v)$ say. Furthermore, we claim that $x_0:[1,\infty)\to \mathbb{R}$ is a measurable function. It follows that 
\begin{eqnarray*}
F_t(v)= \inf_{\substack{f=f_0+f_1\\f_i\in L^{p_i}_{w_i}}} w_0|f_0|^{p_0}(v)+tw_1|f_1|^{p_1}(v)=tw_1(v)G_v(x_0(v))
=w_0|x_0(v)|^{p_0}+tw_1(v)|f(v)-x_0(v)|^{p_1},
\end{eqnarray*}
namely \eqref{choosea} with $\alpha$ replaced by $x_0$.
Thus, as $x_0$ is measurable,  we may choose $\alpha(v)=x_0(v)$ and so the proof will be complete once we prove the existence and measurability of the minimising function $x_0$.

Set $a=a(v)=f(v)$ and $b=b(v)=\frac{w_0(v)}{tw_1(v)}>0$ so that 
\[G_v(x)=b|x|^{p_0}+|a-x|^{p_1}=b|x|+(a-x)^2.\]
We have that
\[\inf_{x\in\mathbb{R}}G_v(x)=\min\{\inf_{x\geq 0}G_v(x),\inf_{x\leq0}G_v(x)\}\]
and 
\begin{equation}\label{plusminus}G_v(x)=\left\{\begin{array}{ll}G_v^+(x) & \mbox{if } x\geq 0 \\ G_v^-(x)&\mbox{if }x\leq0,\end{array}\right.
\end{equation}
where
\[G_v^+(x)=x^2+(b-2a)x+a^2 \hspace{5mm}\mbox{and}\hspace{5mm}G_v^-(x)=x^2-(2a+b)x+a^2.\]
Hence
\begin{equation}\label{infeqn}
\inf_{x\in\mathbb{R}}G_v(x)=\min\{\inf_{x\geq 0}G_v^+(x),\inf_{x\leq0}G_v^-(x)\},
\end{equation}
and we need to determine the minimisers of $G_v^+$ and $G_v^-$ in the indicated domains.

As $G_v^{\pm}$  are monic quadratic polynomials, it follows that (over the whole real line) $G_v^{\pm}$ are minimised at the unque points where $(G_v^{\pm})'=0$. Now 
\[(G_v^+)'(x)=2x+(b-2a)\hspace{5mm} \mbox{ and }\hspace{5mm} (G_v^-)'(x)=2x-(2a+b),\]
so
\[\inf_{x\in\mathbb{R}}G_v^+(x)=G_v^+(\frac{2a-b}{2})  \hspace{5mm} \mbox{ and }\hspace{5mm}\inf_{x\in\mathbb{R}}G_v^-(x)=G_v^-(\frac{2a+b}{2}).\]
Hence (as $G_v^\pm$ are monotonic on either side of their minimiser)
\begin{equation}\label{options} \inf_{x\geq0}G_v^+(x)=\left\{\begin{array}{ll}G_v^+(\frac{2a-b}{2}) &\mbox{if }2a-b\geq0\\G_v^+(0)&\mbox{if } 2a-b<0\end{array}\right., \hspace{5mm}  \inf_{x\leq0}G_v^-(x)=\left\{\begin{array}{ll}G_v^-(\frac{2a+b}{2}) &\mbox{if }2a+b\leq0\\G_v^-(0)&\mbox{if } 2a+b>0\end{array}\right.. 
\end{equation}

We are now ready to determine the minimiser of $G_v$. Recalling that $b>0$, we distinguish three possible cases.
\begin{description}
\item[Case 1:] Assume $2a\geq b$. Then since $2a-b\geq0$ and $b>0$, we have $2a+b=(2a-b)+2b>0$. So by \eqref{options} we have
\[\inf_{x\geq0}G_v^+(x)=G_v^+(\frac{2a-b}{2})=ab-\frac{b^2}{4} \hspace{5mm}\mbox{ and }\hspace{5mm}  \inf_{x\leq0}G_v^-(x) = G_v^-(0)=a^2.   \]
In particular, $\inf_{x\leq0}G_v^-(x)-\inf_{x\geq0}G_v^+(x)=a^2-ab+\frac{b^2}{4}=(a-\frac{b}{2})^2\geq0$, so by \eqref{infeqn} and \eqref{plusminus}
\[\inf_{x\in\mathbb{R}}G_v(x)=\inf_{x\geq0}G_v^+(x)=G_v^+(\frac{2a-b}{2})=G_v(\frac{2a-b}{2})\]
and so  the infimum is attained at $x=x_0(v)=\frac{2a-b}{2}$.

\item[Case 2:] Assume $-b<2a<b$. Then $2a-b<0$ and $2a+b>0$. So using \eqref{options}, we have
\[\inf_{x\geq0}G_v^+(x)=G_v^+(0)=a^2 \hspace{5mm} \mbox{ and }\hspace{5mm} \inf_{x\leq0}G_v^-(x)=G_v^-(0)=a^2.\]
So by \eqref{infeqn} and \eqref{plusminus},
\[\inf_{x\in\mathbb{R}}G_v(x)=\inf_{x\geq0}G_v^+(x)=G_v^+(0)=G_v(0),\]
and so the infimum is attained at $x=x_0(v)=0$.

\item[Case 3:] Assume $2a\leq-b$. Then $2a+b\leq 0$ and (as $2a\leq -b<0$) $2a-b<0$. Thus by \eqref{options},
\[\inf_{x\geq0}G_v^+(x)=G_v^+(0)=a^2 \hspace{5mm} \mbox{ and }\hspace{5mm} \inf_{x\leq0}G_v^-(x)=G_v^-(\frac{2a+b}{2})=-(ab+\frac{b^2}{4}).\]
In particular, $\inf_{x\geq0}G_v^+(x)-\inf_{x\leq0}G_v^-(x)=a^2+ab+\frac{b^2}{4}=(a+\frac{b}{2})^2\geq0$, so by \eqref{infeqn} and \eqref{plusminus}
\[\inf_{x\in\mathbb{R}}G_v(x)=\inf_{x\leq0}G_v^-(x)=G_v^-(\frac{2a+b}{2})=G_v(\frac{2a+b}{2}),\]
and so the infimum is attained at $x=x_0(v)=\frac{2a+b}{2}$.
\end{description}

Combining the three cases above (and recalling that $a=f(v)$, $b=\frac{w_0}{tw_1}(v)$) yields that $\inf_{x\in\mathbb{R}}G_v(x)=G_v(x_0(v))$, where
\[x_0(v)=\left\{\begin{array}{ccc} \frac{2a-b}{2}, & \mbox{if} & 2a\geq b\\0,&\mbox{if}& -b<2a<b\\\frac{2a+b}{2},&\mbox{if}& 2a\leq -b\end{array}\right.  =\left\{\begin{array}{ccc} \frac{1}{2}[2f(v)-\frac{w_0}{tw_1}(v)], &\mbox{if}& \frac{1}{2}[2f(v)-\frac{w_0}{tw_1}(v)]\geq0\\ 0,&\mbox{if} & -\frac{w_0}{tw_1}(v)< 2f(v) < \frac{w_0}{tw_1}(v)   \\ \frac{1}{2}[2f(v)+\frac{w_0}{tw_1}(v)], &\mbox{if}&\frac{1}{2}[2f(v)+\frac{w_0}{tw_1}(v)]\leq0 \end{array}\right..  \]
We have thus shown the existence of a minimiser $x_0(v)$ for $G_v$ for each $v\geq1$. It remains only to show that $x_0$ is a measurable function of $v$.

Define $g_i:[1,\infty)\to\mathbb{R}$, $i=1,\ldots,4$ by
\[g_1=\frac{1}{2}[2f+\frac{w_0}{tw_1}],\hspace{6mm} g_2=  -\frac{w_0}{tw_1}-2f,\hspace{6mm} g_3=2f-\frac{w_0}{tw_1},\hspace{6mm} g_4=\frac{1}{2}[2f-\frac{w_0}{tw_1}].\]
By virtue of the fact that $f, w_0$ and $w_1$ are measurable, $g_i, i=1,\ldots,4$ are measurable functions. Given $y\in\mathbb{R}$, we show $x_0^{-1}((-\infty,y])$ is measurable.
\begin{itemize}
\item
Suppose $y<0$. Then as $x_0(v)<0$ if and only if $2f+\frac{w_0}{tw_1}<0$ (in which case $x_0=g_1$),
\[x_0^{-1}((-\infty,y])= g_1^{-1}((-\infty,y]),\]
which is measurable because $g_1$ is measurable.
\item
Suppose $y=0$. Then 
\begin{equation}\label{cunion}
x_0^{-1}((-\infty,y])=x_0^{-1}((-\infty,0])=\left(\bigcup_{n=1}^\infty x_0^{-1}((-\infty,-1/n])\right)\cup x_0^{-1}(\{0\}).
\end{equation}
Each of the sets $x_0^{-1}((-\infty,-1/n])$ is measurable by the previous case ($y<0$). Moreover,
\begin{eqnarray*}
x_0^{-1}(\{0\})&=&\{v\in[1,\infty):-\frac{w_0}{tw_1}(v)\leq 2f(v) \leq \frac{w_0}{tw_1}(v)\} \\
&=&\{v\in[1,\infty):g_2(v)<0\mbox{ and }g_3(v)<0\}\cup\{v\in[1,\infty):g_1(v)=0 \mbox{ or }g_4(v)=0\}\\
&=&g_2^{-1}((-\infty,0)) \cap g_3^{-1}((-\infty,0))\cup g_1^{-1}(\{0\}) \cup g_4^{-1}(\{0\}),
\end{eqnarray*}
 so $x_0^{-1}(\{0\})$ is measurable as each $g_i$ is measurable. It then follows from \eqref{cunion} that $x_0^{-1}((-\infty,y])$ is a countable union of measurable sets, and hence is measurable.
\item
Suppose $y>0$. Then
\[x_0^{-1}((-\infty,y])  =x_0^{-1}((-\infty,0])\cup x_0^{-1}((0,y]). \]
By the previous case ($y=0$), $x_0^{-1}((-\infty,0])$ is measurable. On the other hand, as $x_0(v)>0$ if and only if $2f-\frac{w_0}{tw_1}>0$ (in which case $x_0=g_4$)
\[x_0^{-1}((0,y])=g_4^{-1}((0,y])\]
so $x_0^{-1}((0,y])$ is measurable (because $g_4$ is measurable).  Thus $x_0^{-1}((-\infty,y]) $ is the union of two measurable sets and hence is measurable.
\end{itemize}
Thus $x_0^{-1}((-\infty,y]) $ is measurable for each $y\in\mathbb{R}$, and thus $x_0:[1,\infty)\to \mathbb{R}$ is measurable, as desired. This completes the proof.
\end{proof}

\end{document}